%% file: iid_prophet_limited_flex.tex
\documentclass[11pt]{article}

\usepackage[margin=1in]{geometry}

\usepackage{enumitem}
\usepackage[utf8]{inputenc}
\usepackage{amsmath,amssymb,tabu,amsthm,mathtools}
\usepackage{color}
\usepackage{calc}
\usepackage{tikz}
\usetikzlibrary{positioning,arrows,decorations.pathreplacing,shapes}
\usetikzlibrary{calc}
\usepackage{multirow}
\usepackage{boxedminipage}
\usepackage{xifthen}
\usepackage{tabularx}

\usepackage{mathpazo}

\usepackage{algorithmic}
\usepackage[linesnumbered,vlined,ruled]{algorithm2e}
\usepackage{hyperref}
\newcommand{\footremember}[2]{%
    \footnote{#2}
    \newcounter{#1}
    \setcounter{#1}{\value{footnote}}%
}

\newcommand{\R}{\mathbf{R}}

\newcommand{\Z}{\mathbf{Z}}

\DeclareMathOperator{\E}{\mathbf{E}}
\DeclareMathOperator{\Prob}{\mathbf{P}}

\newtheorem{claim}{Claim}%[section]
\newtheorem*{claim*}{Claim}
\newtheorem{theorem}{Theorem}
\newtheorem{corollary}[theorem]{Corollary}
\newtheorem{lemma}[theorem]{Lemma}
\newtheorem{proposition}[theorem]{Proposition}
\newtheorem{remark}{Remark}

\newcommand{\modif}[1]{\textcolor{black}{#1}}

\title{The IID Prophet Inequality with Limited Flexibility}

\author{
	Sebastian Perez-Salazar\footremember{1}{Rice University, \texttt{sperez@rice.edu}}
	\and Mohit Singh\footremember{2}{Georgia Institute of Technology, \texttt{mohit.singh@isye.gatech.edu}}
	\and Alejandro Toriello\footremember{3}{Georgia Institute of Technology, \texttt{atoriello@isye.gatech.edu}}
}

\begin{document}

\maketitle

\begin{abstract}
In online sales, sellers usually offer each potential buyer a posted price in a take-it-or-leave fashion. Buyers can sometimes see posted prices faced by other buyers, and changing the price frequently could be considered unfair. The literature on posted price mechanisms and prophet inequality problems has studied the two extremes of pricing policies, the fixed price policy and fully dynamic pricing. The former is suboptimal in revenue but is perceived as fairer than the latter. This work examines the middle situation, where there are at most $k$ distinct prices over the selling horizon. Using the framework of prophet inequalities with independent and identically distributed random variables, we propose a new prophet inequality for strategies that use at most $k$ thresholds. We present asymptotic results in $k$ and results for small values of $k$. For $k=2$ prices, we show an improvement of at least $11\%$ over the best fixed-price solution. Moreover, $k=5$ prices suffice to guarantee almost $99\%$ of the approximation factor obtained by a fully dynamic policy that uses an arbitrary number of prices. From a technical standpoint, we use an infinite-dimensional linear program in our analysis; this formulation could be of independent interest to other online selection problems.
\end{abstract}

\input{Introduction.tex}

\input{Main.tex}
\input{conclusions.tex}

%\newpage

{\bibliographystyle{amsplain}\fontsize{10}{0}\selectfont\bibliography{biblio}
}

\appendix
\input{Appendix.tex}

\end{document}

%% file: Introduction.tex
\section{Introduction}

Pricing is one of the elements of a business operation with the highest impact on profitability
%considered a strategic area that businesses can tackle to get the highest impact on their profitability~
\cite{phillips2012prices}. A recent survey by McKinsey~\cite{McKinsey} shows that a $1\%$ improvement in pricing can yield a $6\%$ increase in profits, while in contrast a $1\%$ reduction in variable costs can produce up to a $3.8\%$ increase in profits. Strategic and dynamic pricing is as old as business and widely studied. Nevertheless, the deregulation of airline prices in the US in 1978 and the development of revenue management generated significant interest in pricing from the research community~\cite{talluri2004theory}. Despite the success of dynamic pricing from a profitability viewpoint, in many businesses changing prices too often could be considered unfair from the customers' standpoint, particularly when customers can learn the prices faced by others \cite{garbarino2003dynamic,pk2001dynamic,rotemberg2011fair}. Therefore, fixed-price policies are often preferred over dynamic pricing, particularly in retail transactions~\cite{phillips2012prices}. Although a fixed price is considered fairer than its dynamic counterpart, it can lead to suboptimal revenues; this generates a natural tension between revenue and customer goodwill. Particularly for products that will perish over time, such as food, airplane seats~\cite{adelman2007dynamic} and sponsored ads~\cite{agarwal2009spatio}, it is natural to consider price variation. In this work we explore a middle ground, dynamic pricing with limited prices, using the context of Bayesian online selection and, more precisely, the prophet inequality problem.

We use the prophet inequality problem as an underlying model for online selection. The prophet inequality problem is one of the classical problems in \emph{stopping theory}, and has recently gained increasing attention for its deep connections with pricing problems~\cite{chawla2010multi,correa2019pricing,hajiaghayi2007automated}. In the classical formulation of the prophet inequality problem, there are $n$ nonnegative independent random variables (r.v.) $X_1,\ldots,X_n$, and a decision-maker (DM) observes their values sequentially. Upon observing a value, the DM must immediately decide whether to accept or reject the value. Accepting stops the process, and the DM gets the observed value as a reward, while rejecting the value is irrevocable and allows the DM to observe the next random variable's value. The DM's goal is to maximize their expected value. To benchmark the DM's strategy, the value obtained by the DM is compared against the maximum value in hindsight, $\E[\max_i X_i]$ -- the so-called \emph{prophet value} -- and the goal is to obtain a lower bound on the ratio between the DM's value and the prophet value. It is known that a simple \emph{threshold} strategy that computes a value $x$ using the distributions and accepts the first observed value that surpasses $x$ attains a ratio of $1/2$, and this result is tight; see e.g.~\cite{krengel1977semiamarts,samuel1984comparison}. This simple threshold solution is appealing for pricing mechanism design in online sales, as we can interpret the threshold as a price~\cite{chawla2010multi}.
A posted-price mechanism (PPM) is a method to sell items where a seller offers a menu of prices to each buyer in a take-it-or-leave-it fashion. \modif{The solution to the classical prophet inequality problem exhibits this posted pricing form: $X_i$ represents the valuation of the $i$-th buyer, the prophet's value $\E[\max_i X_i]$ is the social welfare and $x$ denotes the posted price}. Recent works have established the equivalence between PPMs and prophet inequalities in several settings~\cite{chawla2010multi,correa2019pricing,hajiaghayi2007automated}.

In massive markets, assuming identical valuations is reasonable when granular information about buyers' valuations is hard to come by. In this case, the prophet inequality problem corresponds to the observation of independent and identically distributed (i.i.d.) random variables $X_1,\ldots,X_n$. Hill \& Kertz~\cite{hill1982comparisons} were the first to show that a prophet inequality with an approximation ratio better than $1/2$ can be obtained in the i.i.d.\ setting, with an impossibility result showing that no ratio better than $\bar{\gamma}\approx 0.745$ can be attained, where $\bar{\beta}=1/\bar{\gamma}$ is the (unique) solution of the equation $\int_0^1 (\beta -1 + y(\log y -1))^{-1}\mathrm{d}y=1$. Recently, Correa et al.~\cite{correa2021posted} showed an algorithm that attains a prophet inequality of $\bar{\gamma}$ for any number of random variables $n$.

As opposed to the solution of the prophet inequality problem for general independent random variables, which we can solve by computing a single threshold (for instance, the median of $\max \{X_i: {i\in [n]}\}$), the optimal solution for the i.i.d.\ prophet inequality problem consists of $n$ decreasing thresholds, computed via dynamic programming~\cite{hill1982comparisons} or by sampling $n$ thresholds from a tailored distribution \cite{correa2021posted}. These solutions are fully dynamic, as they change the thresholds of acceptance as time progresses. In the pricing language, this is equivalent to every buyer observing a new posted price. On the other hand, the best fixed threshold policy for the i.i.d.\ prophet inequality guarantees a fraction $(1-1/e)$ of the prophet value~\cite{correa2021posted}. In this work, we fill the gap between these two extremes. Using the lens of optimization, we provide near-optimal solutions for the problem of selecting at most $k$ thresholds over a time horizon of length $n$.

\subsection{Problem Formulation and Summary of Contributions}

The inputs of the problem are (1) a distribution $\mathcal{D}$ with nonnegative support, (2) $n$ i.i.d.\ random variables $X_1,\ldots,X_n\sim \mathcal{D}$ drawn from this distribution, and (3) a fixed integer $k\geq 1$. A decision-maker (DM) observes the realizations of $X_1,\ldots,X_n$ sequentially and needs to decide whether to accept the observed value and stop the process, or irrevocably move on to the next value; the goal is to maximize the expected accepted value. If the DM implements an algorithm or policy $A$, we denote the value collected by $A$ as $\E[X_A]$. The algorithm $A$ is said to attain a \emph{$\gamma$-approximation} (of the prophet value) if $\E[X_A] \geq \gamma \E[\max_{i\in [n]} X_i]$ for any input distribution $\mathcal{D}$. The prophet inequality literature has focused on finding an algorithm $A^*$ that maximizes $\gamma$.

In this work, we are interested in \emph{$k$-dynamic} algorithms (or policies), which choose at most $k$ time intervals $[1,i_1], [i_1+1,i_2],\ldots,[i_{k-1}+1,n]$ and compute a static threshold in each window $\bar{x}_1,\ldots,\bar{x}_k$ such that the first value $X_i$ observed in $[i_{t}+1,i_{t+1}]$ that exceeds $\bar{x}_t$ is accepted; see Figure~\ref{fig:windows}.

\begin{figure}
	\centering
	\scriptsize
	\begin{tikzpicture}[xscale=0.7, yscale=1.7]
		\draw[->] (-0.5, 0) -- (16.5, 0) node[right] {Time horizon};
		\draw[->] (0, -0.1) -- (0, 3) node[above,text width=2cm,text centered] {Threshold};
		
		\draw[-,thick] (3,{+0.02}) -- (3,{-0.05}) node[below] {$i_1$};
		\draw[-,thick] (5,{+0.02}) -- (5,{-0.05}) node[below] {$i_2$};
		\draw[-,thick] (8,{+0.02}) -- (8,{-0.05}) node[below] {$i_3$};
		%	\draw[-,thick] (4,{0.6+0.002}) -- (4,{0.6-0.005}) node[below] {$4$};
		\draw[-,thick] (14,{+0.02}) -- (14,{-0.05}) node[below] {$i_{k-1}$};
		\draw[-,thick] (16,{+0.02}) -- (16,{-0.05}) node[below] {$n$};
		
		\draw (11,-0.05) node[below] {$\cdots$};
		
		\draw (0,2.5) node[left] {$\bar{x}_1$};
		\draw (0,1.5) node[left] {$\bar{x}_2$};
		\draw (0,1) node[left] {$\bar{x}_3$};
		\draw (0,0.2) node[left] {$\bar{x}_k$};
		
		%thresholds
		\draw[-,very thick] (0.2,2.5)-- (3,2.5);
		\draw[-,very thick] (3.2,1.5)-- (5,1.5);
		\draw[-,very thick] (5.2,1)-- (8,1);
		\draw[-,very thick] (14.2, 0.2)-- (16,0.2);
		
		\draw[-,thin] (3,0) -- (3,2.5);
		\draw[-,thin] (5,0) -- (5,1.5);
		\draw[-,thin] (8,0) -- (8,1);
		\draw[-,thin] (16,0) -- (16,0.2);
		
		\draw[-,dashed] (0,1.5) -- (3,1.5);
		\draw[-,dashed] (0,1) -- (5,1);
		\draw[-,dashed] (0,0.2) -- (14,0.2);
		
		%sampled values
		\draw[-] (1,0) -- (1,1.6) node[circle,fill=black,scale=0.4] {};
		\draw[-] (2,0) -- (2,0.9) node[circle,fill=black,scale=0.4] {};
		\draw[-] (3,0) -- (3,0.7) node[circle,fill=black,scale=0.4] {};
		
		\draw[-] (4,0) -- (4,1.3) node[circle,fill=black,scale=0.4] {};
		\draw[-] (5,0) -- (5,1.1) node[circle,fill=black,scale=0.4] {};
		
		\draw[-] (6,0) -- (6,0.3) node[circle,fill=black,scale=0.4] {};
		\draw[-] (7,0) -- (7,1.4) node[circle,fill=black,scale=0.4] {};
		
		% accepted value
		\draw[stealth-,] (7.4,1.4) -- (9,1.4) node[right,text width=1cm,text centered] {Accepted value};

		%		\node(a) at (1,0.6079) [circle,fill=black,scale=0.4] {};
		%		\node(b) at (2,0.7045) [circle,fill=black,scale=0.4] {};
		%		\node(c) at (3,0.723) [circle,fill=black,scale=0.4] {};
		%		\node(d) at (4,0.731) [circle,fill=black,scale=0.4] {};
		%		\node(e) at (5,0.737) [circle,fill=black,scale=0.4] {};

	\end{tikzpicture}
	\caption{Illustration of a $k$-dynamic policy, with thresholds plotted over the time horizon. For each interval $I_t=[i_{t-1}+1,i_t]$, we have an associated threshold $\bar{x}_t$. The first outcome of the random variables with index in $I_t$ that surpasses $\bar{x}_t$ is accepted. In the picture, this occurs at the second value observed in the third interval $I_3$.}\label{fig:windows}
\end{figure}
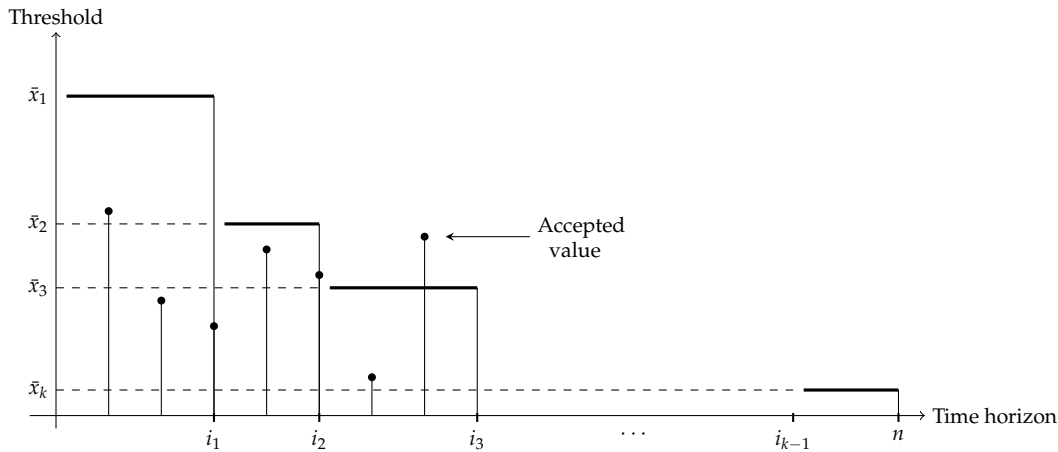

Formally, we are interested in solving
\[
\gamma_{n,k}^*= \sup_{\substack{A \\\text{ $k$-dynamic}}} \inf_{\substack{\mathcal{D} \text{ distr.}\\X_1,\ldots,X_n \sim \mathcal{D}}}  \frac{\E[X_A]}{\E[\max_{i\in [n]} X_i]}.
\]
For $k=1$, we obtain the \emph{static solution} that guarantees $\gamma_{n,1}^* \approx 1-1/e$ for $n$ large. For $k=n$ we obtain the \emph{fully dynamic} solution that guarantees $\gamma_{n,n}^* = \bar{\gamma} \approx 0.745$. Our goal is to understand $\gamma_{n,k}^*$ for intermediate values of $k$ and provide new prophet inequalities for the i.i.d.\ problem under $k$-dynamic algorithms.

\paragraph{Summary of Contributions} We present an extensive study of $k$-dynamic algorithms and new prophet inequalities. \modif{We propose an infinite-dimensional linear program that encodes optimal $k$-dynamic algorithms and their approximation ratio $\gamma_{n,k}^*$. We use this formulation to analytically compute $\gamma_{n,1}^*=1-(1-1/n)^n\approx 1-1/e$, and to numerically calculate $\gamma_{n,2}^* \approx 0.708 $. In other words, allowing one additional threshold over the time horizon improves the approximation ratio by $11\%$. For larger values of $k$, analyzing this infinite-dimensional formulation becomes increasingly difficult; hence, we employ a relaxation.} Moreover, we restrict our attention to solutions where the sizes of the intervals $[1,i_1], [i_1+1,i_2],\ldots,[i_{k-1}+1,n]$ are fixed up front. \modif{The relaxation of our infinite linear program provides a universal lower bound for $\gamma_{n,k}^*$;} the dual of this new formulation is oblivious to the input distribution and its solution yields distributions for each interval, from which we can reconstruct $k$-dynamic algorithms as follows. For each distribution obtained from the program, we compute a quantile value $q\in [0,1]$, from which we obtain a threshold $x_t$ via the equation $q=1- F(x_t)$, where $F$ is the CDF of the input distribution $\mathcal{D}$. Furthermore, we fully characterize the optimal primal and dual solutions as functions of $n$; hence, we can avoid solving a complex formulation. Since the optimal solution is difficult to analyze for finite $n$, we study its asymptotic behavior when $n\to \infty$ and $k$ is fixed. We show that the asymptotic behavior of the objective of the infinite linear program is $\bar{\gamma} (1 - C \log k / k)$, for $k$ large but constant. This shows that as $k$ grows, $k$-dynamic algorithms obtain the optimal ratio of a fully dynamic solution. We also provide numerical results for $k \leq 10$, which empirically show that $\gamma_{n,k}^*$ is already close to optimal for $k \geq 5$. In Figure~\ref{fig:approximations}, we present some of the approximations computed for small $k$ using our methodologies.

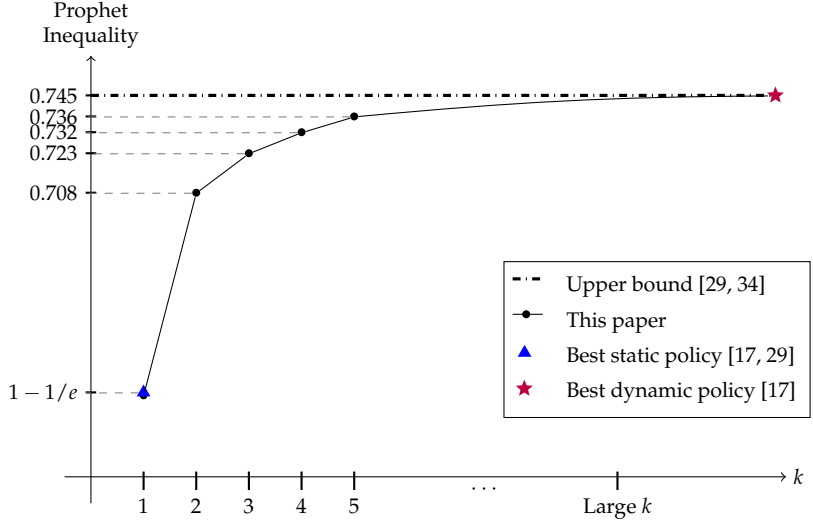
\begin{figure}
	\centering
	\scriptsize
	\begin{tikzpicture}[xscale=0.7, yscale=35]
		\draw[->] (-0.5, 0.6) -- (13.2, 0.6) node[right] {$k$};
		\draw[->] (0, 0.59) -- (0, 0.76) node[above,text width=2cm,text centered] {Prophet Inequality};
		
		\draw[-,thick] (1,{0.6+0.002}) -- (1,{0.6-0.005}) node[below] {$1$};
		\draw[-,thick] (2,{0.6+0.002}) -- (2,{0.6-0.005}) node[below] {$2$};
		\draw[-,thick] (3,{0.6+0.002}) -- (3,{0.6-0.005}) node[below] {$3$};
		\draw[-,thick] (4,{0.6+0.002}) -- (4,{0.6-0.005}) node[below] {$4$};
		\draw[-,thick] (5,{0.6+0.002}) -- (5,{0.6-0.005}) node[below] {$5$};
		
		\draw (7.5,0.6) node[below] {$\cdots$};
		
		\draw[-,thick] (10,{0.6+0.002}) -- (10,{0.6-0.005}) node[below] {Large $k$};
		
		%	\draw[-,thick] (0.1,{0.6079}) -- (-0.1,{0.6079}) node(x) [left] {$6/\pi^2$};
		\draw[-,thick] (0.1,{0.6321}) -- (-0.1,{0.6321}) node(y) [left] {$1-1/e$};
		\draw[-,thick] (0.1,{0.708}) -- (-0.1,{0.708}) node(z) [left] {$0.708$};
		\draw[-,thick] (0.1,{0.723}) -- (-0.1,{0.723}) node(u) [left] {$0.723$};
		\draw[-,thick] (0.1,{0.731}) -- (-0.1,{0.731}) node(v) [left] {$0.732$};
		\draw[-,thick] (0.1,{0.737}) -- (-0.1,{0.737}) node(w) [left] {$0.736$};
		
		\draw[-,thick] (0.1,{0.745}) -- (-0.1,{0.745}) node[left] {$0.745$};
		
		\draw[-,dash dot,very thick] (0,0.745)-- (13,0.745);
		
		\node(a) at (1,0.631) [circle,fill=black,scale=0.4] {};
		\node(b) at (2,0.708) [circle,fill=black,scale=0.4] {};
		\node(c) at (3,0.723) [circle,fill=black,scale=0.4] {};
		\node(d) at (4,0.731) [circle,fill=black,scale=0.4] {};
		\node(e) at (5,0.737) [circle,fill=black,scale=0.4] {};
		
		\node(correa) at (1,0.6321) [regular polygon,regular polygon sides=3,fill=blue,scale=0.4] {};
		
		\draw[-,dashed,gray] (correa) -- (y);

		\node(f) at (13,0.745) {};
		
		\draw[-] (a) -- (b) -- (c) -- (d) -- (e) ;
		
		%	\draw[-,dashed,gray] (a) -- (x);
		\draw[-,dashed,gray] (b) -- (z);
		\draw[-,dashed,gray] (c) -- (u);
		\draw[-,dashed,gray] (d) -- (v);
		\draw[-,dashed,gray] (e) -- (w);
		
		\node(correa2) at (13,0.745) [star, fill=purple, scale=0.4, star points=5,star point ratio=2.25] {};
		
		\draw[-]  (e) to [bend left=0.05] (f) ;
		
		\matrix [draw,below left] at (current bounding box.east) {
			\draw[-,dash dot,very thick] (-0.2,0)-- (0.2,0); &  \node [label=right:Upper bound~\cite{hill1982comparisons,kertz1986stop}] {}; \\
			\draw (-0.2,0) -- (0.2,0) ; \node at (0.02,0.04) [circle,fill=black,scale=0.4] {};& \draw node [label=right:This paper] {}; \\
			\node [regular polygon,regular polygon sides=3,fill=blue,scale=0.4] {}; &\node [label=right:Best static policy~\cite{correa2021posted,hill1982comparisons}] {}; \\
			\node [star, fill=purple, scale=0.4, star points=5,star point ratio=2.25] {}; &\node [label=right:Best dynamic policy~\cite{correa2021posted}] {}; \\
		};
		
	\end{tikzpicture}
	\caption{\modif{Asymptotic approximations obtained by our method. The plot represents the values computed via $\gamma_{n,k}^*$ for $k=1,2$ in Theorems~\ref{thm:optimal_cr} and~\ref{prop:dual_P_nk} and $v_{\infty,k}^*$ for $k\in \{3,4,5\}$ in Theorem~\ref{thm:main1}} below.}\label{fig:approximations}
\end{figure}

\subsection{Technical Results}

\modif{Our first technical result characterizes the optimal approximation factor for $k$-dynamic algorithms. We utilize a quantile-based approach similar to the one proposed in~\cite{correa2021posted}. The idea is to make decisions based on the quantile $q$ given by the value $x$ such that $\Prob(X\geq x) = q$, instead of directly computing the optimal policy in the $x$ values.
\begin{theorem}[Optimal Approximation Factor]\label{thm:optimal_cr}
	The optimal approximation ratio for $k$-dynamic algorithms is
	\[
	\gamma_{n,k}^* = \max_{\substack{\tau_1,\ldots,\tau_k\in \Z_+\\ \tau_1+\cdots+\tau_k=n}} \gamma_{n,k}^*(\tau_1,\ldots,\tau_k) ,
	\]
	where $\gamma_{n,k}^*(\tau_1,\ldots,\tau_k)$ is the optimal value of
	\begin{subequations}\small\label{form:}
		\begin{align}
			\inf_{\substack{\mathbf{d}\in \R_+^k \\ f:[0,1]\to \R_+}} \quad & d_1 \notag\\
			(D)_{n,k}\qquad\qquad \text{s.t.} \quad & d_t \geq \left(\frac{1-(1-q)^{\tau_t}}{q}\right) \int_0^q f_u \, \mathrm{d}u + (1-q)^{\tau_t} d_{t+1}, & \forall t\in[k],\forall q\in [0,1] \label{const:dynamic_constr} \\
			\quad &  \int_0^1 n (1-u)^{n-1} f_u \, \mathrm{d}u = 1,  & \label{const:max_value_const} \\
			\quad & f_u \geq f_{u'},& \forall u\leq u' . \label{const:nonincreasing}
		\end{align}
	\end{subequations}
\end{theorem}}
\modif{The theorem follows directly once we show that $\gamma_{n,k}^*(\tau_1,\ldots,\tau_k)$ is the best approximation ratio that a $k$-dynamic algorithm can attain using fixed windows of size $\tau_1,\ldots,\tau_k$. To show this, first, we note that for $f:[0,1]\to \R_+$ satisfying Constraints~\eqref{const:max_value_const} and~\eqref{const:nonincreasing}, we can construct a nonnegative random variable $X$ with CDF $F$ given by $F^{-1}(1-u)=f_u$. A calculation shows that $\E[\max_{i\in [n]} X_i]=\int_0^1 n (1-u)^{n-1} f_u \, \mathrm{d}u=1$. From here, we prove that for any feasible solution $(\mathbf{d},f)$ to $(D)_{n,k}$ and any $k$-dynamic algorithm $A$ using windows of size $\tau_1,\ldots,\tau_k$, we get $\E[X_A]\leq d_1$, implying that the approximation ratio of $A$ is at most $d_1$. Moreover, from the optimal $k$-dynamic algorithm $A$ for windows of size $\tau_1,\ldots,\tau_k$, we can generate a feasible solution $(\mathbf{d},f)$ to $(D)_{n,k}$ such that $d_1\leq \E[X_A]$. From here, the proof follows. We present the details in Appendix~\ref{app:proof_thm_1}.}

%Indeed, from a feasible solution to $(D)_{n,k}$, say $(\mathbf{d},f)$, we construct a random variable $X$ with distribution $\mathcal{D}$ and CDF $F$ given by $F^{-1}(1-u)=f_u$. The prophet value will be $\E[\max_{i\in [n]} X_i]=\int_0^1 n (1-u)^{n-1} f_u \, \mathrm{d}u=1$. A quick calculation shows that $\E[X_A]\leq d_1$ for any $k$-dynamic algorithm $A$ that uses windows of size $\tau_1,\ldots,\tau_k$, implying that the approximation ratio of $A$ is at most $d_1$. Moreover, the optimal $k$-dynamic algorithm $A$ induces a feasible solution of $(D)_{n,k}$ with $d_1\leq \E[X_A]$ for any $f$ that satisfies~\eqref{const:max_value_const} and~\eqref{const:nonincreasing}. From here, the first part of the proof follows. By optimizing over the windows sizes, we get that $\gamma_{n,k}^*$ is the worst-case approximation ratio. We present the details in Appendix~\ref{app:proof_thm_1}.

\modif{Our next result shows that the optimal value of $(D)_{n,k}$, $\gamma_{n,k}^*(\tau_1,\ldots,\tau_k)$, equals the value of a maximization problem that we interpret as the dual of $(D)_{n,k}$. This technical result is crucial for the design and analysis of optimal algorithms, since from this dual problem we can extract optimal $k$-dynamic algorithms.}
\modif{\begin{theorem}[Optimal Policy]\label{prop:dual_P_nk}
		For any integers $\tau_1,\ldots,\tau_k\geq 0$ such that $\tau_1+\cdots + \tau_k=n$, the value $\gamma_{n,k}^*(\tau_1,\ldots,\tau_k)$ equals %the value of the maximization problem
		\begin{subequations}\label{form:P_nk}\small
			\begin{align}
				\sup_{\substack{  \alpha:[k]\times [0,1]\to \R_+ \\ \eta:[0,1]\to \R_+, \eta\in \mathcal{C}} }\quad & v \notag\\
				(P)_{n,k}\qquad s.t. \quad & \int_0^1 \alpha_{1,q}\, \mathrm{d}q \leq 1, & \label{const:constr_1_P_nk} \\
				\quad & \int_0^1 \alpha_{t+1,q} \, \mathrm{d}q\leq \int_0^1 (1-q)^{\tau_t} \alpha_{t,q} \, \mathrm{d}q,  & \forall t\in [k-1] ,\label{const:constr_t_P_nk}\\
				\quad &   v n (1-u)^{n-1} + \frac{\mathrm{d}\eta_u}{\mathrm{d}u} \leq \sum_{t=1}^k\int_u^1 \left( \frac{1-(1-q)^{\tau_t}}{q}   \right) \alpha_{t,q} \, \mathrm{d}q, & [0,1]-\text{a.e.}  , \label{const:P_nk_esp_max_value_const} \\
				\quad & \eta_0 = 0, \eta_1=0, \label{const:P_nk_initial_condition}
			\end{align}
		\end{subequations}
		where a.e.\ stands for ``almost everywhere with respect to the Lebesgue measure in $[0,1]$.''
\end{theorem}
Functions $\alpha=\{\alpha_{t}\}_t$ are dual variables associated with Constraints~\eqref{const:dynamic_constr}, variable $v$ is associated with Constraint~\eqref{const:max_value_const}, and function $\eta$ is the dual variable associated with Constraints~\eqref{const:nonincreasing}. As usual in proving duality, we split the proof into weak and strong duality. The infinite dimensions of the primal and dual make the proof more technical; we defer it to Appendix~\ref{app:strong_dual}.}

From a solution of $(P)_{n,k}$ we can obtain a $k$-dynamic algorithm as follows. Given a feasible solution $(\alpha=\{ \alpha_{t} \}_{t},v,\eta)$ of $(P)_{n,k}$, we sample $q_t$ according to the probability mass distribution $\alpha_{t,q}/\int_0^1 \alpha_{t,q} \, \mathrm{d}q$ and set the threshold $x_t$ in the $t$-th window to satisfy $q_t=\Prob(X \geq x_t)$. We can show that this algorithm guarantees an approximation of at least $v$ fraction of the value of the prophet $\E[\max_{i} X_i]$ for any continuous distribution. Notice that $(P)_{n,k}$ is \emph{independent} of the input distribution $\mathcal{D}$, yet it provides a prophet inequality. We present the formal algorithm and analysis in Section~\ref{sec:prophet_bounded_threshold}. \modif{This construction shows that the optimal value of $(P)_{n,k}$ is at most the best approximation factor for $k$-dynamic algorithms that use windows of size $\tau_1,\ldots,\tau_k$. Equality holds due to the strong duality in Theorem~\ref{prop:dual_P_nk}.} %ATH: Maybe move this discussion to where the actual proof is? %Ideally, we would like to give a self-contained proof of this equality using only $(P)_{n,k}$, e.g., using the optimal $k$-dynamic algorithm to get a solution to $(P)_{n,k}$ with objective value set to the optimal approximation. Unfortunately, this approach encounters some challenges stemming from two key issues: (1) The optimal $k$-dynamic algorithm might not be distribution agnostic, rendering the precise definition of $\alpha_{t,q}$ unclear. (2) Even if the optimal algorithm makes decision based on quantiles (as our algorithm does), the construction of a feasible $(\alpha,v,\eta)$ of $(P)_{n,k}$ remains unclear. Notably, $\eta$ must satisfy a differential inequality with an arbitrary $\alpha$ and $v$. Hence, the use of $(D)_{n,k}$ and strong duality seems unavoidable to prove the optimality of $(P)_{n,k}$.}

%The feasible region codifies all the policies that work by sampling values $q_1,\ldots,q_k$, while the objective is a lower bound of the ratio $\E[X_A]/\E[\max_{i\in [n]} X_i]$, where $A$ is the policy that uses thresholds $x_t$ computed via $q_t= \Prob(X\geq x_t)$.

\modif{Using the formulation in Theorem~\ref{thm:optimal_cr} (and Theorem~\ref{prop:dual_P_nk}), we compute $\gamma_{n,1}^*=1-(1-1/n)^n$, which recovers the known asymptotic bound of $1-1/e$ (see, e.g.,~\cite{correa2021posted}), and we also compute the new value $\gamma_{n,2}^*\approx 0.708$ by finding matching primal and dual solutions of problem $(D)_{n,2}$ and $(P)_{n,2}$, respectively, and numerically approximating the objective values; see Section~\ref{sec:small_thresholds}. However, for $k\geq 3$, characterizing these solutions via duality becomes a non-trivial task. Instead, we consider a relaxation of $(D)_{n,k}$ that drops constraints \eqref{const:nonincreasing} and fixes the sizes of the windows $\tau_1,\ldots,\tau_k$. This is the same as restricting $(P)_{n,k}$ to feasible solutions with $\eta=0$. In this case, we can explicitly characterize the optimal solution for any $k$ as a function of $n$. In the following theorem, we present the resulting formulation obtained by dropping constraints~\eqref{const:nonincreasing}.}

%As in~\cite{correa2021posted}, we focus on \modif{quantile}-based policies. The idea is to compute $k$ \modif{quantile values} $q_1,\ldots,q_k$ independently from the input distribution, and then compute $k$ thresholds $x_1,\ldots,x_k$, where $q_t=\Prob(X\geq x_t)$. The algorithm splits the time horizon into $k$ windows of time and in the $t$-th interval accepts the first value exceeding $ x_t$, if one is observed.
%
%The $k$ values $q_1,\ldots,q_k$ are randomly sampled from $k$ distributions. Our first result establishes a lower bound over $\gamma_{n,k}^*$ and a mechanical way to construct these $k$ distributions.
%
\begin{theorem}[Near-Optimal Policy for Fixed Time Windows] \label{thm:main1}
Let $\tau_1,\tau_2,\ldots,\tau_k$ be nonnegative integers such that $n=\tau_1+\tau_2+\cdots + \tau_k$. \modif{Then $\gamma_{n,k}^* \geq v_{n,k}^*$, where $v_{n,k}^*$ is the optimal value of}
%\begin{subequations}
%	\begin{align*}
%		\sup_{\alpha:[k]\times [0,1]\to \R_+}\inf_{u\in [0,1]} \quad & \frac{\sum_{t=1}^k\int_u^1 q^{-1}\left({1-(1-q)^{\tau_t}}  \right) \alpha_{t,q} \, \mathrm{d}q}{ n(1-u)^{n-1}}\\
%		(CLP)_{n,k}\qquad\qquad\qquad s.t. \quad & \int_0^1 \alpha_{1,q} \, \mathrm{d}q \leq 1  &  \\
%		\quad &  \int_0^1 \alpha_{t+1,q} \, \mathrm{d}q \leq \int_0^1 (1-q)^{\tau_t} \alpha_{t,q} \, \mathrm{d}q  & \forall t\in [k-1] .
%	\end{align*}
%\end{subequations}
\begin{subequations}\small
	\begin{align*}
		\sup_{ \alpha:[k]\times [0,1]\to \R_+ }\quad & v  \notag\\
		(CLP)_{n,k} \qquad\qquad\qquad  s.t. \quad & \int_0^1 \alpha_{1,q}\, \mathrm{d}q \leq 1, &\\
		\quad & \int_0^1 \alpha_{t+1,q} \, \mathrm{d}q\leq \int_0^1 (1-q)^{\tau_t} \alpha_{t,q} \, \mathrm{d}q,  & \forall t\in [k-1]\\
		\quad &   v n (1-u)^{n-1} \leq \sum_{t=1}^k\int_u^1 \left( \frac{1-(1-q)^{\tau_t}}{q}   \right) \alpha_{t,q} \, \mathrm{d}q, & \forall u\in [0,1].
	\end{align*}
\end{subequations}

Furthermore, we fully characterize the optimal solutions of $(CLP)_{n,k}$ when $\tau_1=\cdots=\tau_{k-1}\geq \tau_k \geq 0$.
\end{theorem}
Notice that the feasible region of $(CLP)_{n,k}$ is contained in the feasible region of $(P)_{n,k}$. Hence, from any solution $(\alpha,v)$ of $(CLP)_{n,k}$, we can obtain algorithms using the framework used for solutions in Theorem~\ref{prop:dual_P_nk}. To characterize the optimal solution of $(CLP)_{n,k}$, we utilize duality again; however, this time we provide explicit primal and dual solutions with matching objective values. The optimal solution $\alpha=\{ \alpha_{t} \}_{t}$ exhibits disjoint consecutive supports in $[0,1]$ for different $t$. As a consequence, the sampled values $q_1,\ldots,q_k$ are increasing, and the resulting thresholds are decreasing. \modif{For the case $k=n$ and $\tau_1=\cdots=\tau_n=1$, the optimal solution of $(CLP)_{n,k}$ recovers the distributions from \cite{correa2021posted}, which attains the optimal approximation for the classic i.i.d. prophet inequality problem.}

We further the study of $v_{n,k}^*$, the optimal value of $(CLP)_{n,k}$ and approximation for the $k$-dynamic algorithms, in the particular case of $\tau_1=\cdots=\tau_{k-1} =\lceil n/k \rceil$, $\tau_k= n-(k-1)\lceil n/k \rceil$. Our next result gives an asymptotic lower bound on $v_{n,k}^*$ and hence, also an asymptotic lower bound over $\gamma_{n,k}^*$.
\begin{theorem}[Asymptotic Lower Bound]\label{thm:asymptotic_informal}
	There is a constant $c_1$ such that for $n$ large and $k \ll n$,
	$$\gamma_{n,k}^* \geq v_{n,k}^* \geq \bar{\gamma} \left( 1 - c_1  \left(  \frac{\log k}{k} + \frac{k^2}{n} \right)  \right).$$
	Here, $\bar{\beta}=1/\bar{\gamma}$ is the unique solution of $\int_0^1 ( \beta -1 + y (\log y -1) )^{-1} \, \mathrm{d}y = 1$.
\end{theorem}

%The theorem implies the same lower bound for $\gamma_{n,k}^*$.
Thus, if $n$ goes to infinity, the optimal prophet inequality for $k$-dynamic algorithms approaches $\bar{\gamma}$ at least as fast as $(1-c_1 \log k /k)$. To show Theorem~\ref{thm:asymptotic_informal}, we let the number of random variables $n$ grow with an appropriate scaling of the time horizon $[n]$, and we analyze a version of $(CLP)_{n,k}$ that is free from $n$. We call this model the \emph{infinite model}, while for finite $n$, we refer to the problem as the \emph{finite model}. We also characterize the optimal solutions of this infinite model, and show that the optimal value $v_{n,k}^*$ in the finite case converges to $v_{\infty,k}^*$, the optimal value of the infinite model. We also show that the solution in the infinite model approaches the points of an ordinary differential equation (ODE) that has appeared before in the literature; see e.g., \cite{correa2021posted,kertz1986stop}. This allows us to show that $v_{\infty,k}^* \geq \bar{\gamma}(1-c'' \log k /k )$, where $c''$ is a universal constant independent of $k$. By backtracking from the infinite model, we can recover solutions for the finite model with the guarantee exhibited in Theorem~\ref{thm:asymptotic_informal}.

From the infinite model, we can deduce a limit to our approach with intervals of the same length. This partially complements Theorem~\ref{thm:asymptotic_informal}, by showing an upper bound over $v_{\infty,k}^*$ that converges to $\bar{\gamma}$.
\begin{theorem}[Asymptotic Upper Bound]\label{thm:asymptotic_informal_upper}
	There is a constant $c_2$ such that for large enough $k$,
	\[
	\lim_{n\to \infty} v_{n,k}^* = v_{\infty,k}^* \leq \bar{\gamma} \left( 1 - c_2 \frac{\log k}{k}  \right).
	\]
\end{theorem}
%
%Even though asymptotic results are interesting for both the practice and the theory of design of algorithms,
%
\modif{We can use Theorem~\ref{thm:optimal_cr} to compute the optimal value for $k=1$ and numerically approximate the optimal value for $k=2$. For small values of $k\geq 3$, computing the optimal solution using Theorem~\ref{thm:optimal_cr} becomes difficult. Nevertheless, we can use the analysis for Theorem~\ref{thm:asymptotic_informal} to provide lower bounds for the approximation factor;} we can also obtain algorithms for the finite model from its infinite counterpart. Moreover, a simple greedy procedure allows us to compute optimal solutions of the infinite model, provided that we can compute the integral of $-\log y/(1-y^{1/k})$.

\modif{For any $k$, there are essentially two gaps between $v_{n,k}^*$ and $\gamma_{n,k}^*$. The first and most significant one comes from dropping constraint~\eqref{const:nonincreasing}. For $k=1$, for instance, $v_{n,1}^* \approx 6/\pi^2 \approx 0.607$ (see Corollary~\ref{cor:value_k_1}), while $\gamma_{n,1}^*\approx 1-1/e\approx 0.632$. The second gap comes from fixing the sizes of the windows. Our theoretical results show that these gaps disappear as $k$ grows (Theorems~\ref{thm:asymptotic_informal} and~\ref{thm:asymptotic_informal_upper}). Moreover, our experimental results show that for $k\geq 2$, dropping Constraints~\eqref{const:nonincreasing} does not significantly impact the approximation guarantees. For $k=2$, for instance, $\gamma_{n,2}^*\approx 0.708$, while $v_{n,2}^*\approx 0.704$, roughly a $0.5\%$ difference.}

%We remark that there is a gap between $v_{n,k}^*$ and $\gamma_{n,k}^*$. One part of this gap is introduced because we use a lower bound on the ratio $\E[X_A]/\E[\max_{i\in [n]}X_i]$. For instance, for $k=1$, $v_{\infty,1}^*=6/\pi^2 \approx 0.607$ (see Corollary~\ref{cor:value_k_1} in Section~\ref{sec:asymptotic_analysis}), while it is known that $\gamma_{\infty,1}^*\doteq  \lim_{n\to \infty} \gamma_{n,1}^*=1-1/e\approx 0.632$ \cite{correa2021posted}. Equivalently, it can be shown that the dual of the linearized version of $(P)_{n,k}$ (see Section~\ref{sec:existence_characterization_solutions}) has a feasible region that strictly contains (the inverse of) all possible distributions for the prophet inequality problem. A second part of the gap is introduced because we use intervals $[1,i_1],\ldots,[i_{k-1}+1,n]$ of equal length. ).

\subsection{Organization}

We follow this section with a summary of related work. In Section~\ref{sec:prophet_bounded_threshold}, we present our algorithm, and we show that the solutions of $(P)_{n,k}$ offer an approximation that implies a prophet inequality for our problem. We also recover the guarantee for $k=1$. Towards the end of the section, we provide $(CLP)_{n,k}$ as a relaxation of $(D)_{n,k}$ (restriction of $(P)_{n,k}$). In Section~\ref{sec:existence_characterization_solutions}, we characterize the optimal solutions of $(CLP)_{n,k}$ for a particular choice of size of intervals. Since the optimal value of $(CLP)_{n,k}$ is hard to analyze for finite $n$, in Section~\ref{sec:asymptotic_analysis}, we present an asymptotic analysis of $v_{n,k}^*$ by studying the infinite model, when $n\to \infty$. In Section~\ref{sec:small_thresholds}, we present analysis for small number of thresholds $k$ using $(D)_{n,k}$ and $(CLP)_{n,k}$.

\section{Related Work}

%Extensive work in stochastic, tcs community

The prophet inequality problem was introduced by Krengel \& Sucheston~\cite{krengel1977semiamarts}. The initial solution was fully dynamic and obtained via dynamic programming, with a tight approximation factor guarantee of $1/2$ of the prophet value; see also~\cite{hill1981ratio}. Samuel-Cahn~\cite{samuel1984comparison} showed that a simple threshold strategy is enough to achieve the same guarantee; a byproduct is that the order in which the random variables are observed is immaterial. The survey~\cite{hill1992survey} offers a classical overview of prophet inequalities, while \cite{correa2019recent} presents recent advancements.

The study of the i.i.d.\ prophet inequality problem started with Hill \& Kertz~\cite{hill1982comparisons}. For every $n$, they provided constants $1.1<a_n<1.6$ so that
\[
\textstyle\E\left[\max_{i\in [n]} X_i\right] \leq a_n \sup \left\{\E[X_\tau]: \tau \in \mathcal{T}_n\right\},
\]
where $\mathcal{T}_n$ is the set of stopping times for $X_1,\ldots, X_n$; the bound $a_n$ is best possible. Later, Kertz~\cite{kertz1986stop} showed that $a_n$ has a limit $\bar{\beta}\approx 1.341$, which is the unique solution to the equation $\int_0^1 (y(1-\log y)+(\beta-1))^{-1} \, \mathrm{d}y=1$. Nevertheless, for finite $n$, the best known prophet inequality was $(1-1/e)$ until \cite{abolhassani2017beating} showed that a prophet inequality of $0.738$ was possible. Recently, \cite{correa2021posted} showed that the approximation $1/\bar{\beta} \approx 0.745$ of the prophet value is achievable for any finite $n$. This approach is based on \modif{quantile} stopping rules, computing $n$ probabilities of acceptance $q_1<\cdots < q_n$ (independent of the input random variables) and converting these into thresholds via $q_t = \Prob(X_t \geq x_t)$. The values $q_t$ are obtained from distributions in a similar fashion to our method, and for $k=n$ thresholds, we recover their distribution. Nevertheless, unlike their approach, we obtain the distributions as a byproduct of an optimization problem. The idea of \modif{quantile} stopping rules can be found in \cite{samuel1984comparison} and has been used to construct strategies for other problems, such as the prophet secretary problem \cite{correa2021prophet}.

The theory of prophet inequalities has resurfaced in recent years because of its connections with auctions and, in particular, with posted priced mechanisms (PPM), which are now used in online sales (see, e.g.,~\cite{alaei2014bayesian,chawla2010multi,dutting2020prophet,hajiaghayi2007automated,kleinberg2012matroid,rubinstein2017combinatorial}). The first works showcasing the applicability of prophet inequalities in PPMs were~\cite{chawla2010multi,hajiaghayi2007automated}. They show that any prophet-type inequality implies a PPM with the same approximation guarantee; the converse was recently shown in~\cite{correa2019pricing}. In particular, our results for $k$-dynamic algorithms imply PPMs with few prices. The connection with pricing problems has motivated extensions of prophet inequalities with knapsack constraints~\cite{jiang2022tight,dutting2020prophet}, matroid constraints~\cite{hajiaghayi2007automated,kleinberg2012matroid} and other combinatorial constraints~\cite{rubinstein2017combinatorial,ehsani2018prophet}. For an overview of prophet inequalities and pricing, see~\cite{lucier2017economic}. Pricing problems with limited prices have recently been studied in~\cite{alaei2022descending}, in the context of multi-unit prophet inequalities where at most $m$ values are selected. This model selects $k$ prices and runs at most $k$ passes over the values, selecting at most $m$. In contrast, we make only one pass over the values; see also~\cite{jiang2022tight,dutting2016revenue,chawla2020static}.

Recent works have also focused on models without complete knowledge of the underlying distribution of the random variables. For the general case where random variables do not necessarily share the same distribution, it is known that one sample from each distribution is enough to guarantee (asymptotically in $n$) an approximation of $1/2$ of the prophet value; see, e.g.,~\cite{azar2014prophet}. For the i.i.d.\ prophet inequality problem, more samples are required to improve the guarantee. In~\cite{correa2019prophet}, it was shown that $\mathcal{O}(n^2/\varepsilon)$ samples from the distribution are enough to learn it and guarantee a factor $\bar{\gamma}-\mathcal{O}(\varepsilon)$. This result was improved in~\cite{rubinstein2019optimal} by requiring only $\mathcal{O}(n/\varepsilon^6)$ samples. Even though in this article we assume knowledge about the common distribution of the random variables, our results extend to the sampling setting by adapting the methods from \cite{correa2019prophet,rubinstein2019optimal}.

Crucial to iur analysis are linear programs. Various linear programs have appeared in the design of algorithms for online/sequential problems. Examples include online/stochastic matching~\cite{mehta2007adwords,goyal2019online} and online knapsack~\cite{babaioff2007knapsack,kesselheim2014primal}. Closer to our problem is the design and analysis of algorithms via linear programs for secretary problems~\cite{buchbinder2014secretary,chan2014revealing}.
In prophet inequality problems, the literature has several factor-revealing linear programs, e.g., see~\cite{feldman2016online,lee2018optimal}. Recently, \cite{jiang2022tight} and~\cite{jiang2022tightness} presented tight guarantees for the multi-unit prophet inequality problem using a linear program. In contrast to this and other works, our (infinite-dimensional) linear program is indexed by \modif{quantiles} $q\in [0,1]$, where $q$ is the tail probability of obtaining a value of at least $x$. This differs from formulations in the literature based on the support of the random variables, with most of them assuming finite supports. The closest that we can find to our work is the work~\cite{li2022query} where they used a quantile-based formulation to provide guarantees over their algorithms; however, they do not use their formulations to deduce policies as we do in our work. In our work, we obtain policies naturally from the infinite linear program. To our knowledge, linear programs similar to ours have not appeared in the literature.

\modif{We use an infinite model, where we let $n$ tend to infinity to analyze the behavior of our approximations. The infinite model has connection with time-based arrival models in the $[0,1]$ interval~\cite{bubna2022prophet,peng2022order}; other continuous models include Poisson arrivals~\cite{allaart2007prophet}. In our work, the use of the infinite model allows us to deduce approximate solutions to the ODE that encodes the optimal approximation.}

%% file: Main.tex
\section{Prophet Inequality with a Bounded Number of Thresholds}\label{sec:prophet_bounded_threshold}

%\subsection{Prophet Inequality for $k=1$ Threshold}

\modif{In this section, we introduce our algorithmic methodology and analysis. Later we show how to recover the known result of $\gamma_{n,1}^* = 1-(1-1/n)^n$ for a single threshold, $ k = 1 $. We briefly discuss the challenges for higher values of $k$ and provide the relaxation of $(D)_{n,k}$ that we use later to obtain near-optimal solutions.}

\modif{Given $n$ and $k$, we split the $n$ arriving random variables into $k$ consecutive intervals of sizes $\tau_1,\ldots,\tau_k$, where $\tau_1+\cdots + \tau_k=n$. Namely, $I_t=[i_{t-1}+1,i_t]$ for $t\in [k]$, where $i_0=0$ and $i_t=i_{t-1}+\tau_t$ for $t\in [k]$. We solve an infinite-dimensional linear program that receives $n,k,\tau_1,\ldots,\tau_k$ as parameters and use its solution to compute distributions from which we sample $k$ numbers $q_1,\ldots,q_k\in [0,1]$. We compute $k$ thresholds ${x}_1,\ldots,{x}_k$ from these numbers via $q_t=\Prob(X \geq {x}_t)$, and the algorithm accepts the first value that exceeds its corresponding threshold. We assume the input distribution $\mathcal{D}$ is continuous; this is a typical assumption in the literature (e.g., \cite{liu2021variable}), since we can smooth a discrete distribution by adding random noise at the expense of a loss in value that can be made arbitrarily small. % Scan variables in the $t$-th interval $I_t$ and accept the first value observed that is at least $x_t$.
	%
	%\ps{Added continuity assumption over $\mathcal{D}$ here.}
	We present the formal meta-algorithm in Algorithm~\ref{alg:meta_alg}.
	\begin{center}
		\begin{algorithm}[H]
			\KwIn{Integers $\tau_1,\ldots,\tau_k$. Continuous distribution $\mathcal{D}$ with nonnegative support. Functions $\alpha_1,\ldots,\alpha_k:[0,1]\to \R_+$.}
			%		\KwData{this text}
			%	\KwResult{how to write algorithm with \LaTeX2e }
			Let $i_0=0$.\\
			\For{$t=1,\ldots,k$}
			{
				Sample $q_t \sim \mathcal{Q}_t$, where $\mathcal{Q}_t$ is the distribution with CDF $\int_0^q \alpha_{t,q} \, \mathrm{d}q/\int_0^1 \alpha_{t,q}\,\mathrm{d}q$.\\
				Compute largest $x_t$ such that $q_t= \Prob(X\geq {x}_t)$. \\
				Upon scanning $X_{i_{t-1}+1},\ldots, X_{i_{t-1}+\tau_t}$, accept the first value that is at least $x_t$, if any.\\
				Update $i_t=i_{t-1}+\tau_t$.
			}
			%\Return $S_k$
			\caption{Meta-Algorithm} \label{alg:meta_alg}
		\end{algorithm}
	\end{center}
	The meta-algorithm (Algorithm~\ref{alg:meta_alg}) accepts \emph{any} set of non-negative integrable functions $\{\alpha_{t,q}\}_{t}$ with $\int_0^1 \alpha_{t,q}\,\mathrm{d}q > 0$. We obtain the function $\{\alpha_{t,q}\}_{t}$ from the dual of $(P)_{n,k}$ in Theorem~\ref{prop:dual_P_nk}. In the next proposition, we show that providing a solution of $(P)_{n,k}$ as input to Algorithm~\ref{alg:meta_alg} yields the approximation guarantee given by the objective value of that solution.
	\begin{proposition}\label{prop:approximation_relaxed_LP}
		Suppose that $(\alpha=\{\alpha_{t,q} \}_{t\in [k]},v)$ is a feasible solution for $(P)_{n,k}$. Then Algorithm~\ref{alg:meta_alg}, run with the input $\tau_1,\ldots,\tau_k$ and $\alpha_1,\ldots,\alpha_k$, guarantees $\E[X_A]\geq v \E[\max_i X_i]$ for any continuous distribution $\mathcal{D}$ with nonnegative support and $X_1,\ldots,X_n\sim \mathcal{D}$.
	\end{proposition}
	\begin{proof}
		Without loss of generality, we assume that $\alpha$ satisfies constraints~\eqref{const:constr_1_P_nk}-\eqref{const:constr_t_P_nk} at equality. Let $F$ be the cumulative distribution function (CDF) of $\mathcal{D}$ (and thus $X_1,\ldots,X_n$). For simplicity, we refer to Algorithm~\ref{alg:meta_alg} by $A$ when run with $\alpha$ and $\tau_1,\ldots,\tau_k$. We also denote by $M_t=\int_0^1 \alpha_{t,q}\, \mathrm{d}q$ the integral of $\alpha_{t,q}$ in $[0,1]$. Let $Q_t$ be the random value $q_t$ selected during the scanning of interval $I_t$; hence, the density of $Q_t$, the quantile of the $t$-th interval in Algorithm~\ref{alg:meta_alg}, is $\alpha_{t,q}/M_t$.
		We say that $A$ \emph{reaches interval $I_t$} if $A$ has not chosen a value in the intervals $I_1,\ldots,I_{t-1}$.
		\begin{claim*}
			The probability that $A$ reaches interval $I_t$ is $\Prob(A\text{ reaches }I_t)= \int_0^1 (1-q)^{\tau_{t-1}} \alpha_{t-1,q}\,\mathrm{d}q=M_t$, with $\alpha_{0,q}=1$ for all $q\in [0,1]$ and $\tau_0=0$.
		\end{claim*}
		\begin{proof}[Proof of claim]
			We prove this by induction in $t$. For $t=1$, we have,
			\[
			\Prob(A \text{ reaches }I_1)= 1 = \int_0^1 \alpha_{0,q}\, \mathrm{d}q=\int_0^1 \alpha_{1,q} \, \mathrm{d}q=M_1.
			\]
			Assume that the result is true for $t\geq 1$ and let us show it for $t+1$. Note that
			\[
			\Prob(A\text{ reaches }I_{t})  = \int_0^1 (1-q)^{\tau_{t-1}} \alpha_{t-1,q}\, \mathrm{d}q = \int_0^1 \alpha_{t,q} \, \mathrm{d}q=M_t,
			\]
			using \eqref{const:constr_1_P_nk}-\eqref{const:constr_t_P_nk}. Then, using $q=\Prob(X \geq \overline{x})$ if and only if $\overline{x}=F^{-1}(1-q)$, we have
			\begin{align*}
				\Prob(A\text{ reaches }I_{t+1}) & = \Prob(A\text{ reaches }I_t)\times \Prob(\text{No value }\geq x_t \text{ in }I_t)\\
				& = M_t \times \int_0^1 \frac{\alpha_{t,q}}{M_t} \Prob(\text{No value }\geq x_t \text{ in }I_t\mid Q_t=q) \, \mathrm{d}q \\
				& = \int_0^1 \alpha_{t,q}\Prob(\forall i=1,\ldots ,\tau_t : X_i < F^{-1}(1-q) ) \, \mathrm{d}q \tag{using $x_t=F^{-1}(1-Q_t)$} \\
				& = \int_0^1 \alpha_{t,q} \Prob(X_t < F^{-1}(1-q))^{\tau_t} \, \mathrm{d}q \\
				& = \int_0^1 \alpha_{t,q} (1-q)^{\tau_t} \, \mathrm{d}q.
			\end{align*}
			In the first equality, we used the independence of values in $I_t$ with respect to values in $I_1,\ldots,I_{t-1}$. From the first to second equality, we conditioned on the event $Q_t=q$ and the inductive hypothesis. The equality $\Prob(A\text{ reaches }I_{t+1})=M_{t+1}$ follows from constraint~\eqref{const:constr_t_P_nk} and the assumption that it is tight.
		\end{proof}
		\begin{claim*}
			The expected value obtained in interval $I_t$ by $A$, conditioned on being reached, is
			\[
			\int_0^1 F^{-1}(1-u) \int_u^1 \frac{\alpha_{t,q}}{M_t} \left( \frac{1-(1-q)^{\tau_t}}{q} \right) \, \mathrm{d}q\, \mathrm{d}u .
			\]
		\end{claim*}
		\begin{proof}[Proof of claim]
			Suppose we sample $Q_t=q$ for interval $I_t$; then, if we ever observe an $X_i$ with value at least $x_t=F^{-1}(1-Q_t)$ in $I_t$, the expected value obtained equals $\E[X\mid X\geq x_t]$. Thus, the expected value obtained in interval $I_t$ can be computed by conditioning on $Q_t=q$,
			\begin{align*}
				&\int_0^1 \frac{\alpha_{t,q}}{M_t} \Prob(\exists i=1,\ldots, \tau_t : X_i \geq F^{-1}(1-q)) \E[X\mid X\geq F^{-1}(1-q)] \, \mathrm{d}q \\
				&= \int_0^1 \frac{\alpha_{t,q}}{M_t} \left( 1- (1-q)^{\tau_t} \right) \frac{1}{q} \int_{F^{-1}(1-q)}^\infty x \, \mathrm{d}F(x) \, \mathrm{d}q\\
				& = \int_0^1 \frac{\alpha_{t,q}}{M_t} \left( \frac{1-(1-q)^{\tau_t}}{q}\right) \int_0^q F^{-1}(1-u) \, \mathrm{d}u \, \mathrm{d}q . \tag{Change of variable $x=F^{-1}(1-u)$.}
			\end{align*}
			The conclusion follows by changing the order of integration between $u$ and $q$.
		\end{proof}
		Using these two claims together, we obtain
		\begin{align*}
			\E[X_A] & = \sum_{t=1}^k \int_0^1 F^{-1}(1-u) \int_u^1 \alpha_{t,q}\left( \frac{1-(1-q)^{\tau_t}}{q} \right)\,\mathrm{d}q\, \mathrm{d}u.
		\end{align*}
		A similar calculation yields
		\[
		\E[\textstyle{\max_i} X_i]  = \int_0^1 F^{-1}(1-u) n (1-u)^{n-1}\, \mathrm{d}u.
		\]
		With this, we can show that $\E[X_A]\geq v\cdot \E[\textstyle{\max_i} X_i]$.
		The proof of this inequality is akin to the proof of weak duality in Theorem~\ref{prop:dual_P_nk} which we provide in Appendix~\ref{app:strong_dual}; hence we skip it here. This finishes the proof and shows that solutions to $(P)_{n,k}$ provide $k$-dynamic algorithms with guaranteed prophet inequalities.
	\end{proof}}

\subsection{Exact Solution for a Single Threshold}

\modif{In this subsection, we utilize $(P)_{n,k}$ to provide the optimal solution for the single-threshold algorithm and its optimal approximation guarantee. %This is summarized in the following proposition.
	\begin{proposition}
		For any $n$, $\gamma_{n,1}^*=1-(1-1/n)^n$. In particular, $\gamma_{n,1}^*\geq 1-1/e$.
	\end{proposition}
	\begin{proof}
		%	We now provide $\gamma_{n,1}^*\geq 1-1/e$.
		We denote by $\delta_{\{u_0\}}:[0,1]\to \R_+$ the dirac function with all the mass in $u_0\in [0,1]$. Let $\alpha_{1,q}=\delta_{\{1/n\}}(q)$, $v= 1- (1-1/n)^n$ and $\eta:[0,1]\to \R_+$ satisfy
		\[
		\eta_u = \begin{cases}
			v\left(  nu + (1-u)^n - 1 \right)  & u \leq 1/n\\
			v(1-u)^n & u> 1/n .
		\end{cases}
		\]
		It is easy to verify that $(\alpha,\eta,v)$ is feasible for $(P)_{n,k}$; hence, by Proposition~\ref{prop:dual_P_nk}, we obtain
		\[
		\gamma_{n,1}^* \geq  1- \left( 1 - \frac{1}{n} \right)^n .
		\]
		We show the equality by providing a solution to $(D)_{n,1}$. Consider the following function $f:[0,1]\to \R_+$,
	\[
	f_q = \frac{(1-1/n)^n}{n(1-(1-1/n)^n)} \delta_{\{0\}}(q) + \left(\frac{1-2(1-1/n)^n}{1-(1-1/n)^n}\right) \mathbf{1}_{(0,1)}(q).
	\]
	It can be shown that for $d=1-(1-1/n)^n$, the pair $(d,f)$ is a feasible solution to $(D)_{n,1}$ with objective value $1-(1-1/n)^n$. Indeed, the optimal value of
	\begin{align*}
		&\quad \sup_{q\in [0,1]} \left\{  \left( \frac{1-(1-q)^n}{q} \right)\int_0^q f_u \, \mathrm{d}u \right\} \\
		&= \sup_{q\in [0,1]} \left\{  \left( \frac{1-(1-q)^n}{q} \right) \left( \frac{(1-1/n)^n}{n(1-(1-1/n)^n)}+ \left(\frac{1-2(1-1/n)^n}{1-(1-1/n)^n}\right) q \right)  \right\} ,
	\end{align*}
	occurs at $q=1/n$, with optimal value $d$. Due to the weak duality between $(D)_{n,1}$ and $(P)_{n,1}$, we conclude that $\gamma_{n,1}^*=1-(1-1/n)^n$ for all $n$. Using the inequality $1+x\leq e^x$, we conclude $\gamma_{n,1}^* \geq 1-1/e$.
\end{proof}
%\begin{remark}
Notice that the solution for $k=1$ takes the form of a delta function at $1/n$. When using it as input for Algorithm~\ref{alg:meta_alg}, the algorithm behaves deterministically, always choosing the threshold $x_1$ satisfying $\Prob(X\geq x_1)=1/n$. % (assuming that $X$ is continuous).
%\end{remark}
}

\subsection{Relaxation of the Exact Formulation}

\modif{In the last subsection, we provided an exact solution to $(P)_{n,k}$ for $k=1$. Unfortunately, providing closed-form solutions for $(P)_{n,k}$ becomes a difficult task even for $k=2$. We provide this analysis in Section~\ref{sec:small_thresholds}, showing $\gamma_{n,2}^* \approx 0.708$. As in the single-threshold case, the optimal solution for $k=2$ also exhibits the form of deterministic quantiles, which we compute in closed form. In general, to prove the optimality of a solution to $(P)_{n,k}$, we must exhibit a matching solution to $(D)_{n,k}$. We do this in Section~\ref{sec:small_thresholds} for $k=2$, but this strategy is difficult to generalize for larger values of $k$. To illustrate this point, in Appendix~\ref{app:small_thresholds} we present the dual solution to $(D)_{n,2}$, which occupies a full manuscript page in small font.}

\modif{We continue the analysis for larger values of $k$ by relaxing Constraint~\eqref{const:nonincreasing} in $(D)_{n,k}$, which is equivalent to restricting $(P)_{n,k}$ with $\eta=0$. The resulting linear program is $(CLP)_{n,k}$ in Theorem~\ref{thm:main1}, which we rewrite below for concreteness.}
\begin{subequations}\label{form:CLP_nk}
	\begin{align}
		\sup_{ \alpha:[k]\times [0,1]\to \R_+ }\quad & v  \notag\\
		(CLP)_{n,k} \qquad\qquad\qquad  s.t. \quad & \int_0^1 \alpha_{1,q}\, \mathrm{d}q \leq 1, & \label{const:infinity_dynamic_1} \\
		\quad & \int_0^1 \alpha_{t+1,q} \, \mathrm{d}q\leq \int_0^1 (1-q)^{\tau_t} \alpha_{t,q} \, \mathrm{d}q,  & \forall t\in [k-1] ,\label{const:infinity_dynamic_2} \\
		\quad &   v n (1-u)^{n-1} \leq \sum_{t=1}^k\int_u^1 \left( \frac{1-(1-q)^{\tau_t}}{q}   \right) \alpha_{t,q} \, \mathrm{d}q, & \forall u\in [0,1]. \label{const:infinity_multiobj}
	\end{align}
\end{subequations}

Dropping the a.e.\ condition in Constraint~\eqref{const:infinity_multiobj} does not affect optimality. Since solutions of $(CLP)_{n,k}$ are feasible for $(P)_{n,k}$, then, Proposition~\ref{prop:approximation_relaxed_LP} also guarantees that with solutions of $(CLP)_{n,k}$ we obtain prophet inequalities in Algorithm~\ref{alg:meta_alg}. In the next section, for the case $\tau_1=\cdots = \tau_{k-1} \geq \tau_k \geq 0$, we characterize the optimal solution of $(CLP)_{n,k}$. Because $(CLP)_{n,k}$ is obtained by relaxing $(D)_{n,k}$, a consequence is $\gamma_{n,k}^*\geq v_{n,k}^*$, where $v_{n,k}^*$ is the optimal value of $(CLP)_{n,k}$. This proves the first part of Theorem~\ref{thm:main1}.

\modif{%\begin{remark} % ATH: Not sure this needs to be a highlighted remark?
Furthermore, our findings in Section~\ref{sec:small_thresholds} show that droping Constraint~\ref{const:nonincreasing} causes only a minimal loss in the approximation of the prophet inequalities for $k\geq 2$. For example, for $k=2$ we have $v_{n,2}^*\approx 0.704$, while $\gamma_{n,2}^* \approx 0.708$ (see Appendix~\ref{app:threshold_k_2}). Moreover, for $k=5$, $v_{n,k}^*$ is already roughly 99\% of $\bar{\gamma}\approx 0.745$.
%\end{remark}
}

\section{\modif{Relaxation with} Equidistant Thresholds}\label{sec:existence_characterization_solutions}

In this section, we characterize the optimal solution of $(CLP)_{n,k}$ when $ \tau_1= \dotsb = \tau_{k-1} \geq \tau_k \geq 0$, proving the second part of Theorem~\ref{thm:main1}. Since this optimal solution is hard to study for finite values of $n$, in the next section (Section~\ref{sec:asymptotic_analysis}), we study its asymptotic behavior in $n$ for any fixed $k$. %We provide tight guarantees for $\lim_{n \to \infty} v_{n,k}^*$.

\modif{Let $\tau=\tau_1=\cdots =\tau_{k-1}$ and $\sigma=\tau_k = n-(k-1)\tau\geq 0$; we assume that $\tau,\sigma\geq 1$. In order to characterize the optimal solution of $(CLP)_{n,k}$, we use its dual, denoted $(DLP)_{n,k}$, where dual variables correspond to nonnegative measures in $[0,1]$. We use $(CLP)_{n,k}$ and its dual, instead of directly relaxing Constraints~\eqref{const:nonincreasing} in $(D)_{n,k}$, because the linearization allows us to more easily construct optimal solutions. $(DLP)_{n,k}$ is given by}
%
%; call this relaxation $(D')_{n,k}$. The values of $(CLP)_{n,k}$ and $(D')_{n,k}$ coincide. However, it is not evident that $(D')_{n,k}$ has solutions that attain their optimal value. The general dual of $(CLP)_{n,k}$, as opposed to $(D')_{n,k}$, is more useful in our case as we can find explicitly its optimal solution. The dual of $(CLP)_{n,k}$ corresponds to the problem $(DLP)_{n,k}$,}
%
\begin{subequations}
	\begin{align}
		\inf_{\substack{\mu\in \mathcal{M}_+[0,1] \\ d\geq 0} } \quad & d_1 \notag\\
		(DLP)_{n,k}\qquad\qquad \text{s.t.} \quad & d_t \geq\left(  \frac{1-(1-q)^{\tau_t}}{q} \right) \mu[0,q] + (1-q)^{\tau_t} d_{t+1} , & \forall t\in[k],\forall q\in [0,1] \label{const:d_infinity_dynamic} \\
		\quad & 1 \leq \int_0^1 n(1-u)^{n-1} \, \mathrm{d}\mu(u).  & \label{const:d_infinity_multiobj}
	\end{align}
\end{subequations}
%
%{ $(DLP)_{n,k}$ \begin{tabular}{cp{0.75\linewidth}}
%		$\displaystyle\inf_{\substack{\mu\in \mathcal{M}_+[0,1] \\ d\geq 0} }$ &  $\quad d_1$ \\
%		s.t. &  {\vspace{-1.3cm}  \begin{align}
%				d_t& \geq \left(  \frac{1-(1-q)^{\tau_t}}{q} \right) \mu[0,q] + (1-q)^{\tau_t} d_{t+1} & \forall t\in [k], \forall q \in [0,1] \label{const:d_infinity_dynamic} \\
%				1 &\leq \int_0^1 n(1-u)^{n-1} \, \mathrm{d}\mu(u),  & \label{const:d_infinity_multiobj}
%			\end{align}\vspace{-0.5cm}}
%\end{tabular}}
%
where $\mathcal{M}_+[0,1]$ is the space of positive Borel measures over $[0,1]$. We also use the notation $\int_a^b f(u)\, \mathrm{d}\mu(u) = \int_{[a,b]} f(u) \mathrm{d}\mu(u)$.
%
%A routine calculation shows the weak duality between the programs $(CLP)_{n,k}$ and $(DLP)_{n,k}$.

\begin{proposition}[Weak duality]
	For any feasible $(\alpha,v)$ for $(CLP)_{n,k}$ and feasible $(\mu,d)$ for $(DLP)_{n,k}$ we have $d_1 \geq v$.
\end{proposition}
The proof of this proposition is similar to the proof of weak duality in Theorem~\ref{prop:dual_P_nk} (see Appendix~\ref{app:strong_dual}) and it is skipped for brevity.

We prove strong duality between $(CLP)_{n,k}$ and $(DLP)_{n,k}$ by exhibiting primal and dual solutions with the same objective value. We devote the rest of the section to this task, proving the next result.
\begin{theorem}
	Strong duality holds between $(CLP)_{n,k}$ and $(DLP)_{n,k}$. Moreover, we can fully characterize their solution when $\tau_1=\cdots=\tau_{k-1}=\tau$, with $n/k \leq \tau \leq n/(k-1)$.
\end{theorem}

\subsection{Primal Feasible Solution}

Consider the functions
\[
\alpha_{t,q} = \begin{cases}
	v \frac{q}{1-(1-q)^{\tau}} n(n-1) (1-q)^{n-2} \mathbf{1}_{[\varepsilon_{t-1},\varepsilon_t]}(q) & t=1,\ldots,k-1 \\
	v \frac{q}{1-(1-q)^{\sigma}} n(n-1) (1-q)^{n-2} \mathbf{1}_{[\varepsilon_{k-1},\varepsilon_k]}(q) & t=k ,
\end{cases}
\]
where $0=\varepsilon_0 \leq \varepsilon_1 \leq \cdots \leq \varepsilon_k\leq 1$. The goal is to demonstrate that there is a choice of $\varepsilon_1,\ldots,\varepsilon_k$ such that $(\alpha,v)$ is a feasible solution for $(CLP)_{n,k}$. Using Constraints~\eqref{const:infinity_dynamic_1}-\eqref{const:infinity_dynamic_2}, this is the case if
\begin{align*}
	1 &= v n(n-1) \int_{0}^{\varepsilon_1} \frac{q}{1-(1-q)^{\tau}} (1-q)^{n-2} \, \mathrm{d}q \\
	\int_{\varepsilon_{t-1}}^{\varepsilon_t} \frac{(1-q)^{\tau} q(1-q)^{n-2} }{1-(1-q)^{\tau}} \, \mathrm{d}q & = \int_{\varepsilon_t}^{\varepsilon_{t+1}} \frac{q(1-q)^{n-2}}{1-(1-q)^{\tau}} \, \mathrm{d}q,\quad  t=1,\ldots,k-2,\\
	\int_{\varepsilon_{k-2}}^{\varepsilon_{k-1}} \frac{(1-q)^{\tau} q(1-q)^{n-2} }{1-(1-q)^{\tau}} \, \mathrm{d}q & = \int_{\varepsilon_{k-1}}^{\varepsilon_{k}} \frac{q(1-q)^{n-2}}{1-(1-q)^{\sigma}} \, \mathrm{d}q,
\end{align*}
and $\varepsilon_k=1$; this last constraint on $\varepsilon_k$ is necessary to satisfy Constraint~\eqref{const:infinity_multiobj}. From this system of equations, for a fixed $v\geq 0$ we can define $\varepsilon_1=\varepsilon_1(v)$ uniquely. Once $\varepsilon_1$ is defined, $\varepsilon_2$ is uniquely defined by $\varepsilon_1$, itself defined as a function of $v$. In general, every $\varepsilon_t$ is uniquely defined as a function of $v$. In the next proposition we show that the $\varepsilon_t$ are decreasing as a function of $v$, and that there is a $v^*\geq 0$ such that $\varepsilon_k(v^*)=1$.

\begin{proposition}\label{prop:epsilon_decreasing_gamma}
	For $t=1,\ldots,k$, $\varepsilon_t=\varepsilon_t(v)$ is differentiable and strictly decreasing in $v$. Moreover, $\varepsilon_t' \leq - 1/(v^2 n(n-1))$.
\end{proposition}

The proof of this proposition is deferred to Appendix~\ref{app:existence_characterization_solutions}. As a corollary, we obtain that $\alpha$ defined as above is feasible for $(CLP)_{n,k}$.

\begin{corollary}
	There is a $v^* > 0$ such that $\varepsilon_{k}(v^*)=1$. Hence, the solution $(\alpha, v^*)$ is feasible for $(CLP)_{n,k}$ and has objective value equal to $v^*$. Thus, $v_{n,k}^* \geq v^*$.
\end{corollary}
\begin{proof}
	A simple calculation shows that $\varepsilon_k(1) < 1$ (see Proposition~\ref{prop:eps_k_1_less_1} in Appendix~\ref{app:existence_characterization_solutions}). Also, for any differentiable function $f:\R\to \R_+$ such that $f'(x)\leq -c/x^2$ we must have $\lim_{x\to 0+} f(x) = +\infty$ (here $c$ is a constant). Thus, using the previous proposition, we can find a value $0 < v^* < 1$ such that $\varepsilon_k=\varepsilon_k(v^*) = 1$. The remaining conclusions follow from the construction of $\alpha$ and $v^*$.
\end{proof}

\subsection{Dual Feasible Solution and Strong Duality}

In this subsection, we construct a feasible solution for $(DLP)_{n,k}$ with objective value $v^*$. This shows that strong duality holds and $v_{n,k}^*=v^*$. Let $\varepsilon_1\leq \cdots \leq \varepsilon_k$ be the values obtained in the previous subsection. We define a set of auxiliary quantities that help define a measure $\mu$. With this measure $\mu$, we define a sequence of $d=\{d_t\}_{t\in [k]}$ such that $(\mu,d)$ is feasible for $(DLP)_{n,k}$ and has objective value $ d_1 = v^*$. Recall that $\tau=\tau_1=\cdots = \tau_{k-1}$ and $\sigma=\tau_k$. For $r=1,\ldots,k-1$, let
\[
h_r = (1-\varepsilon_r)^{\tau_{r+1}} \left(\frac{1-(1-\varepsilon_r)^{\tau_{t}}} {1-(1-\varepsilon_r)^{\tau_{t+1} }} \right)
\]
and %consider the values $a_t$ defined as follow
\[
a_t = v^* \left(\frac{1+ \sum_{s=t}^{k-1}\prod_{r=s}^{k-1}h_r}{1+ \sum_{s=1}^{k-1}\prod_{r=s}^{k-1}h_r} \right).
\]
%for $t=1,\ldots, k$.
These values hold several important properties that we list in the following proposition. The proof is a simple calculation, skipped for brevity.

\begin{proposition}\label{prop:properties_of_a_t}
	The values $a_1,\ldots,a_k$ satisfy the following properties:
	\begin{enumerate}
		\item For any $t\leq k-1$, $a_t \geq a_{t+1}$.
		
		\item For any $t\leq k-1$, $\displaystyle \frac{(a_t - a_{t+1}(1-\varepsilon_t)^{\tau_t}) \varepsilon_t}{1-(1-\varepsilon_t)^{\tau_t}} = \frac{(a_{t+1} - a_{t+2}(1-\varepsilon_t)^{\tau_{t+1}}) \varepsilon_t}{1-(1-\varepsilon_t)^{\tau_{t+1}}}$.
		
		\item For any $t\leq k-1$, $(a_t-a_{t+1}) = (a_{t+1}- a_{t+2}) h_t$.
	\end{enumerate}
\end{proposition}

Now, consider the function $F$ defined over $[0,1]$ as
\begin{align*}
	F(q) = \sum_{t=1}^{k}  \frac{(a_t - a_{t+1}(1-q)^{\tau_t}) q}{1-(1-q)^{\tau_t}} \mathbf{1}_{[\varepsilon_{t-1},\varepsilon_t)}(q),
\end{align*}
with $a_{k+1}=0$, $\tau=\tau_1=\cdots = \tau_{k-1}$ and $\sigma=\tau_k$. Then, by the construction of $a_1,\ldots,a_k$, $F$ is continuous over $[0,1)$. Moreover, $F$ is strictly increasing, which can be verified by deriving the function $F$ in each interval $(\varepsilon_{t-1},\varepsilon_t)$.

We extend $F$ to $\R$ by setting the extension to $0$ in $(-\infty,0)$ and $F(1)$ in $[1,+\infty)$. We keep denoting this extension by $F$. The function $F$ is right-continuous and nondecreasing. Thus, we can define the Lebesgue-Stieltjes measure $\mu_F$ generated by the function $F$ (see Chapter 1 in~\cite{folland1999real}). Note that $\mu_F \in \mathcal{M}_+[0,1]$, $\mu_F(0,q] = F(q)-F(0)$ and $\mu[0,q] = F(q)$.

We construct $d_k, d_{k-1},\ldots,d_1$ to satisfy Constraints~\eqref{const:d_infinity_dynamic} as %follows.
\begin{align*}
	d_t &= \sup_{q\in [0,1]}\left\{ \left( \frac{1-(1-q)^{\tau_t}}{q} \right) \mu_F[0,q]+ (1-q)^{\tau_t} d_{t+1}  \right\} \\
	&= \sup_{q\in [0,1]}\left\{ \left( \frac{1-(1-q)^{\tau_t}}{q} \right)F(q)+ (1-q)^{\tau_t} d_{t+1}  \right\} ,
\end{align*}
with $d_{k+1}=0$.

\begin{lemma}\label{lem:dual_feasibility}
	The solution $(\mu_F,v)$ is feasible for $(DLP)_{n,k}$. Moreover, it has objective value $d_1=v^*$.
\end{lemma}

In the rest of this subsection, we present the proof of this lemma. We verify first that $\mu_F$ satisfies Constraint~\eqref{const:d_infinity_multiobj}; indeed,
\begin{align*}
	\int_0^1 n (1-u)^{n-1} \mathrm{d}\mu_F(u) & = \int_{0}^1 n (n-1)\int_{u}^1 (1-q)^{n-2} \, \mathrm{d}q  \, \mathrm{d} \mu_F (u) \\
	& = n (n-1)\int_0^1 \int_0^q  \, \mathrm{d} \mu_F(u) (1-q)^{n-2} \, \mathrm{d}q \\
	& = n(n-1) \sum_{t=1}^k \int_{\varepsilon_{t-1}}^{\varepsilon_t}  \frac{(a_t - a_{t+1}(1-q)^{\tau_t}) q}{1-(1-q)^{\tau_t}} (1-q)^{n-2} \, \mathrm{d}q \\
	& = n(n-1)\left(a_1 \int_0^{\varepsilon_1} \frac{q(1-q)^{n-2}}{1-(1-q)^{\tau_1}} \, \mathrm{d}q  \right. \\
	&\quad \left. + \sum_{t=1}^{k-1} a_{t+1} \left(  \int_{\varepsilon_{t}}^{\varepsilon_{t+1}} \frac{q(1-q)^{n-2}}{1-(1-q)^{\tau_{t+1}}} \, \mathrm{d}q - \int_{\varepsilon_{t-1}}^{\varepsilon_t} \frac{q(1-q)^{\tau_t+n-2}}{1-(1-q)^{\tau_{t}}} \, \mathrm{d}q \right) \right) \\
	& = \frac{a_1}{v^*}=1.
\end{align*}

By construction, $(\mu_F,v)$ satisfies Constraint~\eqref{const:d_infinity_dynamic}. To finish the proof of Lemma~\ref{lem:dual_feasibility}, we need to show that $d_1=v^*$. This is a consequence of the following more general result.

\begin{proposition}\label{prop:backward_induction_vt_at}
	For any $t=1,\ldots,k$, we have $d_t=a_t$. In particular, for $t=1$, $d_1=a_1=v^*$.
\end{proposition}

To prove this proposition we define $g_t(q) = (1-(1-q)^{\tau_t}) q^{-1} F(q) +(1-q)^{\tau_t} a_{t+1}$. Starting at $t=k$, we show that $d_k=\sup_{s\in [k]} \{ g_k(\varepsilon_s)  \}$ and that $g_k(\varepsilon_1)\leq \cdots \leq g_k(\varepsilon_k)=a_k$. This shows $d_k=a_k$. Now, assuming that the result is true for $t+1$, we can proceed by backward induction and show that $d_t=\sup \{ \sup_{s\leq t}\{ g_t(\varepsilon_s)  \}, \sup_{s> t} \{ g_t(\varepsilon_{s-1}) \}  \}$. We again show $g_t(\varepsilon_1)\leq \cdots \leq g_t(\varepsilon_t)=a_t$ and $g_{t}(\varepsilon_{k-1})\leq g_t(\varepsilon_{k-2})\leq \cdots \leq g_{t}(\varepsilon_t)=a_t$ and this shows $d_t=a_t$. We defer the details to Appendix~\ref{app:existence_characterization_solutions}.

\section{Asymptotic Analysis}\label{sec:asymptotic_analysis}

We characterized the optimal solution of $(CLP)_{n,k}$ in the previous section. This characterization is implicit in terms of $\varepsilon_1,\ldots,\varepsilon_k$, and thus not easy to analyze for finite $n$. In this section, we focus on analyzing the optimal value $v_{n,k}^*$ for $k$ fixed and large values of $n$. For this, we let $n$ go to infinity and study a model independent of $n$ that we call the \emph{infinite model}. We also show that the infinite model is not far from the finite model; specifically, we show that $v_{n,k}^*$ converges to the optimal value $v_{\infty,k}^*$ of the infinite model, and, as a byproduct of our analysis, we show that from solutions to the infinite model, we can recover solutions to the finite model (see Theorem~\ref{thm:approximation_large_n}).

%\atnote{You go back and forth between using ``limit model'' and ``infinite model''. We need consistent terminology; I suggest ``infinite''.}

Our approach can be interpreted as follows. By re-scaling the index $q \in [0,1]$ in both $(CLP)_{n,k}$ and $(DLP)_{n,k}$ via the transformation $q=-\log y / n$ and letting $n$ go to infinity, we obtain two new problems that are independent of $n$. The limits of these optimization problems are %the following programs respectively.

{ $(CLP)_{\infty,k}$ \begin{tabular}{cp{0.75\linewidth}}
		$\displaystyle\sup_{\substack{\omega:[k] \times [0,1]\to \R_+ \\ v \geq 0}}$ &  $\quad v$ \\
		s.t. &  {\vspace{-1.5cm}  \begin{align}
				\int_0^1 \frac{\omega_{1,y}}{y} \, \mathrm{d}y& \leq 1 &  \label{const:limit_infinity_dynamic_1} \\
				\int_0^1 \frac{\omega_{t+1,y}}{y} \, \mathrm{d}y& \leq \int_0^1 y^{1/k}  \frac{\omega_{t,y}}{y} \, \mathrm{d}y  & \forall t\in [k-1]  \label{const:limit_infinity_dynamic_2} \\
				v \overline{y}&\leq  \sum_{t=1}^k\int_0^{\overline{y}} \left(\frac{1-y^{1/k}}{-y\log y}   \right) \omega_{t,y} \, \mathrm{d}y & \forall \overline{y}\in [0,1] , \label{const:limit_infinity_multiobj}
			\end{align}\vspace{-0.5cm}}
\end{tabular}}

and

{ $(DLP)_{\infty,k}$ \begin{tabular}{cp{0.75\linewidth}}
		$\displaystyle\inf_{\substack{\overline{\mu}\in \mathcal{M}_+[0,1] \\ d\geq 0}}$ &  $\quad d_1$ \\
		s.t. &  {\vspace{-1.5cm}  \begin{align}
				d_t& \geq \left( \frac{1-\overline{y}^{1/k}}{-\log \overline{y}} \right) \overline{\mu}[\overline{
				y},1] + \overline{y}^{1/k} d_{t+1} & \forall t\in[k], \forall \overline{y} \in [0,1] \label{const:limit_d_infinity_dynamic} \\
				1 &\leq \int_{[0,1]} y\, \mathrm{d}\overline{\mu}(y) .   & \label{const:limit_d_infinity_multiobj}
			\end{align}\vspace{-0.5cm}}
\end{tabular}}

%\atnote{Since $\omega_k$ is a function, I prefer functional notation, i.e.\ $\omega_1(y)$ instead of $\omega_{1,y}$.}\ps{It's a lot to work to change all the variables. I would also have to change the $\alpha$ in the previous formulation.}
\begin{proposition}[Existence of Solutions and Strong Duality]
	$(DLP)_{\infty,k}$ and $(CLP)_{\infty,k}$ both have finite optimal solutions and their optimal values coincide.
\end{proposition}
The proof is similar to the analogous proof for finite $n$, and is deferred to the Appendix. As in the previous section, we can characterize the optimal solution of $(CLP)_{\infty,k}$ as
\[
\omega_{t,y} = -v_{\infty,k}^* \frac{y\log y}{1-y^{1/k}}\mathbf{1}_{(y_{t},y_{t-1})}(y) ,
\]
with $y_0=1\geq y_1 \geq \cdots \geq y_k=0$, and $v_{\infty,k}^*$ is the optimal value of $(CLP)_{\infty,k}$. A byproduct of this characterization is the following simple result for $k=1$ threshold.

\begin{corollary}\label{cor:value_k_1}
	For $k=1$, we have $v_{\infty,1}^* = 6/\pi^2$.
\end{corollary}

\begin{proof}
	For $k=1$, the optimal solution of $(CLP)_{\infty,1}$ becomes
	\[
	\omega_{1,y} = - v_{\infty,1}^* \frac{y \log y}{1-y}.
	\]
	The value $v_{\infty,1}^*$ is then determined by
	\[
	v_{\infty,1}^* = \left(  \int_0^1  -\frac{\log y}{1-y} \, \mathrm{d}y  \right)^{-1} = \left(\sum_{k\geq 1} \frac{1}{k^2} \right)^{-1} = \frac{6}{\pi^2}.\qedhere
	\]
\end{proof}

Notice that $v_{\infty,k}^*$ is also defined implicitly as a function of $y_1,\ldots,y_k$. In Section~\ref{subsec:approximation_guarantee_infinite_model}, we show that $y_1,\ldots,y_k$ converge to the points produced by an ODE. This will allows us to produce tight guarantees for $v_{\infty,k}^*$. Before this, we show that the finite model converges to the infinite model.

\begin{theorem}\label{thm:approximation_large_n}
	Fix $k\geq 1$. There is a constant $c>0$ such that for large $n$
	\[
	v_{\infty,k}^* \left(  1 + 4 \frac{\log y_{k-1}}{n-1}  - c\frac{k^2}{n} \right)\leq v_{n,k}^* \leq v_{\infty,k}^* \left(  1 + c_{n,k} \right),
	\]
	where $c_{n,k}\to 0$ in $n$ for $k$ fixed. In particular, since $y_{k-1}$ does not depend on $n$, we have $v_{n,k}^* \to v_{\infty,k}^*$ as $n\to \infty$.
\end{theorem}

%\atnote{I find it curious to state the lower bound in detail but then only use $ c_{n,k} $ for the upper bound. Can we add more detail there, even if the proof is in the appendix?}

In the next subsection we show that $y_{k-1}\geq 1/32k$. Hence, the theorem implies a prophet inequality guarantee of $v_{n,k}^* \geq v_{\infty,k}^*(1- Ck^2/n)$, where $C$ is a constant. %It is worth remarking that Corollary~\ref{cor:value_k_1} and Theorem~\ref{thm:approximation_large_n} imply $v_{n,1}^* \to 6/\pi^2$. This quantifies the effect of removing Constraint~\eqref{const:nonincreasing} in $(D)_{n,k}$ for $k=1$.

To prove Theorem~\ref{thm:approximation_large_n}, we use the optimal solution of $(CLP)_{\infty,k}$, still denoted $(\omega,v_{\infty,k}^*)$, and we construct a solution $(\alpha,v^*)$ of $(CLP)_{n,k}$, where the value $v^*$ approaches $v_{\infty,k}^*$. This shows the lower bound for $v_{n,k}^*$. For the upper bound, we use the optimal solution $(\overline{\mu},d)$ of $(DLP)_{\infty,k}$ and construct a solution $(\mu, d^*)$ of $(DLP)_{n,k}$, where $d_1^*$ is at most $v_{\infty,k}^*(1+c_{n,k})$. In the remainder of the section, we prove the lower bound, because it also presents a way to construct policies from the infinite model.
The term $c_{n,k}$ in the upper bound of Theorem~\ref{thm:approximation_large_n} converges to $0$ due to Lebesgue's dominated converge theorem and it does not have a closed form that depends only on $k$ and $n$ as the lower bound does. We defer the complete proof of Theorem~\ref{thm:approximation_large_n} to Appendix~\ref{app:asymptotic_analysis}.

\subsection{Approximation of the Limit Guarantee}\label{subsec:approximation_guarantee_infinite_model}

In this subsection, we prove tight bounds on $v_{\infty,k}^*$ for large values of $k$, showing that $v_{\infty,k}^*\to \bar{\gamma}$ when $k\to \infty$. More specifically, we show the following result.
\begin{theorem}\label{thm:bound_infnite_model}
	There are constants $a,b\geq 0$ such that for $k$ large enough,
	\[
	\overline{\gamma}\left(  1 - a \frac{\log k}{k}  \right) \leq v_{\infty,k}^* \leq \overline{\gamma}\left(  1 - b \frac{\log k}{k}  \right),
	\]
	where $\bar{\beta}=1/\bar{\gamma} \approx 1.341$ is the unique solution of $\int_0^1 \left( w(1-\log w) +(\beta - 1)  \right)^{-1} \, \mathrm{d}w=1$. %\atnote{Notation overload again; you already used $c_1$ and $c_2$.}
\end{theorem}

The idea of the proof is the following. For $k$ large enough, the values $1=y_0 \geq y_1 \geq \cdots \geq y_{k-1}\geq y_k=0$ that define the optimal solution of $(CLP)_{\infty,k}$ approach points in the solution of the ordinary differential equation (ODE)
\begin{align*}
	w'(t) &= w(t)(\log w(t)-1)-(\bar{\beta}-1) \\
	w(0) &= 1 .
\end{align*}
%\atnote{You haven't defined $\beta_{\infty,k}^*$, which I'm assuming is $1 / v_{\infty,k}^*$?}
Specifically, we show that the sequence $x_t$ defined by Euler's method (see Chapter 1 of~\cite{iserles2009first}) of the ODE is $\mathcal{O}(\log k/k)$ away from $y_t$. This allows us to produce an approximate upper bound of $y_t$ in terms of $w(t/k)$, where $w$ is the solution of the ODE. If $\beta_{\infty,k}^*$ were larger than $(1+c) \overline{\beta}$, with $c= \mathcal{O}(\log k /k)$, our bound would imply $y_k < 0$, which is a contradiction. Therefore, $\beta_{\infty,k}^* \leq (1+c)\overline{\beta}$, which means $v_{\infty,k}^* \geq \overline{\gamma}(1-c)$. The proof of the upper bound is similar. In the remainder of the section, we formalize each of these steps for the lower bound on $v_{\infty,k}^*$; we defer the details of the upper bound to Appendix~\ref{app:upperbound_v}.

To ease notation, we write $\beta^*$ for $\beta_{\infty,k}^*$. We also assume  $\beta^* \geq 1.25$, which we  validate later. From the solution $(\omega,1/\beta^*)$ of $(CLP)_{\infty,k}$, we can deduce that the sequence $1=y_0\geq y_1 \geq \cdots \geq y_{k-1}\geq y_k=0$ satisfies
\begin{align*}
	\beta^* &= \int_{y_1}^{y_0} \frac{-\log y}{1-y^{1/k}} \, \mathrm{d}y \\
	\int_{y_{t+1}}^{y_{t}} \frac{-y^{1/k} \log y}{1-y^{1/k}} \, \mathrm{d}y & = \int_{y_{t}}^{y_{t-1}} \frac{-\log y }{1-y^{1/k}} \, \mathrm{d}y & t=1,\ldots,k-1 .
\end{align*}
From this system, we obtain the implicit recursion
\begin{align}
	\int_{y_{t+1}}^{y_t} \frac{- \log y}{1-y^{1/k}} \, \mathrm{d}y = \beta^* - 1 + y_{t}(1-\log y_{t}), \label{eq:implicit_recursion}
\end{align}
for $t=0,\ldots,k-1$. Note that the right-hand side of this expression corresponds to the negative of the ODE right-hand side if we replace $w(t)$ by $y_t$. From here, we can extract a sequence defined as
\begin{align*}
	x_0 & = 0\\
	x_{t+1} & = \max \left\{ 0, x_t - \frac{1}{k}\left((\beta^* - 1) + x_t(1- \log x_t ) \right) \right\}.
\end{align*}
This sequence can be interpreted as the output of Euler's method applied to the ODE with step-size $1/k$. We can then show the following lemma. % that bounds $y_t$ in terms of $x_t$.

\begin{lemma}\label{lem:y_t_bounded_by_x_t}
	For any $t=0,1,\ldots,k$,  $x_t \leq y_t \leq x_t + 4 \log (32 k)/k$.
\end{lemma}

For the proof of this lemma, we define an additional sequence $z_t$ that can also be interpreted the output of Euler's method with a slightly different step size. We show that $x_t \leq y_t \leq z_t \leq x_t + 4 \log (32 k)/k$, from which the result follows; the details are deferred to Appendix~\ref{app:approx_guarantee_infinite_model}.

Let $w$ be the solution of the ODE; since $\beta^* \geq \overline{\beta}$, $w(\rho^*)=0$ for some $\rho^*\in [0,1]$. We extend $w$ to $[0,1]$ by setting $w( t ) = 0$ for $ t \in [\rho^*,1]$, and keep denoting this extension by $w$.

\begin{lemma}\label{lem:x_t_bounded_by_w}
	For any $t=0,1,\ldots,k$,  $x_t \leq w(t/k)$.
\end{lemma}

The proof of this lemma is by induction in $t$, and is similar to the proof of Lemma 7 in~\cite{correa2021posted}. However, there are some technical details that must be addressed in our case, since we are utilizing a different sequence of points to approximate $w(t/k)$. We defer the proof to Appendix~\ref{app:approx_guarantee_infinite_model}.

Before presenting the bound over $\beta^*$, we need to quantify $\rho^*$. Using Equation~\eqref{eq:implicit_recursion}, we can show that $y_{k-\ell} \geq \ell /(32k)$ (see Proposition~\ref{prop:lower_bound_y_k_ell} in Appendix~\ref{app:approx_guarantee_infinite_model}). This result combined with Lemmas~\ref{lem:y_t_bounded_by_x_t} and \ref{lem:x_t_bounded_by_w} show that $\rho^* \geq 1 - 256 \log(32 k)/k$.

Now, let $\rho:[0,1] \to [0,\rho^*]$ be the inverse function of $w$, the solution of the ODE. We show in the Appendix that $w$ is strictly decreasing and at least twice differentiable, which implies that $\rho$ is also strictly decreasing and differentiable. Then,
\begin{align*}
	\rho(1) - \rho(0) & = \int_0^1 \rho'(w)\, \mathrm{d}w = \int_0^1 (w'(\rho(w)))^{-1} \, \mathrm{d}w  = \int_0^1 \frac{-\, \mathrm{d}w}{(\beta^* - 1) + w(1-\log w)}.
\end{align*}
Thus,
\[
\int_0^1 \frac{\mathrm{d}w}{\beta^* -1 + w (1-\log w)} = \rho^* \geq 1 - c,
\]
where $c= 256 \log (32k)/k$.

Let $I(\beta)= \int_0^1 (\beta-1 +w (1-\log w))^{-1} \, \mathrm{d}w$. Function $I$ is strictly decreasing and for $\overline{\beta}\approx 1.341$ we have $I(\overline{\beta})=1$; see, e.g.\ \cite{kertz1986stop,correa2021posted}.

\begin{lemma}
	We have $\beta^* \leq \overline{\beta} (1+2 c)$.
\end{lemma}

\begin{proof}
By contradiction, assume that $\beta^* > \overline{\beta}(1+2 c)$. Now, for any $w\in [0,1]$, we have
\[
(1+2 c) \overline{\beta} - 1 + w(1-\log w) \geq (1+2c)\left(  \overline{\beta} - 1 + w(1-\log w)  \right).
\]
Thus, rearranging this last inequality and integrating over $w\in [0,1]$, we obtain
\[
I((1+2c)\overline{\beta}) \leq \frac{1}{1+2c} I(\overline{\beta}) = \frac{1}{1+2c}< 1-c \leq \rho^* = I(\beta^*).
\]
Using the monotonicity of $I$, $\beta^* \leq (1+ 2c) \overline{\beta}$, which is a contradiction.
\end{proof}

This result implies
\[
v_{\infty,k}^* = \frac{1}{\beta^*} \geq \overline{\gamma}( 1 - 2 c) = \overline{\gamma} \left( 1 - 512 \frac{\log (32k)}{k}  \right) ,
\]
which finishes the proof of the lower bound in Theorem~\ref{thm:bound_infnite_model}. The proof of the upper bound is similar; the details are in Appendix~\ref{app:upperbound_v}.

\section{Prophet Inequalities for a Small Number of Thresholds}\label{sec:small_thresholds}

In this section, we provide numerical results when the number of thresholds is small. For $k=2$, we obtain the optimal value $\gamma_{n,2}^* \approx 0.708$ using the framework of Theorem~\ref{thm:optimal_cr}. For $k\geq3$, we utilize $(CLP)_{n,k}$ for $n$ tending to infinity to provide asymptotic lower bounds for $\gamma_{n,k}^*$. In Table~\ref{table:values_v}, we summarize the values computed in this section.

\subsection{Two Thresholds}

\modif{We use Theorem~\ref{thm:optimal_cr} and~\ref{prop:dual_P_nk} to find that $\gamma_{n,2}^*\approx0.70804$ when $n \to \infty$ as follows. We construct parametrized primal and dual solutions; by letting $n$ go to infinity, we find simpler expressions for the parameters while both primal and dual solutions have matching values. Optimizing over the parameters gives us the desired result.}

\subsubsection{The primal solution}

\modif{First, we provide the description of the primal solution. Let $0<a_1<a_2<n$, and $\tau_1=\theta n$ and $\tau_2=(1-\theta)n$. Now, take $\alpha_{1,q}=\delta_{\{a_1/n\}}(q)$ and $\alpha_{2,q}=(1-a_1/n)^{\tau_1} \delta_{\{ a_1/n \}}(q)$, where $\delta_{\{u_0\}}$ is the dirac function with mass in $u_0$. By construction, $\alpha_1$ and $\alpha_2$ satisfy Constraints~\eqref{const:constr_1_P_nk}-\eqref{const:constr_t_P_nk}. Consider $v=\left(1-\left(1-a_1/n \right)^{\tau_1}\right)+\left(1-a_1/n \right)^{\tau_1}\left(1-\left(1-a_2/n\right)^{\tau_2}\right)=1-\left(1-a_1/n\right)^{\tau_1}\left(1-a_2/n\right)^{\tau_2}$ and the function
\[
\eta_u = \begin{cases}
	u \left(1-\left(1-\frac{a_1}{n}\right)^{\tau_1}\right) \frac{n}{a_1} + u \left(1-\frac{a_1}{n}\right)^{\tau_1}\left(1-\left(1-\frac{a_2}{n}\right)^{\tau_2}\right)\frac{n}{a_2} - v (1-(1-u)^n), & u \leq \frac{a_1}{n},\\
	\left(1-\left(1-\frac{a_1}{n}\right)^{\tau_1}\right) + u \left(1-\frac{a_1}{n}\right)^{\tau_1}\left(1-\left(1-\frac{a_2}{n}\right)^{\tau_2}\right)\frac{n}{a_2} - v (1-(1-u)^n), & \frac{a_1}{n} <u \leq  \frac{a_2}{n},\\
	v(1-u)^n, & u > \frac{a_2}{n}.
\end{cases}
\]
Note that $\eta_0=0$ and $\eta_1=0$. It is straightforward to verify that $(\alpha,\eta,v)$ satisfies Constraint~\eqref{const:P_nk_esp_max_value_const}. Thus, in order to have $(\alpha,\eta,v)$ feasible for $(P)_{n,2}$, we need to ensure that $\eta_u \geq 0$ for all $u$. Due to the form of $\eta$, given $\theta\in (0,1)$ and $a_1<a_2$, we only need to verify $\eta_u \geq 0$ for $u\leq a_2/n$. Note that $\eta_u$ is convex in each interval $[0,a_1/n]$ and $[a_1/n,a_2/n]$; moreover, $\eta_{a_2/n}=v(1-a_2/n)^n \geq 0$. This implies that there are only 3 points of interest in order to verify the nonnegativity of $\eta$. Let $u_1/n\in \R$ be the global minimum of the function that defines $\eta_u$ in the interval $[0,a_1/n]$ and let $u_2/n\in \R$ be the global optimal point of the function defining $\eta_u$ in $(a_1/n,a_2/n]$, respectively. We arbitrarily set $u_1=0$ and search for solutions such that $u_2 \in (a_1/n,a_2/n)$ and $\eta_{u_2/n}=0$. This is enough to guarantee $\eta_u\geq 0$. These choices may seem arbitrary, but they are a byproduct of complementary slackness and will become clearer when we present the dual solution below.}

\modif{To continue the analysis, we let $n\to \infty$ and  numerically compute values $\theta$, $a_1$ and $a_2$ for which $\eta \geq 0$ and such that $v$ is as large as possible in the limit. Note that for finite $n$, reducing the value $v$ by a small amount $\varepsilon>0$ does not affect the feasibility of $(\alpha,\eta,v)$ in $(P)_{n,k}$ for $n$ large. Thus, when $n$ tends to infinity, the limit value of $(P)_{n,k}$ is at least the limit of $v$, which we can estimate numerically.
}

\modif{By letting $n$ tend to infinity, we obtain a limit version of $\eta$ given by
\[
\bar{\eta}_u = \begin{cases}
	\frac{u}{a_1} \left( 1- e^{-a_1 \theta} \right) + \frac{u}{a_2} e^{-a_1 \theta}\left( 1-e^{-a_2(1-\theta)} \right) - \overline{v} (1-e^{-u}) , & u \leq a_1, \\
	1-e^{-a_1 \theta} + \frac{u}{a_2} e^{-a_1 \theta}\left( 1- e^{-a_2 (1-\theta)} \right) - \overline{v}(1-e^{-u}),  & a_1 < u \leq a_2, \\
	\overline{v} e^{-u}, & u > a_2.
\end{cases}
 \]
Note that $\overline{v}$, the limit of $v$, equals $\overline{v}=\lim_{n\to \infty} 1-(1-a_1/n)^{\tau_1}(1-a_2/n)^{\tau_2}=1-e^{-a_1\theta -a_2(1-\theta)}$, where we used $\tau_1=\theta n$ and $\tau_2=(1-\theta)n$. Following the logic applied to $\eta$ in the finite $n$ case, we compute the points
\[
u_1 = - \log \left(   \frac{(1-e^{-a_1\theta})/a_1 + e^{-a_1 \theta} (1-e^{-a_2 (1-\theta)})/a_2 }{1-e^{-a_1 \theta - a_2 (1-\theta)}}  \right), \quad u_2 = - \log \left(   \frac{ e^{-a_1 \theta} (1-e^{-a_2 (1-\theta)})/a_2 }{1-e^{-a_1 \theta - a_2 (1-\theta)}}  \right) ,
\]
which are the global solutions to the function defining $\overline{\eta}$ in $[0,a_1]$ and $(a_1,a_2]$, respectively. Recall that we are forcing $u_1=0$. Because $u_1=0$ is the minimum of a strictly convex function, $u_2$ is as well, and $\bar{\eta}_{u_2}=0$; we  formulate the  system
\begin{align*}
	1-e^{-a_1 \theta - a_2(1-\theta)} & = \left(\frac{1-e^{-a_1 \theta}}{a_1}\right) + e^{-a_1 \theta}\left( \frac{1-e^{-a_2(1-\theta)}}{a_2} \right) \\
	(1-e^{-a_1\theta - a_2(1-\theta)})(1-e^{-u_2}) & = 1-e^{-a_1 \theta} + u_2 e^{-a_1 \theta}\left( \frac{1-e^{-a_2(1-\theta)}}{a_2} \right) \\
	(1-e^{-a_1\theta - a_2(1-\theta)})e^{-u_2} &= e^{-a_1\theta}\left( \frac{1-e^{-a_2(1-\theta)}}{a_2}\right) .
\end{align*}
Using this system, we can find expressions for $a_1$ and $a_2$ as a function of $u_2$ that guarantee $\bar{\eta}_u\geq 0$ for all $u\geq 0$, and an implicit formula that relates $\theta$ with $u_2$. We summarize this in the following proposition; we skip the proof for brevity.}

\modif{\begin{proposition}\label{prop:formulas_for_a1_a2}
		For $u_2\geq 0$, we have $a_1=1-u_2 e^{-u_2}/(1-e^{-u_2})$ and $a_2=u_2+1$. We also have
		\[
		-e^{-u_2} - u_2 e^{-u_2} = e^{-a_1\theta - a_2 (1-\theta)}(1-e^{-u_2}-u_2 e^{-u_2}) - e^{-a_2\theta}.
		\]
\end{proposition}}

\modif{Then, recalling that $\bar{v}=1-e^{-a_1\theta - a_2(1-\theta)}$, we can obtain the optimal $\theta$ and $u_2$ (and so $a_1$ and $a_2$ via the previous proposition) by solving
\[
\max_{\substack{\theta\in [0,1]\\u\geq 0}} \left\{  1- e^{-a_1\theta-a_2(1-\theta)}:  -e^{-u} - u e^{-u} = e^{-a_1\theta - a_2 (1-\theta)}(1-e^{-u}-u e^{-u}) - e^{-a_2\theta}  \right\} ,
\]
where we replace $a_1$ and $a_2$ by the formulas in Proposition~\ref{prop:formulas_for_a1_a2}. Numerically, we find $u=u_2=1.316097$ and $\theta=0.603285$. This gives us $a_1=0.517708$ and $a_2=2.316097$ with a value $\bar{v} \approx 0.70804$.
}

\subsubsection{The dual solution}

\modif{We now provide a dual solution with value matching the value computed above. We keep denoting by $a_1<u_2<a_2$, $\theta\in (0,1)$ and $\overline{v}$ the optimal values computed above. We provide a solution to $(D)_{n,2}$ with objective value converging to $\overline{v}$. For $a > 0$ and $b\geq c$, we define
\[
f_q = \frac{a}{n} \delta_{\{0\}} + b \mathbf{1}_{(0,u_2/n]}(q) + c \mathbf{1}_{(u_2/n,1)}(q).
\]
Note that $f$ is nonincreasing in $[0,1]$, and its choice is due to complementary slackness. Since $\eta_u$ is strictly positive in $(0,u_2/n)\cup (u_2/n,1)$, the only places where $f$ can change values are $0$, $u_2/n$ and $1$.  Let
\begin{align*}
	d_1 &= \sup_{q\in [0,1]} \left\{  \left( \frac{1-(1-q)^{\tau_1}}{q}\right)  \int_0^q f_u \, \mathrm{d}u +(1-q)^{\tau_1} d_2 \right\}\\
	d_2 &= \sup_{q\in [0,1]} \left\{  \left( \frac{1-(1-q)^{\tau_2}}{q}\right)  \int_0^q f_u \, \mathrm{d}u  \right\} .
\end{align*}
The pair $(\mathbf{d},f)$ satisfy Constraints~\eqref{const:dynamic_constr} and~\eqref{const:nonincreasing} by design. We only need to ensure that $f$ satisfies \eqref{const:max_value_const}; that is,
\[
1= a + b\left(1-\left(1-\frac{u_2}{n}\right)^n\right)+ c\left(1-\frac{u_2}{n}\right)^n.
\]
This imposes restrictions on $a,b$ and $c$. Using complementary slackness, we can argue that $d_1$ must attain its unique maximum at $a_1/n$, while $d_2$ must attain its unique maximum at $a_2$. This is because $\alpha_{1,q}$ is positive only for $q=a_1/n$, while $\alpha_{2,q}$ is positive only for $q=a_2/n$. These two conditions imply that the derivatives of the functions inside the supremum operator defining $d_1$ and $d_2$ have to have vanishing derivative in $a_1/n$ and $a_2/n$, respectively. Before continuing, from now on, we let $n$ tend to infinity to make the derivations easier. Similar to the primal case, from a solution to the limit model with $n\to \infty$, we can recover solutions for finite and large $n$ with arbitrarily small error.
}

\modif{When $n\to \infty$, we obtain %the following system
\begin{align*}
	1& = a + b (1-e^{-u_2}) + c e^{-u_2}\\
	0&= b\left(\frac{1-e^{-\theta a_1}}{a_1}\right) + (a+b a_1)\left( \frac{\theta e^{-\theta a_1} a_1 - (1-e^{-\theta a_1})}{a_1^2} \right) \\
	&\quad - \theta e^{-\theta a_1} \left(\frac{1-e^{-(1-\theta)a_2}}{a_2}\right)(a+b u_2 + c(a_2-u_2)) \\
	0& = c\left( \frac{1-e^{-(1-\theta)a_2}}{a_2} \right) +(a+b u_2 + c(a_2-u_2)) \left( \frac{(1-\theta)e^{-(1-\theta)a_2}a_2 - (1-e^{-(1-\theta)a_2}) }{a_2^2} \right) .
\end{align*}
The first equality follows by taking the limit \eqref{const:max_value_const}, the second and third by the optimality of $a_1/n$ and $a_2/n$ in $d_1$ and $d_2$ respectively, when $n\to \infty$. This is a 3-by-3 system with unknowns $a$, $b$ and $c$. Using the information from Proposition~\ref{prop:formulas_for_a1_a2}, we obtain formulas for $a$, $b$ and $c$ in terms of $u_2$ and $\theta$. We defer the formulas to Appendix~\ref{app:formulas_a_b_c} due to their extensive length. Using the values computed in the previous part, we obtain $a=0.516213$, $b=0.567355$ and $c=0.255744$, which then implies $d_1\approx 0.70804$. }

\modif{\begin{remark}
		Another approach is to restrict $(D)_{n,k}$ (with $n\to \infty$) to functions of the form $f_q=a \delta_{\{0\}} + b \mathbf{1}_{(0,u_2]}(q)+c\mathbf{1}_{(u_2,\infty)}(q)$ with $b\geq c\geq 0$ and $a\geq 0$. This is a semi-infinite linear program with finitely many variables ($d_1,d_2,a,b,c$) and uncountably many constraints. Using the transformation $y=e^{-q}$ with $q\in \R_+$, we get a semi-infinite linear program with infinitely many constraints indexed by $[0,1]$. Using a discretization of $[0,1]$, we can solve this semi-infinite linear program to arbitrary accuracy. A similar approach was used in~\cite{li2022query}, obtaining a slightly weaker bound.
\end{remark}}

\subsection{Prophet Inequalities for Up to Ten Thresholds}

In this subsection, we utilize $v_{\infty,k}^*$ to provide prophet inequalities for $k\geq 3$. Since we have a closed-form expression for the optimal solution of $(CLP)_{\infty,k}$, we can compute $v_{\infty,k}^*$ using a binary search. For guessed values $v\leq v_{\infty,k}^*$ we should be able to find feasible solutions of $(CLP)_{\infty,k}$ with objective value $v$, while for $v>v_{\infty,k}^*$, we should fail at this task. Now, for $0< v\leq v_{\infty,k}^*$, we can construct feasible solutions of $(CLP)_{\infty,k}$ by using the functions
\[
\bar{\omega}_{t,y} = - v \frac{y\log y}{1-y^{1/k}} \mathbf{1}_{(y_{t},y_{t-1})}(y),
\]
where $1=y_0\geq y_1\geq \cdots \geq y_{k-1}\geq y_k=0$. The values $y_1,\ldots,y_{k-1}$ are implicitly defined as a function of $v$. Thus we use auxiliary optimization problems to find them. Let $H(x)=\int_0^x -(1-y^{1/k})^{-1}\log y  \, \mathrm{d}y$. Then,
$$H(y_{t-1})-H(y_t) = \frac{1}{v} \int_0^1 \frac{\bar{\omega}_{t,y}}{y} \, \mathrm{d}y.$$

For $v\leq v_{\infty,k}^*$, we can rewrite the constraints satisfied by $(\bar{\omega}_{t,\cdot})_{t=1}^k$ as a function of $H$,
\begin{align}
	\frac{1}{v} - H(y_0) + H(y_1) & \geq 0 \label{eq:system_y_1} \\
	H(y_0) - 2H(y_1) + H(y_2) + y_0 (\log y_0 - 1) - y_1 (\log y_1 -1) & \geq 0 \label{eq:system_y_2} \\
	                                                                   & \,\,\,\vdots \nonumber \\
	H(y_{k-3}) - 2H(y_{k-2}) + H(y_{k-1}) + y_{k-3}(\log y_{k-3} -1 ) - y_{k-2}(\log y_{k-2}-1) & \geq 0 \label{eq:system_y_k_1} \\
	H(y_{k-2}) - 2H(y_{k-1}) + H(y_k) + y_{k-2}(\log y_{k-2} -1 ) - y_{k-1}(\log y_{k-1}-1) & \geq 0 , \label{eq:system_y_k}
\end{align}
where $y_0=1$ and $y_{k}=1$. We compute $y_1$ by solving
\[
y_1 = \min\{ y\in [0,1]:1/v -H(y_0) +H(y) \geq 0 \} .
\]
Once $y_0,y_1,\ldots,y_{t-1}$ have been determined, we find $y_t$ by solving
\[
y_t = \min\{y \in[0,y_{t-1}]: H(y_{t-2}) -2H(y_{t-1}) +H(y) +y_{t-2}(\log y_{t-2}-1) - y_{t-1}(\log y_{t-1}-1) \geq 0 \}.
\]
Since $H$ is monotone, we can guarantee that $y_1,\ldots,y_{k-1}$ defined above satisfy Inequalities~\eqref{eq:system_y_1}-\eqref{eq:system_y_k}, if and only if $v\leq v_{\infty,k}^*$. Thus, we have the following result.

\begin{proposition}
	If $v=v_{\infty,k}^*$ and $y_1,\ldots,y_{k-1}$ are defined as above, $y_1\geq y_2 \geq \cdots \geq  y_{k-1}$ are the optimal thresholds defining the optimal solution $\omega$ of $(CLP)_{\infty,k}$.
\end{proposition}

We numerically compute the values $y_1,\ldots,y_{k-1}$. Given a equidistant discretization of the interval $[0,1]$ with step size $1/\ell$, we can solve each problem in time $\mathcal{O}(\log \ell)$ via bisection. Therefore, in $\mathcal{O}(k \log \ell)$ time we can compute the $k-1$ values $y_1,\ldots,y_{k-1}$ that by construction satisfy Constraints~\eqref{eq:system_y_1}-\eqref{eq:system_y_k_1}. If they also satisfy  Constraint~\eqref{eq:system_y_k} then we have successfully constructed a feasible solution. We run a bisection in $v \leq v_{\infty,k}^*$ to obtain a solution that is at most $\delta$ far away from the optimal solution, in $\mathcal{O}(\log 1/\delta)$ iterations of the subroutine.

Table~\ref{table:values_v} details the values obtained for $k=3,\ldots,10$. We numerically computed these values for $\ell=10^{12}$ and $\delta=10^{-8}$. We skip the description of $y_1,\ldots,y_{k-1}$ for conciseness.

\begin{table}
	\small
	\begin{center}
	\begin{tabular}{|c|c|}
	\hline
	$~k~$& Approximation\\
	\hline
	1 & \modif{$1-1/e$}\\
	2 & \modif{0.7080} \\
	3  & 0.7233 \\
	4 & 0.7321 \\
	5 & 0.7364 \\
	6 & 0.7389 \\
	7 & 0.7405 \\
	8 & 0.7416 \\
	9 & 0.7423 \\
	10 & 0.7428 \\
	$k$ large & $(1-\Theta(\log k / k))\bar{\gamma}$\\
$\infty$ & $\bar{\gamma}\approx 0.745$\\
	\hline
\end{tabular}
\end{center}
\caption{Numerical values of $v_{\infty,k}^*$ computed with the aforementioned procedure, where $k$ is the number of threshold used. \modif{For $k=1,2$ we utilized the values computed via $\gamma_{n,k}^*$.}}
\label{table:values_v}
\end{table}

%% file: conclusions.tex
\section{Final Remarks}

In this work, we examined a version of the classical i.i.d.\ prophet inequality problem where at most $k$ different thresholds are selected. We presented an infinite-dimensional linear program that exactly models a $k$-dynamic policy and has its approximation ratio as objective value. This linear program recovers the known guarantee $\gamma_{n,1}^*=1-(1-1/n)^n$ for a single threshold, and provides the new result $\gamma_{n,2}^* \approx 0.708$ for two thresholds. For larger values of $k$, we utilized a relaxation of this infinite linear program to  provide a lower bound over the optimal approximation factor that a $k$-dynamic policy can attain. We also characterized the optimal solutions using duality. Since the optimal value of this linear program, $v_{n,k}^*$, is hard to analyze for finite $n$, we provided an asymptotic analysis in $n$. We showed that for $k$ large, $v_{\infty,k}^* =\lim_n v_{n,k}^* = \bar{\gamma}\left( 1 - \Theta(\log k/ k) \right)$.

Surprisingly, even for small values of $k$ we already obtain significant improvements over the best static solution; for example, just considering $k=2$ thresholds yields an $11\%$ improvement. Furthermore, for $k=5$ thresholds our approximation is roughly $1.2\%$ away from the optimal approximation attainable by a fully dynamic solution. We hope this motivates the study of policies that constrain the number of decisions over time in different decision-making settings, including multi-unit/cardinality constraints, combinatorial constraints, and in other models, such as secretary problems.

As we saw in Section~\ref{sec:asymptotic_analysis}, our results exhibit a gap between $v_{n,k}^*$ and $\gamma_{n,k}^*$ that stems from dropping Constraint~\eqref{const:nonincreasing} and assuming that the windows are of fixed size. However, we also show that this gap vanishes as $k$ grows. The first part of the gap is the most influential one, as even when $k=1$, we already have $v_{\infty,k}^*=6/\pi^2$, while $\gamma_{\infty,k}^*=1-1/e$; see Corollary~\ref{cor:value_k_1}. %\modif{Numerically, we verify that this gap disappears rapidly for $k\geq 2$. For $k=2$, for instance, we get $\gamma_{\infty,2}^*\approx 0.708$ while $v_{\infty,2}^*\approx0.704$.} % ATH: You already say exactly this in the first paragraph.
By considering a formulation with extended indices (i.e.\ quantile $q$, index of the interval $t$, length of the interval $\tau_t$), we can eliminate the second part of the gap. This creates a more convoluted formulation from which we can recover algorithms in the same manner as in this work. However, the practical benefits of this may be minimal given how quickly the guarantees approach $\bar{\gamma}$ as $k$ grows. %as with $k=5$ equidistant thresholds, we already get a  solution that is close to $\bar{\gamma}$. %\ps{Check this.}
Furthermore, as we show in Appendix \ref{app:threshold_k_2}, for $k=2$ the solution that sets the interval lengths to $n/2$ guarantees a competitive ratio of $0.701$, while optimizing over the length of the intervals only improves the competitive ratio to $0.704$. %This is just a mild $0.4\%$ improvement over the equidistant solution (see Figure~\ref{fig:graphk2} in Section~\ref{sec:small_thresholds}).

%% file: Appendix.tex
\small

\section{Missing Details in the Proof of Theorem~\ref{thm:optimal_cr}}\label{app:proof_thm_1}

\modif{In this section, we present the missing details in the proof of Theorem~\ref{thm:optimal_cr} that characterizes the optimal competitive ratio for $k$-dynamic algorithms. We prove the following:
\begin{enumerate}
	\item First, we fix the sizes of windows to $\tau_1,\ldots,\tau_k$ and we show that for any continuous CDF $F$ such that $\E[\max_{i\in [n]} X_i] = 1$ where $X_i$ has CDF $F$, we have $\E[X_A] \geq \gamma_{n,k}^*(\tau_1,\ldots,\tau_k)$, where $A$ is the algorithm that implements the optimal stopping rule with windows of size $\tau_1,\ldots,\tau_k$.
	\item Second, with the sizes of windows still fixed to $\tau_1,\ldots,\tau_k$, we take any feasible solution to problem $(D)$, say $(\mathbf{d},f)$, we show that $d_1\geq \E[X_A]$ using backward induction in $t\in [k]$, where $X_i$ follows the distribution with CDF $F$ given by $F^{-1}(u)=f_{1-u}$. With this, we can conclude that the optimal competitive ratio for algorithms with windows of size $\tau_1,\ldots,\tau_k$ is exactly $\gamma_{n,k}^*(\tau_1,\ldots,\tau_k)$.
%	\item Finally, we conclude that $\gamma_{n,k}^*$, the optimal competitive ratio for $k$-dynamic algorithms is exactly $\sup_{\substack{\tau_1,\ldots,\tau_k\in \Z_+\\ \tau_1+\cdots+\tau_k=n}} \gamma_{n,k}^*(\tau_1,\ldots,\tau_k)$.
\end{enumerate}
For the first claim, we define $F^{-1}(u)=\inf\{ x : F(x) = u  \}$, which is well-defined since $F$ is continuous. Then, $f_{u}=F^{-1}(1-u)$ is non-increasing and satisfies Constraint~\eqref{const:nonincreasing}. Moreover,
\[
1=\E[\max_i X_i] = \int_0^1 n(1-u)^{n-1}F^{-1}(1-u)\, \mathrm{d}u = \int_0^1 n(1-u)^{n-1} f_u \, \mathrm{d}u ,
\]
as shown in Proposition~\ref{prop:approximation_relaxed_LP} in Section~\ref{sec:prophet_bounded_threshold}. Let $d_t$ be the expected value obtained from $A$ starting in the $t$-th window. Then, $d_1=\E[X_A]$. Moreover,
\begin{align*}
	d_k&=\sup_{x\geq 0} \left\{ \left(1-\Prob(X \leq x)^{\tau_k} \right) \E[X \mid X\geq x]   \right\}\\
	& = \sup_{q\in [0,1]} \left\{  \left( \frac{1-(1-q)^{\tau_k}}{q}\right) \int_0^q F^{-1}(1-u)\, \mathrm{d}u \right\}.
\end{align*}
The first line follows because $A$ is optimal for $F$, and the second line holds by doing the change of variable $q=\Prob(X\geq x)$. A similar argument shows that
\[
d_t = \sup_{q\in [0,1]} \left\{  \left( \frac{1-(1-q)^{\tau_t}}{q}\right) \int_0^q F^{-1}(1-u)\, \mathrm{d}u  + (1-q)^{\tau_t} d_{t+1}\right\},
\]
for $t=1,\ldots,k-1$. From here, we obtain that $(\mathbf{d},f)$ is a feasible solution to $(D)$, which shows that $\E[X_A]=d_1\geq \gamma_{n,k}^*(\tau_1,\ldots,\tau_k)$.}

\modif{For the second claim, let $(\mathbf{d},f)$ be a feasible solution of $(D)$. Notice that $f$ has to be $L_1$-integrable, thus for any $\varepsilon>0$ it can be approximated by a nonnegative, nondecreasing continuous function $\widehat{f}$ such that $\int_0^1 |f_u - \widehat{f}_u|\, \mathrm{d}u \leq \varepsilon$. Without loss of generality, we can assume that $\widehat{f}$ also satisfies Constraint~\eqref{const:max_value_const}, since we can renormalize the function; this does not affect the monotonicity and nonnegativity of $\widehat{f}$, while the error $\int_0^1 |f_u - \widehat{f}_u|\, \mathrm{d}u$ worsens from $\varepsilon$ to at most $3 \varepsilon$, for small $\varepsilon$. Moreover, we can assume that $\widehat{f}$ is strictly decreasing. Indeed, consider $g_u=\widehat{f}_u + \varepsilon(1-u)$. Then, $g_u$ is nonnegative, continuous, strictly decreasing and a good approximation of $f_u$ in the $L_1$-norm. To avoid notational clutter, we assume $f$ already has these properties. Otherwise, the results remains true up to an error of the order $nk \varepsilon$, which can be made arbitrarily small.}

\modif{We now construct a distribution $\mathcal{D}$ from $f$ as follows. We define $X=f_{Q}$, where $Q\sim \mathrm{Unif}[0,1]$ is a uniform $[0,1]$ random variable. Since $f$ is strictly decreasing, then $f^{-1}$ is well-defined. Moreover, if $F$ denotes the CDF of $X$, then $F$ is continous and it can be shown that $F^{-1}(u) = f_{1-u}$. Thus,
\[
\E[X\mathbf{1}_{\{ X \geq x \}}] =\int_{x}^\infty y \, \mathrm{d}F(y) = \int_0^{q} f_{u} \, \mathrm{d}u ,
\]
where $q=\Prob(X\geq x)$. Now, consider the optimal algorithm $A$ that solves the dynamic program in each window $t$,
\[
\widehat{d}_t=\max_{x\geq 0} \left\{  \left( 1-\Prob(X\leq x)^{\tau_t} \right)\E[X\mid X\geq x]  + \Prob(X\leq x)^{\tau_t}d_{t+1} \right\} ,
\]
with $\widehat{d}_{k+1}=0$; then $\E[X_A]=\widehat{d}_1$. Notice that $A$ is independent of $A$ and only receives $F$ as an input. We now show via backward induction in $t=k,\ldots,1$ that $d_t\geq \widehat{d}_t$. Indeed, for $t=k$, we have
\begin{align*}
	\widehat{d}_{k} &= \sup_{x\geq 0} \left\{ \left( 1-\Prob(X\leq x)^{\tau_k} \right) \E[X\mid X\geq x]  \right\} \tag{Optimality of $A$} \\
	& = \sup_{q\in [0,1]} \left\{ \left(\frac{1-(1-q)^{\tau_k}}{q}\right) \int_0^q f_u \, \mathrm{d}u  \right\} \tag{Change of variable $q=\Prob(X\geq x)$}\\
	& \leq d_k,
\end{align*}
where in the last inequality we used Constraint~\eqref{const:dynamic_constr} for $t=k$. Assuming the result true for $t+1$, we can proceed similarly for $t$. Indeed,
\begin{align*}
	\widehat{d}_t &= \sup_{x\geq 0} \left\{ \left( 1-\Prob(X\leq x)^{\tau_t} \right) \E[X\mid X\geq x] + \Prob(X\leq x)^{\tau_t} \widehat{d}_{t+1}  \right\} \tag{Optimality of $A$}\\
	&\leq \sup_{x\geq 0} \left\{ \left( 1-\Prob(X\leq x)^{\tau_t} \right) \E[X\mid X\geq x] + \Prob(X\leq x)^{\tau_t}  d_{t+1} \right\} \tag{Induction}\\
	& \leq d_t ,
\end{align*}
where again we used Constraint~\eqref{const:dynamic_constr} in the last line. This shows  that $d_1 \geq \E[X_A]$. Since $OPT=\E[\max_i X_i] = \int_0^1 n (1-u)^{n-1} f_u \, \mathrm{d}u = 1$, then, the competitive ratio of $A$ is at most $d_1$. This remains true for any feasible $(\mathrm{d},f)$. Thus, minimizing over $(\mathrm{d},f)$, we conclude that the competitive ratio of $A$ is at most $\gamma_{n,k}^*(\tau_1,\ldots,\tau_k)$.
}

\section{Missing Details in the Proof of Theorem~\ref{prop:dual_P_nk}}\label{app:strong_dual}

First, we present the weak duality result between $(D)_{n,k}$ and $(P)_{n,k}$. We later present the strong duality.

\subsection{Weak Duality Between $(D)_{n,k}$ and $(P)_{n,k}$}
\modif{ For weak duality, we verify it only for solutions $(\mathbf{d},f)$ of $(D)_{n,k}$ with differentiable $f$, because $f$ has to be $L_1$-integrable and can be approximated by smooth functions to within an arbitrarily small error. Taking any $(\alpha,v,\eta)$ feasible for $(P)_{n,k}$, multiplying constraint~\eqref{const:P_nk_esp_max_value_const} with $f_u$ and integrating, we obtain
	\begin{align*}
		v \int_0^1 n(1-u)^{n-1} f_u \, \mathrm{d}u+ \int_0^1 f_u \frac{\mathrm{d}\eta_u}{\mathrm{d}u} \, \mathrm{d}u& \leq \int_0^1 \sum_{t=1}^k \int_u^1 \left( \frac{1-(1-q)^{\tau_t}}{q} \right) \alpha_{t,q} \, \mathrm{d}q  f_u \, \mathrm{d}u \\
		& = \int_0^1 \sum_{t=1}^k \alpha_{t,q} \left( \frac{1-(1-q)^{\tau_t}}{q} \right)\int_0^q f_u \, \mathrm{d}u \, \mathrm{d}q \tag{Change order of integration since everything is nonnegative}\\
		& \leq \int_0^1 \sum_{t=1}^{k-1} \alpha_{t,q}\left( d_t - (1-q)^{\tau_t}d_{t+1} \right)\, \mathrm{d}q +\int_0^1 \alpha_{k,q} d_k\, \mathrm{d}q \tag{By using Constraint~\eqref{const:dynamic_constr}}\\
		& =  d_1 \int_{0}^1 \alpha_{1,q}\, \mathrm{d}q +\sum_{t=2}^{k} d_t\left( \int_0^1 \alpha_{t,q} \, \mathrm{d}q - \int_0^1 (1-q)^{\tau_t}\alpha_{t-1,q} \, \mathrm{d}q  \right)\\
		& \leq d_1 . \tag{By using Constraints~\eqref{const:constr_1_P_nk} and~\eqref{const:constr_t_P_nk}}
	\end{align*}
	The conclusion now follows by showing that $\int_0^1 f_u (\mathrm{d}\eta_u/\mathrm{d}u)\, \mathrm{d}u \geq 0$. Indeed, intergrating by parts,
	\[
	\int_0^1 f_u \frac{\mathrm{d} \eta_u}{\mathrm{d}u} \, \mathrm{d}u = f_u \eta_u \bigg|_{u=0}^{u=1} - \int_0^1 \eta_u \frac{\mathrm{d} f_u}{\mathrm{d}u}\, \mathrm{d}u  =  -\int_0^1 \eta_u \frac{\mathrm{d} f_u}{\mathrm{d}u}\, \mathrm{d}u \geq 0,
	\]
	where in the second equality we used the initial condition of $\eta$ given by Constraints~\eqref{const:P_nk_initial_condition} and the nonincreasing property of $f$ given by Constraint~\eqref{const:nonincreasing}, which implies that the derivative of $f$ is nonpositive.}

\subsection{Strong Duality Between $(D)_{n,k}$ and $(P)_{n,k}$}

%\section{Missing Proof from Section~\ref{sec:prophet_bounded_threshold}}\label{app:prophet_bounded_threshold}

\modif{We now show strong duality in Theorem~\ref{prop:dual_P_nk}. We do this by approximating both problems with a discretization of the $[0,1]$ interval. Let $\lambda_{n,k}^*=\lambda_{n,k}^*(\tau_1,\ldots,\tau_k)$ be the optimal value of $(P)_{n,k}$ in Theorem~\ref{prop:dual_P_nk}. We aim to show $\gamma_{n,k}^* = \lambda_{n,k}^*$ to conclude strong duality. For a large integer $m$, consider the following finite-dimensional linear minimization problem,
	\begin{subequations}\label{form:D_nkm}\small
		\begin{align}
			\min_{\substack{\mathbf{d}\in \R_+^k \\ \mathbf{f}\in \R^{m+1}_+}} \quad & d_1 \label{eq:finite_primal_1threshold_obj}\\
			(D)_{n,k,m}\qquad \text{s.t.} \quad & d_t \geq \left(\frac{1-(1-i/m)^{\tau_t}}{i/m}\right) \sum_{\ell=0}^i f_\ell \frac{1}{m}+ (1-i/m)^{\tau_t} d_{t+1}, & \forall t\in[k],\forall i\in [m] \label{const:finite_dynamic_constr} \\
			\quad &  \sum_{\ell=0}^m n \left(1-\frac{\ell}{m}\right)^{n-1} f_\ell \frac{1}{m} = 1,  & \label{const:finite_max_value_const} \\
			\quad & f_{\ell-1} \geq f_{\ell},& \forall \ell \in [m] , \label{const:finite_nonincreasing}
		\end{align}
	\end{subequations}
	and its dual,
	\begin{subequations}\label{form:P_nkm}\small
		\begin{align}
			\max_{\substack{ \alpha \in \R_+^{k\times m} \\ \eta\in  \R_+^{m+2}} }\quad & v  \label{eq:finite_P_nk_eps_primal_1threshold_obj}\\
			(P)_{n,k,m}\qquad \text{s.t.} \quad & \sum_{i=1}^m \alpha_{1,i}\leq 1, & \label{const:finite_alpha_constr_1} \\
			\quad & \sum_{i=1}^m \alpha_{t+1,i}  \leq\sum_{i=1}^m (1-i/m)^{\tau_t} \alpha_{t,i},  & \forall t\in [k-1] ,\label{const:finite_alpha_constr_t}\\
			\quad &   v \frac{n}{m}\left(1-\frac{\ell}{m}\right)^{n-1} + \eta_{\ell+1} - \eta_{\ell} \leq \sum_{t=1}^k \sum_{i=\max\{\ell,1\}}^m \left( \frac{1-(1-i/m)^{\tau_t}}{i/m}   \right)\frac{\alpha_{t,i}}{m}, & \forall \ell \in \{0\}\cup [m]  , \label{const:finite_P_nk_esp_max_value_const} \\
			\quad & \eta_{0} = 0, \eta_{m+1}=0 , \label{const:finite_P_nk_initial_condition}
		\end{align}
	\end{subequations}
where $\alpha_{t,i}$ is the dual variable associated with Constraint~\eqref{const:finite_dynamic_constr}, $v$ with Constraint~\eqref{const:finite_max_value_const}, and $\eta_\ell$ with Constraint~\eqref{const:finite_nonincreasing} for $\ell=1,\ldots,m$; we add $\eta_0$ and $\eta_{m+1}$ to write Constraint~\eqref{const:finite_P_nk_esp_max_value_const} in closed form. Strong duality between these two programs holds since $(P)_{n,k,m}$ is feasible, with a finite optimal value. We show that feasible solutions of $(D)_{n,k,m}$ produce feasible solutions to $(D)_{n,k}$ up to some small error. We also show that feasible solutions of $(P)_{n,k,m}$ converge to solutions of $(P)_{n,k}$ when $m\to \infty$. If $\gamma_{n,k,m}^*$ and $\lambda_{n,k,m}^*$ denote the optimal values of the problems $(D)_{n,k,m}$ and $(P)_{n,k,m}$ respectively, the first result and the strong duality between the problems $(D)_{n,k,m}$ and $(P)_{n,k,m}$ show that
	\[
	\gamma_{n,k}^* -\varepsilon \leq \gamma_{n,k,m}^* = \lambda_{n,k,m}^*.
	\]
Our second result shows that $\limsup_{m} \lambda_{n,k,m}^* \leq \lambda_{n,k}^*$, and this proves strong duality. We formalize this in the following two propositions.
}

\modif{\begin{proposition}
		Let $n\geq 10$ and $m\geq n^3$. Then $\gamma_{n,k,m}^* \geq (1 - 12 n/\sqrt{m})\gamma_{n,k}^*$, where $\gamma_{n,k}^*=\gamma_{n,k}^*(\tau_1,\ldots,\tau_k)$ denotes the optimal value of $(D)_{n,k}$.
\end{proposition}}
\begin{proof}
	\modif{To avoid notational clutter, we assume $\sqrt{m}$ is integer; the proof is analogous if we replace $\sqrt{m}$ by $\lceil \sqrt{m}\rceil$. Given a solution $(\mathbf{d},\mathbf{f})$ of $(D)_{n,k,m}$, we can find another solution $(\bar{\mathbf{d}},\bar{\mathbf{f}})$ of $(D)_{n,k,m}$ where $\bar{f}_i=f_i$ for $i<m-\sqrt{m}$, $\bar{f}_i=0$ for $i\geq m-\sqrt{m}$ and $\bar{d}_1\leq (1+2n/m) d_1$. In other words, we can assume that $f$ is $0$ in the last $\sqrt{m}$ coordinates by paying a small multiplicative loss in the objective value; hence, from now on we assume that $f_i=0$ for $i\geq m-\sqrt{m}$. Let
		\[
		g_u = \frac{1}{1- 1/\sqrt{m}} \sum_{i=1}^m f_{i-1} \mathbf{1}_{[(i-1)/m,i/m)}(u).
		\]
		Note that $g$ is nonincreasing in $u$ and $g_u=0$ for $u\geq 1-1/\sqrt{m}$. Let $d_t'=(1+4n/\sqrt{m}) d_t$. We are going to show that $(\mathbf{d}',g)$ is feasible for $(D)_{n,k}$. First,
		\[
		\int_0^1 n (1-u)^{n-1} g_u \, \mathrm{d}u = \frac{1}{1-1/\sqrt{m}}\sum_{i=1}^m f_{i-1} \int_{(i-1)/m}^{i/m} n(1-u)^{n-1}\, \mathrm{d}u \geq \sum_{i=1}^m f_{i-1} n \left(1- \frac{i}{m}\right)^{n-1}\frac{1}{m} =1.
		\]
		Thus, $g$ satisfies Constraint~\eqref{const:max_value_const}. This is enough for feasibility; if the inequality is strict, we can decrease $g$ and this can only help the minimization in the objective value.}
	
	\modif{We now show that $(\mathbf{d}',g)$ satisfies Constraint~\eqref{const:dynamic_constr} for all $t=1,\ldots,k$. We detail the case $t=k$ and skip $t<k$ for brevity; the latter argument is completely analogous to the former. }
	\modif{We need the following approximation for the proof. For $i\leq m-\sqrt{m}/2$ and $q\in [(i-1)/m,i/m]$, we have
		\begin{align}
			\frac{(1-q)^j}{(1-i/m)^j} \leq \left( 1 + \frac{1}{m-i} \right)^j \leq \left( 1 + \frac{2}{\sqrt{m}} \right)^n \leq e^{2n/\sqrt{m}} \leq 1+ \frac{4n}{\sqrt{m}}, \label{ineq:approximation_1}
		\end{align}
		where we used the fact that $m\geq n^3$, which implies $2n/\sqrt{m}\leq 2/\sqrt{n}\leq 1$ and that $e^x \leq 1+2x$ holds for $x\in [0,1]$.}
	
	\modif{Now we show that $(\mathbf{d}',g)$ satisfies Constraint~\eqref{const:dynamic_constr} for $t=k$. Let $i\leq m-\sqrt{m}/2$ and $q\in [(i-1)/m,i/m]$. Then,
		\[
		\left( \frac{1-(1-q)^{\tau_k}}{q} \right) \int_0^q g_u \, \mathrm{d}u \leq \left( \frac{1+4n/\sqrt{m}}{1-1/\sqrt{m}} \right) \left( \frac{1-(1-i/m)^{\tau_k}}{i/m} \right) \sum_{j=0}^i f_j \frac{1}{m} \leq  \left( 1+ \frac{8n}{\sqrt{m}} \right) d_k ,
		\]
		where we used Inequality~\eqref{ineq:approximation_1} and the feasibility of $(\mathbf{d},\mathbf{f})$ in $(D)_{n,k,m}$.	Likewise, for $i>m-\sqrt{m}/2$ we have
		\[
		\int_0^q g_u \, \mathrm{d}u = \sum_{j=0}^i f_j \frac{1}{m} = \sum_{j=0}^{m-\sqrt{m}} f_j \frac{1}{\sqrt{m}} ,
		\]
		since $f_j=0$ for $j\geq m-\sqrt{m}$, and thus for $q\in [(i-1)/m,i/m]$ we have
		\[
		\left( \frac{1-(1-q)^{\tau_k}}{q} \right) \int_0^q g_u \, \mathrm{d}u \leq \left(  \frac{1-(1-(m-\sqrt{m})/m)^{\tau_k}}{(m-\sqrt{m})/m} \right) \sum_{j=0}^{m-\sqrt{m}} f_i \frac{1}{m} \leq d_k,
		\]
where we used the fact that $(1-(1-q)^{\tau_k})q^{-1}=\sum_{j=0}^{\tau_k-1}(1-q)^j$ is decreasing. Thus, $(\mathbf{d}',g)$ satisfies Constraint~\eqref{const:dynamic_constr} for $t=k$; the case for $ t < k $ can be argued similarly. %In a similar fashion it can be shown that $(\mathbf{d}',g)$ holds Constraint~\eqref{const:dynamic_constr} for $t<k$. 
Thus, $(\mathbf{d}',g)$ is feasible for $(D)_{n,k}$ and has objective value $(1-4n/\sqrt{m}) d_1$. The result now follows, taking into account the additional multiplicative loss from assuming $f_i=0$ for $i\geq m-\sqrt{m}$. }
\end{proof}

\modif{\begin{proposition}
		Let $n\geq 10$ and $m\geq n^3$. Then $\lambda_{n,k,m}^*(1-n/\sqrt{m})\leq \lambda_{n,k}^*$.
\end{proposition}}
\begin{proof}
	\modif{The proof of this proposition follows a similar scheme as the previous one. We take a feasible solution of $(P)_{n,k,m}$ and construct a feasible solution of $(P)_{n,k}$, with error decreasing in $m$.} 
	\modif{Let $(\alpha,\eta,v)$ be a solution of $(P)_{n,k,m}$. Similarly to the previous proof, we can show that there is another solution $(\bar{\alpha}, \bar{\eta},\bar{v})$ of $(P)_{n,k,m}$ such that $\alpha_{t,i}=0$ for $i\geq m-\sqrt{m}$ and $\bar{v}\geq (1-n/\sqrt{m}) v$. From now on, we assume that $\alpha_{t,i}=0$ for $i\geq m-\sqrt{m}$. Consider the following function $\hat{\alpha}:[0,1]\to \R_+$,
		\[
		\hat{\alpha}_{t,q}=\left(1- \frac{2}{\sqrt{m}}\right)^{t-1}\sum_{i=1}^{m-\sqrt{m}}  \alpha_{t,i}  \delta_{\{(i+1)/m \}} (q) = \left(1- \frac{2}{\sqrt{m}}\right)^{t-1} \sum_{i=1}^{m-1}  \alpha_{t,i}  \delta_{\{(i+1)/m \}} (q) ,
		\]
		where $\delta_{\{x_0\}}$ is the dirac function with mass in $x_0$. The function $\hat{\alpha}_{t}$ can be transformed into a proper function by replacing $\delta_{\{i/m\}}$ by $(1/\varepsilon)\mathbf{1}_{(i/m-\varepsilon/2,i/m+\varepsilon/2)}$ with $\varepsilon<1/m$; the rest of the proof remains almost unchanged with $\varepsilon\to 0$. Note that
		\[
		\int_0^1 \hat{\alpha}_{1,q} \, \mathrm{d}q =\sum_{i=1}^{m-1} \alpha_{1,i} \leq 1 ,
		\]
		and for $t\geq 1$,
		\begin{align*}
			\int_{0}^1 \hat{\alpha}_{t,q} (1-q)^{\tau_{t}} \, \mathrm{d}q &= \left(1- \frac{2}{\sqrt{m}} \right)^{t-1} \sum_{i=1}^m (1-(i+1)/m)^\tau \alpha_{t,i} \\
			&\geq \left(1- \frac{2}{\sqrt{m}}\right)^t \sum_{i=1}^m (1-i/m)^{\tau_t} \alpha_{t,i} \tag{since $\alpha_{t,i}=0$ for $i\geq m-\sqrt{m}$}\\
			& = \int_0^1 \hat{\alpha}_{t,q} \, \mathrm{d}q.
		\end{align*}
		Hence, Constraints~\eqref{const:constr_1_P_nk}-\eqref{const:constr_t_P_nk} hold. Before constructing the function $\hat{\eta}\in \mathcal{C}$, we make some assumptions over $\eta$ that do not change the feasibility of $(\alpha,\eta,v)$. Since for $\ell > m-\sqrt{m}$ we have $\alpha_{t,\ell}=0$, Constraint~\eqref{const:finite_max_value_const} gives us the following set of inequalities,
		\begin{align*}
			v n (1-(m-\sqrt{m}+1)/m)^{n-1}/m + \eta_{m-\sqrt{m}+2}-\eta_{m-\sqrt{m}+1} &\leq 0\\
			v n (1-(m-\sqrt{m}+2)/m)^{n-1}/m + \eta_{m-\sqrt{m}+3}-\eta_{m-\sqrt{m}+2} &\leq 0\\
			&\,\,\,\vdots\\
			v n (1-(m-1)/m)^{n-1}/m + \eta_{m}-\eta_{m-1} &\leq 0\\
			\eta_{m+1}-\eta_{m} &\leq  0 .
		\end{align*}
		From here, first, we notice that we can assume that all inequalities except maybe for the last one can be set to equality without affecting the feasibility of $(\alpha,\eta,v)$. Indeed, start by increasing $\eta_{m-\sqrt{m}+2}$ until the first inequality tightens. Notice that this step does not affect the other inequalities and keeps $\eta\geq 0$. Keep doing this until all inequalities except maybe for the last one are tightened. From here, we get that $\eta_\ell$ is decreasing in $\ell$ for $\ell \geq m-\sqrt{m}+2$. Now, note that reducing the value of $\eta_{\ell}$ for $\ell \geq m-\sqrt{m}+2$ by an scalar $\varepsilon>0$ does not affect the feasibility of $(\alpha,\eta,v)$ as long as $\varepsilon\leq \eta_m$. Hence, we can reduce $\eta_{m-\sqrt{m}+2},\ldots,\eta_{m-1},\eta_{m}$ by $\eta_m\geq 0$ without affecting the feasibility of $(\alpha,\eta,v)$. This shows that we can assume $\eta_m=0$.}
	
	\modif{We define the function $\hat{\eta}:[0,1]\to \R_+$ as follows:
		\[
		\hat{\eta}_u = \left( 1- \frac{2}{\sqrt{m}} \right)^{n} \left( \left( u m - \ell \right)(\eta_{\ell+1}-\eta_\ell) + \eta_{\ell}\right), \quad \text{for }u\in [\ell/m,(\ell+1)/m].
		\]
		Note that $\mathrm{d}\hat{\eta}_u / \mathrm{d}u = \left( 1- \frac{2}{\sqrt{m}} \right)^{n}  (\eta_{\ell+1}- \eta_\ell)m$ for all $u\in (\ell/m,(\ell+1)/m)$, $\ell=0,1,\ldots,m-1$. Moreover, $\hat{\eta}$ is a nonnegative continuous function with $\hat{\eta}_0=\eta_0$ and $\bar{\eta}_{1}= \eta_m=0$. Now, let $\hat{v}=(1-2/\sqrt{m})^{n} v$. Then, for $\ell=0,\ldots, m-1$ and $u\in (\ell/m,(\ell+1)/m)$ we have
		\begin{align*}
			\sum_{t=1}^k \int_u^1 \hat{\alpha}_{t,q} \left( \frac{1-(1-q)^{\tau_t}}{q} \right)\, \mathrm{d}q & \geq \left( 1- \frac{2}{\sqrt{m}} \right)^{n-1} \sum_{t=1}^k \sum_{i=\ell}^{m-1} \alpha_{t,i} \left(\frac{1-(1-(i+1)/m)^{\tau_t}}{(i+1)/m}\right) \\
			& \geq \left( 1- \frac{2}{\sqrt{m}} \right)^{n} \sum_{t=1}^k \sum_{i=\ell}^{m-1} \alpha_{t,i} \left(\frac{1-(1-i/m)^{\tau_t}}{i/m}\right) \\
			& \geq \left( 1- \frac{2}{\sqrt{m}} \right)^{n}  \left(m(\eta_{\ell+1}-\eta_\ell) + vn (1-\ell/m)^{n-1}  \right)  \\
			& \geq \frac{\mathrm{d} \hat{\eta}_u}{\mathrm{d}u} + \hat{v} n (1-u)^{n-1}.
		\end{align*}
		This shows that $(\alpha,\hat{\eta},\hat{v})$ is feasible for $(P)_{n,k}$. The conclusion of the proposition follows from here.}
\end{proof}

\section{Missing Proofs from Section~\ref{sec:existence_characterization_solutions}}\label{app:existence_characterization_solutions}

\begin{proof}[Proof of Proposition~\ref{prop:epsilon_decreasing_gamma}]
	For short, we are going to write $\varepsilon_t'$ to denote $\varepsilon_t'(v)$, the derivative of $\varepsilon_t$ with respect to $v$. The proof of the proposition proceeds by induction in $t$. For $t=1$ we have
	\[
	\frac{1}{n(n-1)v} =  \int_{0}^{\varepsilon_1} \frac{q}{1-(1-q)^{\tau}} (1-q)^{n-2} \, \mathrm{d}q \implies \frac{-1}{v^2 n(n-1)} = \frac{\varepsilon_1 (1-\varepsilon_1)^{n-2}}{1-(1-\varepsilon_1)^{\tau_1}} \varepsilon_1'.
	\]
	From here, we can deduce that $\varepsilon_1' <0 $. Before proving the inductive case, let's show an identity.
	\begin{claim}~\label{claim:identity_recursion_1}
		For any $t=0,\ldots,k-1$, %it holds
		\[
		\int_{\varepsilon_t}^{\varepsilon_{t+1}} \frac{q(1-q)^{n-2}}{1-(1-q)^{\tau_{t+1}}} \, \mathrm{d}q =\frac{1}{v n(n-1)} - \int_{0}^{\varepsilon_t} q (1-q)^{n-2} \, \mathrm{d}q.
		\]
	\end{claim}
	\begin{proof}
		Let
		\[
		A_t = \int_{\varepsilon_{t-1}}^{\varepsilon_{t}} \frac{q(1-q)^{n-2}}{1-(1-q)^{\tau_t}} \, \mathrm{d}q.
		\]
		For $t=1,\ldots,k-2$, we have
		\[
		A_{t+1} - A_t = - \int_{\varepsilon_{t-1}}^{\varepsilon_t} q (1-q)^{n-2} \, \mathrm{d}q.
		\]
		For $t<k-1$, $\tau_t=\tau$ and we can telescope the RHS of this identity, while the RHS will be an integral from $0$ to $\varepsilon_t$. This exactly gives us the result of the claim for $t<k-1$. To prove the result for $t=k-1$, we just need to notice that
		\[
		\int_{\varepsilon_{k-1}}^{\varepsilon_k} \frac{q(1-q)^{n-2}}{1-(1-q)^\sigma} \, \mathrm{d}q = A_{k-1} - \int_{\varepsilon_{k-2}}^{\varepsilon_{k-1}} q (1-q)^{n-2}\, \mathring{d}q.
		\]
		Here we use the identity of the claim for $A_{k-1}$ one more time and conclude the result.
	\end{proof}
	
	Using the claim, we can derive $v$ and obtain,
	\begin{align*}
		\frac{\varepsilon_{t+1}(1-\varepsilon_{t+1})^{n-2}}{1-(1-\varepsilon_{t+1})^{\tau_{t+1}}} \varepsilon_{t+1}' & = \frac{-1}{v^2 n(n-1)} - \varepsilon_t(1-\varepsilon_t)^{n-2}\varepsilon_t' + \frac{\varepsilon_{t}(1-\varepsilon_{t})^{n-2}}{1-(1-\varepsilon_{t})^{\tau_{t+1}}} \varepsilon_{t}'\\
		& = \frac{-1}{v^2 n(n-1)} + \frac{(1-\varepsilon)^{\tau_{t+1}}}{1-(1-\varepsilon_t)^{\tau_{t+1}}}  \varepsilon_t(1-\varepsilon_t)^{n-2}\varepsilon_t'.
	\end{align*}
	Using induction, we can show that the RHS of the equality is negative. Therefore, $\varepsilon_{t+1}' < 0$.
	
	Moreover, for any $t\geq 0$ we have
	\[
	\varepsilon_{t+1}' \leq - \frac{1}{v^2 n(n-1)} \left( \frac{1-(1-\varepsilon_{t+1})^{\tau_{t+1}}}{\varepsilon_{t+1}(1-\varepsilon_{t+1})^{n-2}} \right) \leq \frac{-1}{v^2 n(n-1)},
	\]
	where we used the fact that $1-(1-\varepsilon_{t+1})^{\tau_{t+1}} = \sum_{\ell=0}^{\tau_{t+1}-1} \varepsilon_{t+1}(1-\varepsilon_{t+1})^{\ell} \geq \varepsilon_{t+1}$.
\end{proof}

\begin{proposition}\label{prop:eps_k_1_less_1}
We have	$\varepsilon_k(1) < 1$.
\end{proposition}

\begin{proof}
	First, we prove by induction in $t$ that $\varepsilon_t (1) \leq t \tau / (n(n-1))$ for $t\leq k-1$. We use this result to show that $\varepsilon_k(1) \leq ((k-1)\tau + \sigma)/(n(n-1)) =1/(n-1)$, which implies the result.
	
	For the sake of notation, we write $\varepsilon_t= \varepsilon_t(1)$. Now, for $t=1$ we have
	\[
	\frac{1}{n(n-1)} = \int_0^{\varepsilon_1} \frac{q(1-q)^{n-1}}{1-(1-q)^\tau} \, \mathrm{d}q \geq \frac{1}{\tau (n-1)} ( 1 - (1-\varepsilon_1)^{n-1}).
	\]
	From here, by rearranging the inequality and using Bernoulli's inequality, we obtain $\varepsilon_1\leq \tau / (n(n-1))$.
	
	Assume the result for $t\geq 1$ and let us show it for $t+1 \leq k-1$. By using the identity in the previous proposition's proof, we have
	\[
	\frac{1}{n(n-1)} \geq \int_{\varepsilon_t}^{\varepsilon_{t+1}} \frac{q(1-q)^{n-2}}{1-(1-q)^{\tau}} \, \mathrm{d}q \geq \frac{1}{\tau (n-1)} ( (1-\varepsilon_{t})^{n-1} - (1-\varepsilon_{t+1})^{n-1}  ) .
	\]
	Rearranging terms, using the inductive hypothesis over $\varepsilon_t$ and Bernoulli's inequality, we obtain
	\[
	(1 - \varepsilon_{t+1})^{n-1} \geq  (1-\varepsilon_t)^{n-1} - \frac{\tau}{n} \geq \left(1 - \frac{t \tau}{n}\right) - \frac{\tau}{n} = 1 - (t+1) \frac{\tau}{n}.
	\]
	Using Bernoulli's inequality one more time and rearranging terms we obtain $\varepsilon_{t+1} \leq (t+1)\tau / (n(n-1))$.
	
	The calculation for $\varepsilon_k$ is similar to the previous analysis and is skipped for brevity.
\end{proof}

\begin{proof}[Proof of Proposition~\ref{prop:backward_induction_vt_at}]
	We prove the result by backward induction in $t=k,k-1,\ldots,1$. Before we proceed with the proof, we need a fact that can be deduced from Proposition~\ref{prop:properties_of_a_t}.
	\begin{claim}\label{claim:equality_for_s_less_t}
		Fix $t\in [k]$. Then for any $s$,
		\[
		(a_s - a_{s+1}) \left(  \frac{1-(1-\varepsilon_s)^{\tau_t}}{1-(1-\varepsilon_s)^{\tau_s}}  \right) = (a_{s+1}-a_{s+2}) (1-\varepsilon_{s})^{\tau_{s+1}} \left(  \frac{1-(1-\varepsilon_s)^{\tau_t}}{1-(1-\varepsilon_s)^{\tau_{s+1}}}  \right) .
		\]
	\end{claim}	
	Now, let
	\[
	g_t(q)=\left( \frac{1-(1-q)^{\tau_t} }{q} \right)F(q) + (1-q)^{\tau_t} a_{t+1}.
	\]
	We begin the backward induction for the proof at $t=k$. Using $\tau_k = \sigma \leq \tau_s$ for any $s$, we have
	\begin{align*}
		d_k &= \sup_{q\in [0,1]}\left\{ g_k(q) \right\} \\
		&= \sup_{\substack{s\in [k] \\ q\in [\varepsilon_{s-1},\varepsilon_s]}} \left\{  (a_s-a_{s+1})\left( \frac{1-(1-q)^{\tau_k}}{1-(1-q)^{\tau_s}} \right) + (1-(1-q)^{\tau_k})a_{s+1}  \right\}\\
		& = \sup_{s\in [k]} \left\{  g_k(\varepsilon_s) \right\} ,
	\end{align*}
	where we used the fact that the function $q\mapsto (1-(1-q)^{a})/(1-(1-q)^{b})$ is nondecreasing if $b \geq a$, and the function $q\mapsto (1-(1-q)^{a})$ is always increasing. Note that for $s=k$, $g_k(\varepsilon_k)=a_k$. We claim that for any $s<k$, $g_k(\varepsilon_s) \leq g_k(\varepsilon_{s+1})$. Indeed, using Claim~\ref{claim:equality_for_s_less_t} with $t=k$, we have
	\begin{align*}
		g_{k}(\varepsilon_s)=&(a_s-a_{s+1})\left( \frac{1-(1-\varepsilon_s)^{\tau_k}}{1-(1-\varepsilon_s)^{\tau_s}} \right) + (1-(1-\varepsilon_s)^{\tau_k}) a_{s+1}  \\
		= & (a_{s+1}-a_{s+2}) (1-\varepsilon_s)^{\tau_{s+1}} \left( \frac{1-(1-\varepsilon_s)^{\tau_t}}{1-(1-\varepsilon_s)^{\tau_{s+1} }} \right) + (1-(1-\varepsilon_s)^{\tau_k}) a_{s+1} \\
		= & (a_{s+1}-a_{s+2})\left( \frac{1-(1-\varepsilon_s)^{\tau_t}}{1-(1-\varepsilon_s)^{\tau_{s+1}}} \right) + (1-(1-\varepsilon_s)^{\tau_k}) a_{s+2} \\
		\leq & (a_{s+1}-a_{s+2})\left( \frac{1-(1-\varepsilon_{s+1})^{\tau_t}}{1-(1-\varepsilon_{s+1})^{\tau_{s+1}}} \right) + (1-(1-\varepsilon_{s+1})^{\tau_k}) a_{s+2} ,
	\end{align*}
	where in the last inequality we used the monotonicity of $q\mapsto (1-(1-q)^{a})/(1-(1-q)^{b})$ and the function $q\mapsto (1-(1-q)^{a})$. This shows that $g_k(\varepsilon_s)\leq g_k(\varepsilon_{s+1})$. From here we deduce that $d_k=a_k$.
	
	Now, assume the result is true for $t+1$ and let us show it for $t<k$. We have
	\begin{align*}
		d_t &= \sup_{q\in [0,1]}\left\{ \left( \frac{1-(1-q)^{\tau_t}}{q} \right) F(q) + (1-q)^{\tau_t} d_{t+1}  \right\} \\
		& =\sup_{q\in [0,1]}\left\{ g_t(q) \right\} \tag{Induction} \\
		& =\sup_{\substack{s\in [k]\\ q\in [\varepsilon_{s-1},\varepsilon_{s}]}}\left\{ (a_s-a_{s+1})\left( \frac{1-(1-q)^{\tau_t}}{1-(1-q)^{\tau_s}} \right) + (1-(1-q)^{\tau_t})a_{s+1} + (1-q)^{\tau_t} a_{t+1}  \right\} \\
		& = \sup \left\{  \sup_{s\leq t} \{ g_t(\varepsilon_s) \} , \sup_{s>t} \{ g_t(\varepsilon_{s-1}) \}  \right\}.
	\end{align*}
	The last equality can be justified as follows
	\begin{itemize}
		\item For $s\leq t$, $\tau_s = \tau_t=\tau$ and so, for $q\in [\varepsilon_{s-1},\varepsilon_s]$, we have
		\[
		g_t(q) = (a_s - a_{s+1}) + (1-(1-q)^\tau)a_{s+1} + (1-q)^\tau a_{t+1}
		\]
		since $a_{s+1} \geq a_{t+1}$ (Proposition~\ref{prop:properties_of_a_t}), we have that $g_t(q)$ is increasing in $[\varepsilon_{s-1},\varepsilon_s]$. Thus,
		\[
		\sup_{q\in [\varepsilon_{s-1},\varepsilon_s]} \{ g_t(q) \} = g_t(\varepsilon_s).
		\]
		\item For $s> t$, we have $\tau_s \leq \tau_t$ and so, for $q\in [\varepsilon_{s-1},\varepsilon_s]$, we have
		\[
		g_t(q) =  (a_s-a_{s+1})\left( \frac{1-(1-q)^{\tau_t}}{1-(1-q)^{\tau_s}} \right) + (1-(1-q)^{\tau_t})a_{s+1} + (1-q)^{\tau_t} a_{t+1}.
		\]
		The function $q\mapsto (1-(1-q)^{\tau_t})/(1-(1-q)^{\tau_s})$ is nonincreasing since $\tau_t\geq \tau_s$. Also, $a_{s+1}\leq a_{t+1}$. Thus,
		\[
		\sup_{q\in [\varepsilon_{s-1},\varepsilon_s]} \{ g_t(q) \} = g_t(\varepsilon_{s-1}).
		\]
	\end{itemize}
	Now, for $s=t$, $g_t(\varepsilon_s) = g_t(\varepsilon_t) = a_t$. Using Claim~\ref{claim:equality_for_s_less_t}, we can show again that $g_t(\varepsilon_s)\leq g_t(\varepsilon_{s+1})$ for $s< t$. To conclude, we will show that $g_t(\varepsilon_{s}) \leq g_{t}(\varepsilon_{s-1})$ for $s> t$. Indeed,
	\begin{align*}
		g_t(\varepsilon_{s-1}) & = \left(\frac{1-(1-\varepsilon_{s-1})^{\tau_t}}{\varepsilon_{s-1}}\right) F(\varepsilon_{s-1}) + (1-\varepsilon_{s-1})^{\tau_t} a_{t+1} \\
		& = \left(\frac{1-(1-\varepsilon_{s-1})^{\tau_t}}{1-(1-\varepsilon_{s-1})^{\tau_s}}\right)\left( a_{s} - a_{s+1}(1-\varepsilon_{s-1})^{\tau_{s}}  \right)  + (1-\varepsilon_{s-1})^{\tau_t} a_{t+1} \tag{Using continuity of $F$. See also Proposition~\ref{prop:properties_of_a_t}} \\
		& = (a_s- a_{s+1})\left(\frac{1-(1-\varepsilon_{s-1})^{\tau_t}}{1-(1-\varepsilon_{s-1})^{\tau_s}}\right)+ (1-(1-\varepsilon_{s-1})^{\tau_{t}})a_{s+1}  + (1-\varepsilon_{s-1})^{\tau_t} a_{t+1} \\
		& \geq (a_s- a_{s+1})\left(\frac{1-(1-\varepsilon_{s})^{\tau_t}}{1-(1-\varepsilon_{s})^{\tau_s}}\right)+ (1-(1-\varepsilon_{s})^{\tau_{t}})a_{s+1}  + (1-\varepsilon_{s})^{\tau_t} a_{t+1}\\
		& = g_t(\varepsilon_s),
	\end{align*}
	where in the last inequality we used again that the function $q\mapsto (1-(1-q)^{\tau_t})/(1-(1-q)^{\tau_s})$ is nonincreasing since $\tau_t\geq \tau_s$, and that $a_{s+1}\leq a_{t+1}$. Thus, $d_t=a_t$ and this finishes the proof.	
\end{proof}

\section{Missing Proofs of Section~\ref{sec:asymptotic_analysis}}\label{app:asymptotic_analysis}

\begin{proof}[Proof of Strong Duality in the Infinite Model]
	The proof is similar to the finite model case, with finite $n$. We set the primal solution of $(CLP)_{k,\infty}$ to be
	\[
	\omega_{t,y} = -v^* \frac{y\log y}{1-y^{1/k}}\mathbf{1}_{(y_{t},y_{t-1})}(y)
	\]
	with $y_0=1\geq y_1 \geq \cdots \geq y_k$ and $v^*$ such that $y_k=0$. Repeating verbatim the proof in Section~\ref{sec:existence_characterization_solutions}, we can easily show the existence of such a $v^*$ and that the value of $(CLP)_{k,\infty}$, namely $v_{\infty,k}^*$, is at least $v^*$.
	
	To show the optimality of $v^*$, we produce a dual solution of $(DLP)_{\infty,k}$ with objective value $v^*$.	Define
	\[
	a_t = v^* \left(\frac{1+\sum_{\tau=t}^{k}\prod_{\sigma=\tau}^{k}y_\sigma^{1/k} }{1+\sum_{\tau=1}^k \prod_{\sigma=\tau}^k y_{\sigma}^{1/k}}\right)
	\]
	and $a_{k+1}=0$. Then, the function
	\[
	G(u) = u\sum_{t=1}^k \left( \frac{a_t-a_{t+1}e^{-u/k}}{1-e^{-u/k}} \right)\mathbf{1}_{[u_{t-1},u_t)}(u)
	\]
	is right-continuous in $(0,+\infty)$ and increasing for $u\geq 0$, where $u_t = -\log y_t$ for all $t$, and we use $u_k=+\infty$. The function $G$ defines a Lebesgue-Stieltjes measure $\mu_G$ over $[0,+\infty)$ via
	\[
	\mu_G(E) = \delta_{\{0 \}} (E) G(0) + \int_{E\cap (0,+\infty)} G'(u) \, \mathrm{d}u ,
	\]	
	where $G(0)=\lim_{u\to 0^+} G(u)$.
	
	Let $\overline{\mu}_G(E) = \mu (P^{-1}(E))$ be the pushforward measure, where $P:[0,+\infty)\to (0,1]$ is $P(u)=e^{-u}$. Then $P^{-1}(y)=-\log y$. So, for $\overline{y}>0$,
	\begin{align*}
		\overline{\mu}[\overline{y},1] & = \mu_G[0,-\log \overline{y}] \\
		& = G(0) + \mu_G(0,-\log \overline{y}] \\
		& = G(0) + G(-\log \overline{y}) - G(0) \\
		& = G(-\log \overline{y}) = -\log \overline{y}\sum_{t=1}^k \left( \frac{a_t-a_{t+1}\overline{y}^{1/k}}{1-\overline{y}^{1/k}}  \right)\mathbf{1}_{(y_{t},y_{t-1}]}(\overline{y}) .
	\end{align*}
	Then,
	\begin{align*}
		\int_{[0,1]} y \, \mathrm{d}\overline{\mu}(y) & = \int_{[0,1]} \int_{0}^1 \mathrm{d}\overline{y}  \, \mathrm{d}\overline{\mu}(y) \\
		& = \int_0^1 \int_{[\overline{y},1]} \, \mathrm{d}\mu_G(y)  \, \mathrm{d}\overline{y} \\
		& = \int_0^1  -\log \overline{y}\sum_{t=1}^k \left( \frac{a_t-a_{t+1}\overline{y}^{1/k}}{1-\overline{y}^{1/k}}  \right)\mathbf{1}_{(y_{t},y_{t-1}]}(\overline{y})   \, \mathrm{d}\overline{y} \\
		& = 1 .
	\end{align*}
	
	We define $d_1,\ldots,d_{t+1}$ recursively as follows. Let $d_{t+1}=0$; for $t\leq k$, we have
	\[
	d_{t} = \sup_{y\in [0,1]} \left\{  \left(  \frac{1-y^{1/k}}{-\log y}   \right) \mu[y,1] + y^{1/k} d_{t+1}  \right\}.
	\]	
	Thus, the solution $(\overline{\mu}_G,d)$ is feasible for $(DLP)_{\infty,k}$. Moreover, $d_1=v^*$ as a byproduct of the following general result.
	
	\begin{proposition}
		For all $t=1,\ldots,k$, $d_t=a_t$.
	\end{proposition}
	
	The proof of this proposition is analogous to the proof of Proposition~\ref{prop:backward_induction_vt_at} in Section~\ref{sec:existence_characterization_solutions}.
\end{proof}

\subsection{Proof of Lower Bound in Theorem~\ref{thm:approximation_large_n}}

In this subsection, we show that $v_{n,k}^*$ is lower bounded by the term $(1+4 \log y_{k-1}/(n-1) - c k^2/n) v_{\infty,k}^*$, where $c$ is a constant. To prove this bound, we take the optimal solution of $(CLP)_{\infty,k}$, denoted $(\omega, v_{\infty,k}^*)$ and we construct a solution $(\alpha,v^*)$ of $(CLP)_{n,k}$, where the value $v^*$ holds the aforementioned lower bound.

Before describing $\alpha$, we need to introduce some intermediary values. Fix $\delta>0$ to be a small number, and let $n$ be large enough so that $y_{k-1}> e^{-(1-\delta)(n-1)}$. In principle, $n$ could depend on $\delta$; however, if we assume $\delta\leq 1/2$, it is enough to take any $n \geq -\log \sqrt{y_{k-1}}$, which is independent of $\delta$; this is crucial for the conclusion of the proof. Consider the  scalars
\[
c_1= \left(\frac{n-1}{n}\right)\left(  1 + 4 \frac{\log y_{k-1}}{k(n-1)} \right)^{t-1} ,
\]
for $t=1,\ldots,k$. Note that $c_{t+1}<c_t$ since $y_{k-1}\in (0,1)$. Define the family of functions $\alpha$ as
\[
\alpha_{t,q} = c_t n \omega_{t,e^{-q(n-1)}}\mathbf{1}_{[0,1-\delta)}(q) + \left(\frac{1-\delta}{\delta}\right)\mathbf{1}_{[1-\delta,1]}(q)\int_0^{e^{-(1-\delta)(n-1)}} n \omega_{t,y} \left(  \frac{1-y^{1/k}}{-y\log y}  \right)\, \mathrm{d}y .
\]
First, we verify that Constraints~\eqref{const:infinity_dynamic_1}-\eqref{const:infinity_dynamic_2} hold, and then we define $v^*$ such that $(\alpha,v^*)$ is feasible for $(CLP)_{n,k}$, with $v^*=(1+4 \log y_{k-1}/(n-1) - c k^2/n) v_{\infty,k}^*$ for an appropriate constant $c$.

For $t=1$, we have
\begin{align*}
	\int_0^1 \alpha_{1,q} \, \mathrm{d}q & = c_1 \frac{n}{n-1}\int_{e^{-(1-\delta)(n-1)}}^1 \frac{\omega_{1,y}}{y} \, \mathrm{d}y \tag{change of variable, $q=-\log y/(n-1)$} \\
	& \leq c_1 \frac{n}{n-1} =1.
\end{align*}
Note that we implicitly use that the support of $\omega_{1,q}$ is in $[y_1,1]$, and $y_1\geq y_{k-1} >e^{-(1-\delta)(n-1)}$.

The following proposition bounds the error of approximating $y^{1/k}$ by $(1+ \log y/(n-1))^{\tau}$, where $\tau= \lceil n/k \rceil$. The proof is a simple calculation, deferred to Appendix~\ref{app:asymptotic_analysis}.

\begin{proposition}\label{prop:bound_log_infinite_lp}
	For $y\geq y_{k-1}$ we have
	\[
	\left(  1 + \frac{\log y}{n-1} \right)^{\tau} \geq y^{1/k} \left(  1 + 4\frac{\log y_{k-1}}{k(n-1)}  \right).
	\]
\end{proposition}

\begin{proof}
	We have
	\begin{align*}
		\frac{\log y}{n-1} & = -\frac{\tau k}{n-1} \int_{y^{1/\tau k}}^{1} \frac{1}{u} \, \mathrm{d}u \\
		& \geq - \frac{\tau k}{n-1} \left( 1- y^{1/\tau k}   \right) .
	\end{align*}
	Then,
	\[
	1+ \frac{\log y}{n-1} \geq 1 -\frac{\tau k}{n-1}\left( 1- y^{1/\tau k} \right) \geq y^{1/\tau k} \left( 1 + \left(\frac{\tau k - (n-1)}{n-1}\right)(1-y_{k-1}^{-1/\tau k}) \right),
	\]
	where we used  $y\geq y_{k-1}$. Now,
	\[
	y_{k-1}^{-1/\tau k} = e^{-\log y_{k-1}/\tau k}\leq 1 - 2\frac{\log y_{k-1}}{ \tau k } \implies 1-y_{k-1}^{-1/\tau k } \geq 2\frac{\log y_{k-1}}{\tau k} ,
	\]
	using $y_{k-1} \geq e^{-(1-\varepsilon)(n-1)}$ and for $n$ large, $-\log y_{k-1} \leq \tau k$. We now elevate  $1+\log y /(n-1)$ to the power $\tau$ and use the previous bound to obtain
	\[
	\left( 1 + \frac{\log y}{n-1} \right)^{\tau} \geq y^{1/k} \left(  1 + 2\left( \frac{\tau  k - (n-1)}{k}\right) \frac{\log y_{k-1}}{n-1}  \right) ,
	\]
	where in the last inequality we used Bernoulli's inequality, $(1+x)^r \geq 1+ rx$. The conclusion now follows by using $\tau=\lceil n/k \rceil < n/k + 1$, and so
	\[
	\frac{\tau k - (n-1)}{k} = \tau - \frac{n}{k} + \frac{1}{k} < 1 + \frac{1}{k} \leq 2,
	\]
	which finishes the proof.
\end{proof}

Now, for $1\leq t\leq k-1$ we have
\begin{align*}
	\int_0^1 (1-q)^{\tau} \alpha_{t,q} \, \mathrm{d}q &\geq c_t \frac{n}{n-1}\int_{e^{-(1-\delta)(n-1)}}^1 \left( 1+\frac{\log y}{n-1} \right)^{\tau} \frac{\omega_{t,y}}{y} \, \mathrm{d}y \\
	&\geq  c_t \frac{n}{n-1} \left( 1+4 \frac{\log y_{k-1}}{k(n-1)} \right) \int_{0}^1  y^{1/k} \frac{\omega_{t,q}}{y} \, \mathrm{d}y,
\end{align*}
since the support of $\omega_{t,q}$ is contained in $[y_{t},y_{t-1}]$ and $y_{t}\geq y_{k-1}$. Also,
\begin{align*}
	\int_0^1 \alpha_{t+1,q} \, \mathrm{d}q &= c_{t+1} \frac{n}{n-1}\int_{e^{-(1-\delta)(n-1)}}^{1} \frac{\omega_{k,y}}{y} \, \mathrm{d}y + c_{t+1} n\int_{0}^{e^{-(1-\delta)(n-1)}} \omega_{t,y} \left(  \frac{1-y^{1/k}}{-y\log y} \right) \, \mathrm{d}y \\ % \tag{$\frac{-\log y}{1-y^{1/k}} \geq (1-\delta)(n-1)$ for $y\leq e^{-(1-\delta)(n-1)}$} \\
	& \leq c_{t+1} \frac{n}{n-1}\int_{0}^{1} \frac{\omega_{t,y}}{y} \, \mathrm{d}y  \\
	& \leq c_{t+1} \frac{n}{n-1} \int_0^1 y^{1/k} \frac{\omega_{t+1,y}}{y} \, \mathrm{d}y \\
	& \leq \int_0^1 \alpha_{t,q} (1-q)^{n/k} \, \mathrm{d}q .
\end{align*}
In the equality we use the fact that $\frac{-\log y}{1-y^{1/k}} \geq (1-\delta)(n-1)$ for $y\leq e^{-(1-\delta)(n-1)}$, and in the last inequality we use $c_{t+1}=c_{t}\left( 1 + 4\frac{\log y_{k-1}}{k(n-1)}  \right)$. With this, we have shown that $\alpha$ satisfies Constraints~\eqref{const:infinity_dynamic_1}-\eqref{const:infinity_dynamic_2}.

Let $v = \overline{c}_k  v_{\infty,k}^*$, where $\overline{c}_k=(1-\delta)(1-k^2/n) c_k$. We now prove that $(\alpha,v)$ satisfies Constraint~\ref{const:infinity_multiobj}. Before we do this, we need the following facts, proved in Appendix~\ref{app:asymptotic_analysis}.

\begin{proposition}\label{prop:bound_log_lower_bound}
	For any $y\in [0,1]$ and $t\in [k]$,
	\[
	1-\left( 1 + \frac{\log y}{n-1}  \right)^{\tau_t} \geq \left(1 - \frac{k^2}{n}\right) \left(1 - y^{1/k}\right)
	\]
	where $\tau_t = \lceil n/k \rceil$ if $t< k$ and $\tau_k = n-(k-1)\lceil n/k \rceil$.
\end{proposition}

\begin{proof}
	We do the case $\tau_t=\tau = \lceil n/k \rceil$ for $t<k$ and the case $\tau_k$ separately:
	\begin{align*}
		\left( 1 + \frac{\log y}{n-1} \right)^{\tau} & \leq y^{\tau/(n-1)} \\
		& \leq y^{n/(k(n-1))} \tag{since $\lceil n/k \rceil \geq n/k$ and $y\leq 1$}\\
		& \leq y^{1/k}.
	\end{align*}
	From here, we can deduce the inequality for $t<k$. For $\tau_k$ we have the following. First, it is easy to verify that
	\[
	\tau_k = n- (k-1) \lceil n/k \rceil \geq (1-k^2/n) \lceil n/k \rceil = (1-k^2/n) \tau.
	\]
	Thus, repeating the same calculations as before, we have
	\begin{align*}
		\left(  1 + \frac{\log y}{n-1}  \right)^{\tau_k} &\leq y^{\tau_k/(n-1)} \\
		& \leq y^{(1-k^2/n) \tau/(n-1)} \\
		& \leq y^{(1-k^2/n)/k}.
	\end{align*}
	Now, an easy calculation shows that $\inf_{v\in [0,1]} (1-v^{(1-k^2/n)}/(1-v)) \geq (1-k^2/n)$. Thus, for any $y\in [0,1]$, we have
	\[
	1-\left(1+\frac{\log y}{n-1}\right)^{\tau_k} \geq  1-y^{(1-k^2/n)/k} \geq \left( 1 - \frac{k^2}{n} \right) (1-y^{1/k}).
	\]
\end{proof}

Let $u\in [0,1-\delta)$; using the definition of $\alpha$ and the previous proposition, we get
\begin{align*}
	\sum_{t=1}^k \int_u^1 \alpha_{t,q}\left( \frac{1-(1-q)^{\tau_t}}{q}   \right)\, \mathrm{d}q & = n\sum_{t=1}^k \int_{e^{-(1-\delta)(n-1)}}^{e^{-u(n-1)}} c_t \frac{\omega_{t,y}}{y} \left(  \frac{1-(1+\log y/n)^{n/k}}{-\log y} \right) \, \mathrm{d}y \\
	& \quad + (1-\delta)n\sum_{t=1}^k c_t \int_{0}^{e^{-(1-\delta)(n-1)}} \omega_{t,y} \left(\frac{1-y^{1/k}}{-y\log y} \right)\, \mathrm{d}y \\
	& \geq \overline{c}_{k} n\sum_{t=1}^k \int_{0}^{e^{-u (n-1)}} \omega_{t,y} \frac{1-y^{1/k}}{-y \log y} \, \mathrm{d}y \\
	& \geq \overline{c}_k n  v_{\infty,k}^*  e^{-u(n-1)} \\
	& \geq v^* n (1-u)^{n-1}.
\end{align*}
For $u\in [1-\delta,1)$, we have
\begin{align*}
	\sum_{t=1}^k \int_u^1 \alpha_{t,q}\left( \frac{1-(1-q)^{\tau_t}}{q}   \right)\, \mathrm{d}q & = (1-\delta)n \frac{(1-u)}{\delta} \sum_{t=1}^k \int_0^{e^{-(1-\delta)(n-1)}} \omega_{t,y} \left(\frac{1-y^{1/k}}{-y \log t} \right) \, \mathrm{d}y \\
	& \geq \overline{c}_k v_{\infty,k}^* n \frac{(1-u)}{\delta} e^{-(1-\delta)(n-1)}\\
	& \geq v^*  n (1-u)^{n-1},
\end{align*}
where in the last line we used the fact that the function $f(u)=(1-u)e^{-(1-\delta)(n-1)}/\delta - (1-u)^{n-1}$ is concave, $f(1-\delta) >0$, $f(1)=0$ and $f'(1-\delta)>0$, $f(1)<0$; hence, $f(u)\geq 0$ in $[1-\delta,1]$.

This shows that $(\alpha, v^*)$ is a feasible solution of $(CLP)_{k,n}$, and therefore $v_{n,k}^* \geq \overline{c}_k v_{\infty,k}^*$. Since $v_{n,k}^*$ is independent of $\delta$, we can take $\delta\to 0$ and obtain $v_{n,k}^* \geq (1-k^2/n)c_k v_{\infty,k}^*$. By using Bernoulli's inequality over $c_k$, we conclude
\[
v_{n,k}^* \geq \left(1 - \frac{k^2}{n}\right)\left( 1+ 4 \frac{\log y_{k-1}}{n-1}  \right) v_{\infty,k}^*.
\]
This finishes the proof of the lower bound in Theorem~\ref{thm:approximation_large_n}.

\subsection{Proof of Upper Bound in Theorem~\ref{thm:approximation_large_n}}\label{app:upper_bound_v_n_k}

In this Subsection, we show that $v_{n,k}^*$ is upper bounded by a term that converges to $v_{\infty,k}^*$. We take the optimal solution $(\overline{\mu},d)$ of $(DLP)_{\infty,k}$ and we construct a new solution $(\mu,d^*)$ of $(DLP)_{n,k}$ that shows the desired property.

For simplicity, we assume that $n/k$ is an integer, thus $\tau_1=\cdots = \tau_k=n/k$. This fundamentally does not change the upper bound.

Recall the definition of $G$ given at the beginning of this section,
\[
G(u) = u\sum_{t=1}^k \left( \frac{a_t-a_{t+1}e^{-u/k}}{1-e^{-u/k}} \right)\mathbf{1}_{[u_{t-1},u_t)}(u).
\]
This function is right-continuous and nondecreasing. Now, let $F(q) = G(qn)/n$. Then, $F$ is also right-continuous and nondecreasing. Thus, we can consider the Lebesgue–Stieltjes measure generated by $F$, that we denote $\mu_F$. Note that
\[
\mu_F[0,q] = F(q) = \frac{1}{n}G(q n) = \frac{1}{n}\overline{\mu}_G[e^{-qn},1] ,
\]
where $\overline{\mu}_G$ is the pushforward measure induced by $G$ and the function $P(u)=e^{-u}$ (see the proof at the beginning of the section). Let $c_n=\int_0^1 n(1-u)^{n-1} \, \mathrm{d}\mu_F(u)$. Then,
\begin{align*}
	c_n & = \int_0^1 n(n-1)(1-q)^{n-2}\mu_F[0,q] \, \mathrm{d}q \\
	&  =\int_0^1 n(n-1)(1-q)^{n-2} F(q) \, \mathrm{d}q \\
	& = \int_0^1 (n-1)(1-q)^{n-2} \overline{\mu}_G[e^{-qn},1] \, \mathrm{d}q \\
	& = \int_{e^{-n}}^1 (n-1) (1 + \log y / n)^{n-2} \overline{\mu}_G[y,1] \frac{1}{y n} \mathrm{d}y  \\
	& = \left(\frac{n-1}{n}\right) \int_{e^{-n}}^1 \int_{e^{-n}}^s \frac{(1+\log y/n)^{n-2}}{y} \, \mathrm{d}y \,\mathrm{d}\overline{\mu}_G(s) .
\end{align*}
Let $f_n(s) = \mathbf{1}_{[e^{-n},1]}(s)\int_{e^{-n}}^s y^{-1}(1+\log y /n)^{n-2} \, \mathrm{d}y$. A simple proof shows that $f_n(s)\to f(s)=s$ as $n\to \infty$. Moreover, $f_n(s)$ is bounded by the function $e^2 f(s)=e^2 s$ which is integrable with respect to $\overline{\mu}_G$. Thus, using Lebesgue's dominated convergence theorem, we obtain
\[
c_n \to \int_0^1 f(s) \, \mathrm{d}\overline{\mu}_G(s) = \int_0^1 s \, \mathrm{d}\overline{\mu}_G(s)=1.
\]
Define $\mu_F'=\mu_F/c_n$. Then,
\[
\int_0^1 n (1-u)^{n-1} \, \mathrm{d}\mu_F'(u) = 1.
\]
Let
\[
d_t^* = c_n^{-1}  \alpha_{n,t}^{-1} d_t,
\]
where $\alpha_{n,t}= (1-20k/n)^{k-t+1}$. We now prove that $(\mu_F',d^*)$ is feasible for $(DLP)_{n,k}$. For this, we need the following result. The proof is deferred to end of the subsection.

\begin{proposition}\label{prop:lower_bound_1-exp}
	For any $q\in [0,1]$, $1- e^{-q n/k} \geq \left(  1 - 20 k/n  \right)\left(  1 - (1-q)^{n/k}  \right)$.
\end{proposition}

To prove that $(\mu_F',d_t^*)$ is feasible for $(DLP)_{n,k}$, we only need to show that
\[
d_t^* \geq \sup_{q\in [0,1]} \left\{  \left(  \frac{1-(1-q)^{n/k}}{q}\right) \mu_F'[0,q] + (1-q)^{n/k} d_{t+1}^*    \right\}.
\]
Let $q\in [0,1]$; then
\begin{align*}
	c_n \alpha_{n,t} d_t^* & = d_t \\
	& \geq \left(  \frac{1-e^{-qn/k}}{qn}\right) \overline{\mu}_G[e^{-qn},1] + e^{-qn} d_{t+1} \tag{$y=e^{-qn}$} \\
	& \geq \left( 1 - 20\frac{k}{n}  \right)\left( \frac{1-(1-q)^{n/k}}{q} \right)\mu_F'[0,q] c_n + (1-q)^{n/k} d_{t+1} \tag{Prop.~\ref{prop:lower_bound_1-exp}} \\
	& \geq \left( 1 - 20 \frac{k}{n} \right) c_n \alpha_{n,t+1}\left( \left( \frac{1-(1-q)^{n/k}}{q}  \right) \mu_F'[0,q] + (1-q)^{n/k} d_{t+1}^*  \right).  \tag{Def. of $d_{t+1}^*$}
\end{align*}
We can use $\alpha_{n,t}=(1-20k/n)\alpha_{n,t+1}$, with $\alpha_{n,k+1}=1$, to simplify terms. From here we deduce the feasibility of $(\mu_F',d^*)$ in $(DLP)_{n,k}$.

Now, to show the upper bound of $v_{n,k}^*$, we can minimize over $(\mu,d)$ feasible for $(DLP)_{n,k}$ and obtain the guarantee
\[
v_{n,k}^* \leq d_1^* = c_n^{-1} \alpha_{n,1}^{-1} d_1 = c_{n}^{-1} \alpha_{n,1}^{-1} v_{\infty,k}^* ,
\]
which concludes the proof of the upper bound in Theorem~\ref{thm:approximation_large_n}.

We close this subsection with the proof of Proposition~\ref{prop:lower_bound_1-exp}.

\begin{proof}[Proof of Proposition~\ref{prop:lower_bound_1-exp}]
	Recall that $e^{-q} \leq 1-q + q^2/2$. Thus, it is enough to show that for any $q\in [0,1]$,
	\[
	20\frac{k}{n} \geq \frac{(1-q-q^2/2)^{n/k}-(1-q)^{n/k}}{1-(1-q)^{n/k}} = \frac{ f(1-q+q^2/2) - f(1-q)}{f(1)- f(1-q)},
	\]
	where $f(q)= x^{n/k}$. Note that $f$ is convex. Now, we divide the proof into two cases: $q\leq 10k/n$ and $q> 10k/n$.
	\begin{itemize}
		\item For $q\leq 10 k/n$, note that
		\[
		\frac{q}{2} \cdot 1 + \left(  1 - \frac{q}{2} \right)\cdot (1-q) = 1-q + \frac{q^2}{2} .
		\]
		Thus, using the convexity of $f$, we have
		\[
		\frac{ f(1-q+q^2/2) - f(1-q)}{f(1)- f(1-q)} \leq \frac{q}{2} \leq 5 \frac{k}{n}.
		\]

		\item For $q > 10k/n$, $1-(1-q)^{n/k} \geq 1- (1-10k/n)^{n/k} \geq 1-e^{-10}$. Thus,
		\begin{align*}
			\frac{ f(1-q+q^2/2) - f(1-q)}{f(1)- f(1-q)} & \leq \frac{(n/k) (1-q+q^2/2)^{n/k-1}q^2/2 }{1-e^{-10}} ,
		\end{align*}
		where we used the mean value theorem on $f$ and the monotonicity of $f'(x)=(n/k)x^{n/k-1}$ to bound the numerator. Now, if $q\geq 1/2$, then
		\[
		\frac{ f(1-q+q^2/2) - f(1-q)}{f(1)- f(1-q)} \leq \frac{n}{k}\left( \frac{5}{8} \right)^{n/k-1} \leq 5 \frac{k}{n}
		\]
		for $n$ large enough.
		
		If $q \leq 1/2$, we proceed as follow. Let $h(q)= (1-q+q^2/2)^{n/k-1}q^2/2$. Then, $h(0)=0$, $h(1)=(1/2)^{n/k-1}$, and, using basic calculus, $h$ has two critical points, %at the points
		\[
		q_1 = \frac{1+(k/n)- \sqrt{1-6k/n + (k/n)^2}}{2}, \quad q_2 =\frac{1+(k/n)+ \sqrt{1-6k/n + (k/n)^2}}{2}.
		\]
		It can be shown that $h$ is increasing in $[0,q_1]$ and decreasing in $[q_1,q_2]$. Moreover, $q_2\geq 1/2$, and we know the behavior of $\frac{ f(1-q+q^2/2) - f(1-q)}{f(1)- f(1-q)} $ in that case. Thus, we only need to bound the value of $h(q_1)$ to conclude. It can be shown that $q_1\leq 6k/n$; hence $ h(q_1) \leq 18 \frac{k}{n}$. Thus,
		\begin{align*}
			\frac{ f(1-q+q^2/2) - f(1-q)}{f(1)- f(1-q)} & \leq 20 \frac{k}{n},
		\end{align*}
		which finishes the proof.
	\end{itemize}

\end{proof}

\subsection{Missing Proofs of Subsection~\ref{subsec:approximation_guarantee_infinite_model}}\label{app:approx_guarantee_infinite_model}

In this subsection, we present the proofs of Lemma~\ref{lem:y_t_bounded_by_x_t} and Lemma~\ref{lem:x_t_bounded_by_w}. We begin with the proof of Lemma~\ref{lem:y_t_bounded_by_x_t}. We divide its proof in several lemmas and propositions.

\begin{lemma}\label{lem:bounds_over_integral_y_t}
	For any $t\leq k-1$, we have
	\[
	k(y_t-y_{t+1}) \leq \int_{y_{t+1}}^{y_t} \frac{-\log y}{1-y^{1/k}} \, \mathrm{d}y \leq \frac{k^2}{k-1} (  y_t^{(k-1)/k} - y_{t+1}^{(k-1)/k}  ).
	\]
\end{lemma}

\begin{proof}
	The proof follows by simply making the change of variable $x=y^{1/k}$ in the integral and using the fact that $1\leq -\log v/(1-v)\leq 1/v$ for any $v\in (0,1)$.
\end{proof}

As a corollary, we have the following lower bound over $y_t$.

\begin{proposition}\label{prop:lower_bound_y_k_ell}
	For $k\geq 4$ and for any $\ell$ we have $y_{k-\ell} \geq \ell/(32 k)$.
\end{proposition}
\begin{proof}
	For any $w\in [0,1]$, we have $w(1-\log w)\geq 0$. Then,
	\begin{align*}
		\beta^*-1 & \leq \beta^* -1 +y_t(1-\log y_t) \\
		& = \int_{y_{t+1}}^{y_t} -\frac{\log y}{1-y^{1/k}} \, \mathrm{d}y\\
		& \leq \frac{k^2}{k-1}\left(  y_{t}^{(k-1)/k} - y_{t+1}^{(k-1)/k}   \right) .
	\end{align*}
	Using  $y_k=0$, this shows that for $t=0,\ldots,k-1$,
	\[
	y_{t-\ell}^{(k-1)/k} \geq \ell \frac{(\beta^* - 1)(k-1)}{k^2}.
	\]
	The bound now follows because $\beta^* \geq 1.25$.
\end{proof}

From the corollary, we obtain the lower bound $y_{k-1} \geq 1/(32 k)$.

Define the two following sequences. Let $x_0=1$, $z_0=1$ and
\begin{align*}
	x_{t+1} &= \max\left\{0, x_t - \frac{1}{k} \left( \beta^* - 1 + x_{t}(1-\log x_{t}) \right)\right\} , \\
	z_{t+1} &= \max \left\{0, z_t - \frac{1}{(32k)^{1/k} k}\left( \beta^* - 1 + z_t (1-\log z_t)\right)\right\}.
\end{align*}
If we ignore the max operator, the sequences represent the Euler method applied to (ODE) with step-sizes $1/k$ and $1/((32k)^{1/k}k)$, respectively.
The following result states the order between elements $x_t$, $y_t$ and $z_t$.

\begin{proposition}\label{prop:y_t_bounded_by_x_t_and_z_t}
	For any $t=0,\ldots,k$, we have $x_{t}\leq y_t\leq z_t$.
\end{proposition}

\begin{proof}
	For the first bound, we use the lower bound found in Lemma~\ref{lem:bounds_over_integral_y_t} to deduce that for any $t=1,\ldots,k-1$,
	\[
	y_{t+1} \geq y_t - \frac{1}{k} \left( \beta^* - 1 + y_{t}(1-\log y_{t}) \right).
	\]
	From here, it is evident that $x_t \leq y_t$ for all $t=0,\ldots,k$.
	
	For the second bound, we need to use the fact that $y_{k-1}\geq 1/(32k)$. Let $t+1 \leq k-1$; then
	\begin{align*}
		\beta^* - 1 + y_t (1-\log y_t) & = \int_{y_{t+1}}^{y_t} - \frac{\log y}{1-y^{1/k}} \, \mathrm{d}y\\
	& \leq \frac{k^2}{k-1} ( y_{t}^{(k-1)/k}  - y_{t+1}^{(k-1)/k}  ) \\
	& = k (\xi)^{-1/k} ( y_t - y_{t+1}) \tag{Mean value theorem} \\
	& \leq k (32k)^{1/k} (y_t - y_{t+1}) . \tag{$y_{k-1}\geq 1/(32k)$}
	\end{align*}
	Reordering terms, we obtain
	\[
	y_{t+1} \leq y_t - \frac{1}{(32k)^{1/k} k} (\beta^* - 1 + y_t(1-\log y_t)) ,
	\]
	from which the bound $y_t\leq z_t$ follows for $t=0,\ldots,k-1$. Since $z_{k}\geq 0 = y_{k}$, the bound holds for any $t=0,\ldots,k$.
\end{proof}

The next result states that the sequence $x_t$ and $z_t$ are at most $\mathcal{O}(\log k/k)$ far from each other.

\begin{proposition}\label{prop:z_t_bounded_by_x_t_error}
	$z_t - x_t \leq 4\log(32 k)/k$, for any $t=0,\ldots,k$.
\end{proposition}

\begin{proof}
	Let $f(x) = x(1-\log x)$. Let $e_t=z_t - x_t \geq 0$. We show that $e_t \leq 4 t \log(32 k)/k^2$. The result is true if $z_{t+1}=0$. So, for $z_{t+1}> 0$ we have
	\begin{align*}
		e_{t+1} & = z_{t+1} - x_{t+1} \\
		& \leq z_t - \frac{1}{(32 k)^{1/k} k} (\beta^* - 1 + f(z_t)) - x_t + \frac{1}{k} (\beta^* - 1 + f(x_t)) \\
		& = e_t + \frac{\beta^*-1}{k} \left(  1 - \frac{1}{(32k)^{1/k}}  \right) + \frac{1}{k} (f(x_t) - f(z_t)) + f(z_t) \frac{1}{k} \left(  1 - \frac{1}{(32k)^{1/k}}  \right) .
	\end{align*}
	Note that
	\[
	1 - \frac{1}{(32k)^{1/k}} \leq (32k)^{1/k} - 1 = e^{\log(32 k)/k} - 1 \leq 2\frac{\log (32 k)}{k} ,
	\]
	for $k\geq 6$, using $e^{x}\leq 1+ 2x$ for $x\leq 1$. Since $\beta^* \leq 2$ and $f(x) \leq 1$ for any $x\in [0,1]$ and $f(x_t) \leq f(z_t)$ since $z_t \geq x_t$, we have
	\[
	e_{t+1} \leq e_t + 4 \frac{\log (32 k)}{k^2}.
	\]
	Solving this recursion gives us the desired inequality.
\end{proof}

\begin{proof}[Proof of Lemma~\ref{lem:y_t_bounded_by_x_t}]
	Proposition~\ref{prop:y_t_bounded_by_x_t_and_z_t} immediately gives us $x_t\leq y_t $ for all $t=0,\ldots,k$. Proposition~\ref{prop:y_t_bounded_by_x_t_and_z_t} in combination with Proposition~\ref{prop:z_t_bounded_by_x_t_error} gives us the bound $y_t\leq x_t + 4 \log(32k)/k$ for all $t=0,\ldots,k$. This finishes the proof.
\end{proof}

%-------------------------------------------------------------------
%-------------------------------------------------------------------

\begin{proof}[Proof of Lemma~\ref{lem:x_t_bounded_by_w}]
	For the proof of this lemma, we need several properties of $w$, the solution of (ODE), that we list below. Property 1 and 2 are easy to deduce, while the proof of property 3 appears in~\cite{correa2021posted}.
	\begin{enumerate}
		\item $w'' = w'\log w \geq 0$, and so $w$ is convex.
		
		\item $w''' \geq 0$. In particular, $w'$ is a convex function.
		
		\item $w + w(\log w - 1)/k + \log (w) (w (\log w - 1) - (\beta^*-1))/(2k^2) \geq  w^{k/(k-1)}(k-1)/k$, for any real $k\geq 2$. This is a byproduct of Proposition D.1 in~\cite{correa2021posted}.
	\end{enumerate}
	
	We show that $x_t\leq w(t/k)$ by induction in $t$.	For $t=0$, we have $w(0) = 1 = x_0$. For $t=1$, using the convexity of $w$, we have $w(1/k) \geq w(0) + w'(0)/k = 1- \beta^*/k$. On the other side,
	\[
	x_1 = x_0 - \frac{1}{k} \left( \beta^* - 1 + x_0 (1- \log x_0)  \right) = 1 - \frac{\beta^*}{k} \leq w\left( \frac{1}{k} \right).
	\]
	Assume the result is true for $t$ and let us show it for $t+1$. If $w(t/k) = 0$, then $x_t=0$ by the inductive hypothesis; hence, $x_{t+1} = 0 = w((t+1)/k)$. Let's assume that $w(t/k)>0$.
	
	\begin{claim}
		Let $\rho\in [t/k,(t+1)/k]$; if $w(\rho) >0$, then
		\[
		w(\rho) \geq x_t - \frac{1}{k} \left( \beta^* - 1 + x_t(1-\log x_t) \right).
		\]
	\end{claim}
	\begin{proof}
		The result is immediate if $x_t=0$. Thus, let us assume that $x_t >0$. Using the Taylor expansion around $t/k$, we have
		\begin{align*}
			w(\rho) &= w\left( t/k \right) + (\rho - t/k)w'\left( t/k\right) + \frac{(\rho - t/k)^2}{2} w''\left( t/k  \right) + \frac{(\rho-t/k)^3}{6} w'''(\xi) \\
			& \geq w\left( t/k \right) + (\rho - t/k)w'\left( t/k\right) + \frac{(\rho - t/k)^2}{2} w''\left( t/k  \right),
		\end{align*}
		since $w'''\geq 0$ by the second property listed above. Now, the function
		$$\rho \mapsto (\rho - t/k)w'\left( t/k\right) + \frac{(\rho - t/k)^2}{2} w''\left( t/k  \right)$$
		is decreasing because its derivative is upper bounded by $w'(\rho)\leq 0$ due to the convexity of $w'$ (see the second property of $w$ above). Thus,		
		\begin{align*}
			w(\rho) &\geq w\left( \frac{t}{k} \right) + \frac{1}{k}w'\left( \frac{t}{k}\right) + \frac{1}{2k^2} w''\left( \frac{t}{k}  \right)\\
			& = w\left( \frac{t}{k} \right) +  \left( 1 + \frac{1}{2k} \log w\left( \frac{t}{k}\right) \right)\left(\frac{w(t/k)\left( \log w(t/k) -1 \right) - (\beta^* - 1)}{k}\right)  \\
			& \geq \frac{k-1}{k} w \left( \frac{t}{k} \right)^{k/(k-1)} - \frac{\beta^*-1}{k} \\
			& \geq \frac{k-1}{k} x_t^{k/(k-1)} - \frac{\beta^*-1}{k}.
		\end{align*}
		In the second inequality we used the third property of $w$ listed above, and in the last inequality we used the inductive hypothesis. Now,
		\[
		x_t^{k/(k-1)} = x_t \cdot x_t^{1/(k-1)} = x_{t} e^{\log (x_t)/(k-1)}\geq x_t \left( 1 + \frac{1}{k-1} \log x_t  \right) = \frac{k}{k-1} x_t + \frac{x_t}{k-1}\left(  \log x_t -1 \right)
		\]
		by using the standard inequality $1+x \leq e^x$ for any $x$. The claim now follows by plugging this inequality in the inequality above.
	\end{proof}
	
	Thus, for $\rho \in [t/k,(t+1)/k]$ with $w(\rho)>0$, we have
	\[
	x_{t+1} = \max \left\{  0 , x_t - \frac{1}{k}( \beta_\infty^* - 1 + x_t (1-\log _x))  \right\} \leq w(\rho).
	\]
	If $\rho^*\in [t/k,(t+1)/k]$, since $w(t/k)>0$, then $\rho^* > t/k$, where $\rho^*=\inf\{ \rho\in [0,1] : w(\rho)=0 \}$. By the continuity of $w$ we obtain
	\[
	w\left( \frac{t+1}{k} \right) = 0 = w(\rho^*) = \lim_{\rho\nearrow \rho^*} w(\rho) \geq x_{t+1}.
	\]
	If $\rho^*\notin [t/k, (t+1)/k]$, then we obtain immediately that $w((t+1)/k)\geq x_{t+1}$. This finishes the proof of Lemma~\ref{lem:x_t_bounded_by_w}.
\end{proof}

\subsection{Proof of $v_{\infty,k} \leq \bar{\gamma}(1- c \log k / k)$}\label{app:upperbound_v}

To achieve an upper bound for $v_{\infty,k}$, we follow a similar approach to the one used in the proof of a lower bound of $v_{\infty,k}$. The idea is to study $\rho^*$, the value $\rho\in [0,1]$ such that $w(\rho)=0$, where $w$ is the solution of (ODE). For the lower bound of $v_{\infty,k}^*$ we showed that $\rho^* \geq 1 - C \log k / k$. We defined
\[
I(\beta) = \int_{0}^1 \frac{\mathrm{d}w}{\beta - 1 + w(1-\log w)},
\]
which is decreasing in $\beta$ and $I(\overline{\beta})=1$. Using these two facts, we deduced that $\beta^* \leq \overline{\beta}(1+ C' \log k / k)$. We plan to use a similar approach for the upper bound of $v_{\infty,k}$. We aim to show that $\rho^* \leq 1 - c\log k/k$ for some $c$, which will imply that $\beta^* \geq \overline{\beta}(1-c'\log k/k)$.

Recall the definition of the sequence $z_t$ presented in the previous subsection:
\begin{align*}
z_0 & = 0 \\
z_{t+1} &= \max \left\{0, z_t - \frac{1}{(32k)^{1/k} k}\left( \beta^* - 1 + z_t (1-\log z_t)\right)\right\} .
\end{align*}
By repeating the same analysis as in the previous subsection (see the proof of Lemma~\ref{lem:x_t_bounded_by_w}), we obtain the following result.
\begin{lemma}
	$w(t/k(32k)^{1/k}) \geq z_{t}$ for any $t=1,\ldots,k$.
\end{lemma}

The following result controls the distance between $w(t/k)$ and $w(t/(k(32)^{1/k}))$. The proof follows by the convexity of $w$ and we skip it for brevity.

\begin{proposition}
	For any $t=1,\ldots,k$, $w(t/k) \geq w(t/(k(32k)^{1/k})) - 2\log(32k)/k$.
\end{proposition}

The next proposition controls the error of $\rho^*$ from above.

\begin{proposition}
	We have $\rho^* \leq 1 - 32 \log(32k)/k$.
\end{proposition}

\begin{proof}
	Using the two previous results and Proposition~\ref{prop:y_t_bounded_by_x_t_and_z_t} we obtain
\[
w(t/k) \geq y_t - 2 \log(32k)/k.
\]
Let $t'\in [k]$ be the minimum value such that $w(t'/k)=0$. Thus, $\rho^* \leq t'/k$. Using Proposition~\ref{prop:lower_bound_y_k_ell}, we have
\[
0 = w(t'/k) \geq \frac{k-t'}{32k} - 2 \frac{\log(32 k)}{k} \implies \frac{t'}{k} \geq 1 - 64 \frac{\log (32k)}{k}.
\]
Therefore, $\rho^* \leq 1-32\log(32 k)/k$.
\end{proof}

Let
\[
I(\beta) = \int_{0}^1 \frac{\mathrm{d}w}{\beta - 1 + w(1-\log w)},
\]
so $I(\overline{\beta})=1$, with $\overline{\beta}=1/\overline{\gamma}$, and $I(\beta^*) =\rho^*\leq 1 - 32\log(32 k)/k$.

\begin{lemma}
	We have $\beta^* \geq \overline{\beta} (1+ 16 \log (32k)/k)$.
\end{lemma}

\begin{proof}
	Let $c=8\log (32k)/k$. We assume by contradiction that $\beta^* < \overline{\beta}(1+c)$. Then,
	\[
	(1+c)\overline{\beta} - 1 + w(1-\log w) \leq (1+4c)(\overline{\beta} - 1 +  w   (1- \log w)),
	\]
	for any $w\in [0,1]$. Thus, $I((1+c)\overline{\beta}) \geq I(\overline{\beta})/(1+4c)\geq 1- 4c$. Since the function $I(\cdot)$ is strictly decreasing, we must have $I(\beta^*) > 1-4c = 1-32\log(32k)/k$, which is a contradiction.
\end{proof}

From this lemma, we obtain that
\[
v_{\infty,k}^* = \frac{1}{\beta^*} \leq  \frac{\overline{\gamma}}{(1+ 8 \log (32k)/k)} \leq \overline{\gamma}\left( 1 - 4 \frac{\log (32k)}{k} \right),
\]
which finishes the proof of the upper bound in Theorem~\ref{thm:bound_infnite_model}.

\section{Missing Proofs of Section~\ref{sec:small_thresholds}}\label{app:small_thresholds}

\subsection{Formulas for $a,b$ and $c$ for $k=2$ Thresholds}\label{app:formulas_a_b_c}

\modif{We solve the $3$ by $3$ system using Wolfram Mathematica. The values $a$, $b$ and $c$ as a function of $u=u_2$ and $\theta$ are
{\scriptsize\begin{align*}
	a & = -\frac{\theta \left(u-e^u+1\right)^2 e^{\theta u+\theta-1} \left(u \left(e^{(\theta-1) (u+1)}-2\right)+\theta (u+1)-1\right)}{\splitdfrac{\left(\theta \left(u^2+4 u+3\right)-u^2-2 u-3\right) \exp \left(\frac{-2 \theta (u+1)+e^u (\theta (u+2)+u-1)-u+1}{e^u-1}\right)}{\splitdfrac{+\left(-\theta \left(2 u^2+5 u+3\right)+2 u^2+4 u+3\right) \exp \left(\frac{\theta \left(e^u (u+2)-2 (u+1)\right)}{e^u-1}-1\right)}{\splitdfrac{+(\theta-1) (u+1)^2 \exp \left(-\frac{(2 \theta-1) (u+1)+e^u (-\theta (u+2)+u+1)}{e^u-1}\right)}{\splitdfrac{-(\theta u+\theta-1) \exp \left(\frac{-2 \theta (u+1)+e^u (\theta (u+2)+2 u-1)-2 u+1}{e^u-1}\right)}{\splitdfrac{+\theta \left(2 u^2-3\right) e^{2 (\theta u+\theta-1)}+e^{2 u} \left(\theta \left(u^2-2\right)+2 u\right)+e^{\theta u+\theta+u-1} \left(\theta \left(-4 u^2-4 u+2\right)+u^2+2 u+3\right)}{\splitdfrac{-\theta \left(u^3+u^2-3 u-3\right) e^{2 \theta (u+1)-u-2}+e^{\theta u+\theta-1} \left(\theta \left(u^3+3 u^2+u-1\right)-2 u^2-4 u-3\right)}{\splitdfrac{+(u+1)^2 e^{(\theta-1) (u+1)}-\theta (u+1)^2 e^{2 (\theta-1) (u+1)}-\theta (u-1) e^{2 \theta (u+1)+u-2}+u e^{-\frac{\theta u}{e^u-1}+\theta+u}}{\splitdfrac{-2 u e^{-\frac{\theta u}{e^u-1}+\theta+2 u}+u e^{-\frac{\theta u}{e^u-1}+\theta+3 u}-e^{3 u} (\theta (u-1)+u)+e^u (\theta u+\theta-u)}{+e^{\theta u+\theta+2 u-1} (\theta (3 u-1)-1)}}}}}}}}}
\end{align*}
\begin{align*}
	b& = -\frac{e^u \left(\splitdfrac{e^{2 (\theta-1) (u+1)} \theta (u+1)^2-2 e^{-u+2 \theta (u+1)-2} \theta (u+1)}{\splitdfrac{+e^{2 (u \theta+\theta-1)} \theta+\theta+e^{\frac{\theta \left(-u+e^u-1\right)}{-1+e^u}} u-2 e^{\frac{u \theta}{1-e^u}+\theta+u} u+e^{\frac{u \theta}{1-e^u}+\theta+2 u} u}{\splitdfrac{+\theta u-u-e^{2 u} (\theta (u-1)+u)+2 e^{\frac{-2 \theta (u+1)+e^u (\theta (u+2)-1)+1}{-1+e^u}} (u \theta+\theta-1)}{ \splitdfrac{-e^{\frac{-u-2 \theta (u+1)+e^u (u+\theta (u+2)-1)+1}{-1+e^u}} (u \theta+\theta-1)+e^{u \theta+\theta+u-1} \left((u+1) \theta^2+(u-2) \theta-1\right)}{\splitdfrac{+e^{-\frac{(2 \theta-1) (u+1)+e^u (u-\theta (u+2)+1)}{-1+e^u}} (1-\theta (u+1))+e^{(\theta-1) (u+1)} \left(\theta^2 (u+1)^2-2 \theta (u+1)^2-1\right)}{+e^u \left(2 u+\theta \left(u^2-2\right)\right)+e^{u \theta+\theta-1} \left(-\left(\left(u^2+3 u+2\right) \theta^2\right)+(3 u+4) \theta+2\right)}} }  } } \right)}{\splitdfrac{-e^{(\theta-1) (u+1)} (u+1)^2-e^{-\frac{(2 \theta-1) (u+1)+e^u (u-\theta (u+2)+1)}{-1+e^u}} (\theta-1) (u+1)^2}{ \splitdfrac{+e^{2 (\theta-1) (u+1)} \theta (u+1)^2+e^{u+2 \theta (u+1)-2} \theta (u-1)-e^{\frac{u \theta}{1-e^u}+\theta+u} u+2 e^{\frac{u \theta}{1-e^u}+\theta+2 u} u}{\splitdfrac{-e^{\frac{u \theta}{1-e^u}+\theta+3 u} u+e^{3 u} (\theta (u-1)+u)+e^{\frac{-2 u-2 \theta (u+1)+e^u (2 u+\theta (u+2)-1)+1}{-1+e^u}} (u \theta+\theta-1)-e^u (u \theta+\theta-u)}{\splitdfrac{+e^{u \theta+\theta+2 u-1} (-3 u \theta+\theta+1)+e^{2 (u \theta+\theta-1)} \theta \left(3-2 u^2\right)+e^{-u+2 \theta (u+1)-2} \theta \left(u^3+u^2-3 u-3\right)}{ \splitdfrac{-e^{2 u} \left(2 u+\theta \left(u^2-2\right)\right)+e^{\frac{-u-2 \theta (u+1)+e^u (u+\theta (u+2)-1)+1}{-1+e^u}} \left(u^2+2 u-\theta \left(u^2+4 u+3\right)+3\right)}{ \splitdfrac{+e^{\frac{-2 \theta (u+1)+e^u (\theta (u+2)-1)+1}{-1+e^u}} \left(-2 u^2-4 u+\theta \left(2 u^2+5 u+3\right)-3\right)+e^{u \theta+\theta+u-1} \left(-u^2-2 u+\theta \left(4 u^2+4 u-2\right)-3\right)}{+e^{u \theta+\theta-1} \left(2 u^2+4 u-\theta \left(u^3+3 u^2+u-1\right)+3\right) }}}}}}}
\end{align*}
} We skip the formula for $c$ as it can be deduced from the equality $1=a+b(1-e^{-u})+ce^{-u}$.}

\subsection{Prophet Inequality for $k=2$ Thresholds using $(CLP)$}\label{app:threshold_k_2}

In this subsection, we numerically compute the optimal value of $(CLP)_{n,k}$ for $k=2$ thresholds. The objective of this section is to show that the gap $\gamma_{n,2}^*$ computed in Section~\eqref{sec:small_thresholds} and $v_{n,2}^*$ is small. Specifically, we show that $v_{n,2}^*\approx 0.704$ while $\gamma_{n,2}^*\approx 0.708$. We work with the infinite model, and parameterize our solution via $\theta\in (0,1)$, such that $\tau_1=\theta n$ in the finite model. Optimizing over $\theta$ then gives the optimal threshold.

Repeating the same process as in Section~\eqref{sec:small_thresholds} for $k\geq 3$, we can compute $v_{\infty,2}^*(\alpha)$ by solving

{ $(CLP)_{\infty,2}(\theta)$\hfil \begin{tabular}{cp{0.75\linewidth}}
		$\displaystyle\sup_{\substack{\omega\geq 0}}$ &  $\displaystyle\quad \inf_{\overline{y}\in [0,1]}  \left\{\int_0^{\overline{y}} \left(\frac{ 1- y^\theta}{-\log y}\right) \omega_{1,y}  + \left(\frac{ 1- y^{1-\theta}}{-\log y}\right) \omega_{2,y}  \, \mathrm{d}y\right\}$ \\
		s.t.  & {\vspace{-1cm}  \begin{align*}
				\int_0^1 \frac{\omega_{1,y}}{y}\, \mathrm{d}y& \leq 1 &   \\
				\int_0^1 \frac{\omega_{2,y}}{y}\, \mathrm{d}y& \leq \int_0^1 y^{\theta} \frac{\omega_{1,y}}{y}\, \mathrm{d}y ,  &
			\end{align*}\vspace{-0.3cm}}
\end{tabular}}

with $\theta \geq 1/2$; recall that our guarantees work for $\tau_1 \geq \tau_2$, which in the infinite model translates into $\theta \geq 1/2$. For any $0< v\leq v_{\infty,2}^*(\theta)$, we can construct feasible solutions of the form
\[
\bar{\omega}_{1,y} = - v \frac{y \log y}{1-y^{\theta}} \mathbf{1}_{(y_1,y_0)}(y), \quad \bar{\omega}_{1,y} = - v \frac{y \log y}{1-y^{1-\theta}} \mathbf{1}_{(y_2,y_1)}(y) ,
\]
where $y_0=1$, $y_1\in [0,1]$ and $y_2=0$. The existence of $y_1\in [0,1]$ such that $(\bar{\omega}_{t,\cdot})_{t=1,2}$ is a feasible solution to $(CLP)_{\infty,2}(\theta)$ can be deduced from Proposition~\ref{prop:epsilon_decreasing_gamma} and the limit model; we skip details for brevity. Note that the value $y_1$ is implicitly defined in terms of $v$. To find $y_1$, we proceed as follows.
%\atnote{I suggest stating this LP with the $v$ variable already included; it will help the reader (me) better understand the following argument.}
Let $H_\varphi(x)=\int_0^x -(1-y^\varphi)^{-1} \log y \, \mathrm{d}y$ with $\varphi\in \{ \theta,1-\theta \}$. Note that
$$H_{\theta}(y_0)-H(y_1) = \frac{1}{v} \int_0^1 \frac{\bar{\omega}_{t,1}}{y} \, \mathrm{d}y \quad \text{and}\quad H_{1-\theta}(y_1)-H_{1-\theta}(y_2) = \frac{1}{v}\int_0^1 \frac{\bar{\omega}_{t,2}}{y} \, \mathrm{d}y . $$

Then, from the constraints of $(CLP)_{\infty,2}(\theta)$, we deduce that $(\bar{\omega}_{t,\cdot})_{t=1,2}$ is feasible if there is $y_1\in [0,1]$ such that
\begin{align}
	\frac{1}{v} - H_\theta(y_0) + H_\theta(y_1) & \geq 0 \label{ineq:2_threshold_ineq_1} \\
	H_{\theta}(y_0) - H_\theta(y_1) - H_{1-\theta}(y_1) +H_{1-\theta}(y_2) +y_0(\log y_0 -1) - y_1(\log y_1-1) & \geq 0 \label{ineq:2_threshold_ineq_2}
\end{align}
holds, where $y_0=1$ and $y_2=0$. We compute $y_1$ by setting
\[
y_1 = \min \left\{ y\in [0,1] :  1/v - H_\theta(1) + H_\theta(y) \geq 0 \right\}.
\]
%\atnote{Made some edits to this definition; please confirm.}
It is easy to verify that $y_1$ tightens Inequality~\eqref{ineq:2_threshold_ineq_1}. Moreover, Inequality~\eqref{ineq:2_threshold_ineq_2} will hold if and only if $v\leq v_{\infty,2}^*(\theta)$. When $v=v_{\infty,2}^*(\theta)$, both Inequalities~\eqref{ineq:2_threshold_ineq_1}-\eqref{ineq:2_threshold_ineq_2} must be tight; then we have the following result. %\ps{change wording, please check}

%\atnote{I'm honestly having a hard time following this entire argument. It would help if you explain what $H$ is supposed to be and how it relates to $\omega$. The optimal $\omega$ is defined at the start of Section 5; it would also help to refer to that equation to tie everything up here.}

\begin{proposition}
	Let $\theta \geq 1/2$. If $v=v_{\infty,2}^*(\theta)$ and $y_1$ is defined as above, then $y_1$ defines the optimal threshold for the optimal solution of $(CLP)_{\infty,2}(\theta)$.
\end{proposition}

Given a fixed $\theta \geq 1/2$ and $v\leq v_{\infty,2}^*(\theta)$, the value $y_1$ can be approximately computed by a simple bisection procedure. We discretize $[0,1]$ in intervals of length $1/\ell$ and in $\mathcal{O}(\log \ell)$ bisection operations we obtain an approximate optimal solution $y_1$. We can also run a bisection over $v$ and obtain a solution $v$ with  $ \lvert v - v_{\infty,2}^*(\theta) \rvert \leq \delta $ in $\mathcal{O}(\log 1/\delta)$. We finally discretize $[1/2,1]$ in multiples of $1/r$ and run the aforementioned subroutine for each $\theta $. Overall, computing an approximation ratio that is at worst $\delta$ units off from the best approximation that our methodology can attain takes $\mathcal{O}(r \log(1/\delta)\log(\ell))$ iterations. Each iteration requires computing $H_\varphi$, which we perform numerically.

We set $r=1000$, $\delta=10^{-8}$ and $\ell=10^{12}$. With these values, we obtain an experimental maximum at $\theta^*=0.610$ with $v_{\infty,2}^*(\theta^*) \approx 0.7048$ and $y_1\approx 0.2620$; see Figure~\ref{fig:graphk2} for other values $v_{\infty,2}^*(\theta)$ as a function of $\theta\in [1/2,1]$. Numerically, we observe that as $\theta$ goes to $1$, $v_{\infty,2}^*$ goes to $v_{\infty,1}^* = 6/\pi^2$ (see Corollary\ref{cor:value_k_1}). Also, the difference between $v_{\infty,2}^*(1/2)\approx 0.701$ and $v_{\infty,2}^*(\theta^*)\approx 0.704$ is less than $0.4\%$. Nonetheless, empirically we observe that optimizing over the length of the intervals does give a significant improvement over the approximation.

% TODO: \usepackage{graphicx} required
\begin{figure}
	\centering
	\includegraphics[width=0.8\linewidth]{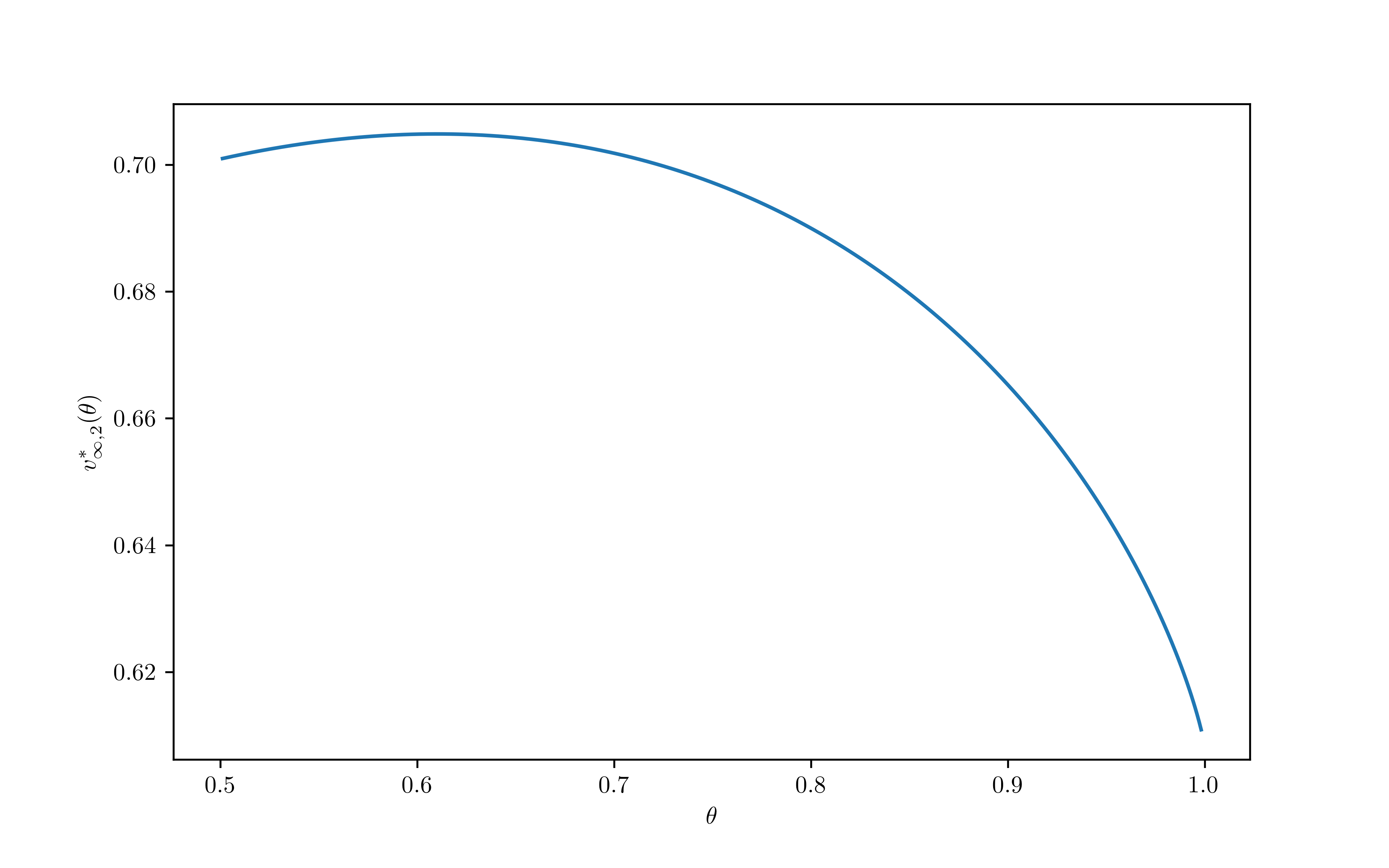}
	\caption{The line represents the numerical values of $v_{\infty,2}^*(\theta)$ computed with the procedure described in this subsection. As $\theta$ goes to $1$, we observe numerically that $v_{\infty,2}^*$ goes to $6/\pi^2$.}
	\label{fig:graphk2}
\end{figure}